\documentclass[10pt]{article}
\usepackage{amsmath}
\usepackage{amsthm}
\usepackage{amssymb}
\usepackage{fullpage}
\usepackage{cite}
\usepackage{hyperref}
\usepackage{ifthen}
\usepackage{subcaption}
\usepackage{graphicx}
\usepackage[linesnumbered, nofillcomment,noend]{algorithm2e}
\usepackage{stackrel}
\usepackage{enumitem}
\usepackage{comment}
\usepackage{MnSymbol}
\usepackage{color}
\usepackage[table]{xcolor}

\newtheorem{theorem}{Theorem}
\newtheorem*{theorem*}{Theorem}
\newtheorem{definition}{Definition}
\newtheorem{lemma}{Lemma}
\newtheorem{lemma*}{Lemma}

\newtheorem{observation}{Observation}

\begin{document}

\title{\textbf{Improved lower and upper bounds on the tile complexity of uniquely self-assembling a thin rectangle non-cooperatively in 3D}}

\author{%
David Furcy\thanks{Computer Science Department, University of Wisconsin Oshkosh, Oshkosh, WI 54901, USA,\protect\url{furcyd@uwosh.edu}.}
\and
Scott M. Summers\thanks{Computer Science Department, University of Wisconsin Oshkosh, Oshkosh, WI 54901, USA,\protect\url{summerss@uwosh.edu}.}
\and
Logan Withers\thanks{Computer Science Department, University of Wisconsin Oshkosh, Oshkosh, WI 54901, USA,\protect\url{withel75@uwosh.edu}.}
}


\date{}
\maketitle

\begin{abstract}
We investigate a fundamental question regarding a benchmark class of shapes in one of the simplest, yet most widely utilized abstract models of algorithmic tile self-assembly.
More specifically, we study the directed tile complexity of a $k \times N$ thin rectangle in Winfree's ubiquitous abstract Tile Assembly Model, assuming that cooperative binding cannot be enforced (temperature-1 self-assembly) and that tiles are allowed to be placed at most one step into the third dimension (just-barely 3D).
While the directed tile complexities of a square and a scaled-up version of any algorithmically specified shape at temperature 1 in just-barely 3D are both asymptotically the same as they are (respectively) at temperature 2 in 2D, the (loose) bounds on the directed tile complexity of a thin rectangle at temperature 2 in 2D are not currently known to hold at temperature 1 in just-barely 3D.
Motivated by this discrepancy, we establish new lower and upper bounds on the directed tile complexity of a thin rectangle at temperature 1 in just-barely 3D.
On our way to proving our lower bound, we develop a new, more powerful type of specialized Window Movie Lemma that lets us derive our lower bound via a counting argument, where we upper bound the number of ``sufficiently similar'' ways to assign glues to a set (rather than a sequence) of fixed locations.
Consequently, our lower bound, $\Omega\left( N^{\frac{1}{k}} \right)$, is an asymptotic improvement over the previous state of the art lower bound and is more aesthetically pleasing since it eliminates the non-constant term $k$ that used to divide $N^{\frac{1}{k}}$.
The proof of our upper bound is based on the construction of a novel, just-barely 3D temperature-1 counter, organized according to ``digit regions'', which affords it roughly fifty percent more digits for the same target rectangle compared to the previous state of the art counter.
This increase in digit density results in an upper bound of $O\left( N^{\frac{1}{\left \lfloor \frac{k}{2} \right \rfloor}}  + \log N \right)$, that is an asymptotic improvement over the previous state of the art upper bound and roughly the square of our lower
bound.
\end{abstract}

\section{Introduction}
\label{sec:intro}

%
%

%
%
%
%
%
A key objective in algorithmic self-assembly is to characterize the extent to which an algorithm can be converted to an efficient self-assembling system comprised of discrete, distributed and disorganized units that, through random encounters with and locally-defined reactions to each other, coalesce into a terminal assembly having a desirable form or function.
In this paper, we study a fundamental question regarding a benchmark class of shapes in one of the simplest yet most popular abstract models of algorithmic self-assembly. 
Ubiquitous throughout the theory of algorithmic self-assembly, Erik Winfree's abstract Tile Assembly Model (aTAM) \cite{Winf98} is a discrete mathematical model of DNA tile self-assembly \cite{Seem82} that augments classical Wang tiling \cite{Wang61} with a mechanism for automatic growth.
In the aTAM, a DNA tile is represented by a unit square (or cube) tile type that may neither rotate, reflect, nor fold. 
Each side of a tile type is decorated with a glue consisting of both a non-negative integer strength and an alpha-numeric label. 
A tile set is a finite set of tile types, from which infinitely many tiles of each type may be instantiated. 
If one tile is positioned at an unoccupied location Manhattan distance 1 away from another tile and their opposing glues are equal, then the two tiles bind with the strength of the opposite glues. 
A special seed tile type is designated and a seed tile, which defines the seed-containing assembly, is placed at some fixed location.
During the process of self-assembly, a sequence of tiles bind to and never detach from the seed-containing assembly, provided that each one, in a non-overlapping fashion, binds to one or more tiles in the seed-containing assembly with total strength at least a certain positive integer value called the temperature.
If the temperature is greater than or equal to 2, then it is possible to enforce cooperative binding, where a tile may be prevented from binding at a certain location until at least two adjacent locations become occupied by tiles. 
Otherwise, only non-cooperative binding is allowed (temperature-1 self-assembly).
A fundamental question regarding a given shape is determining the effect of the value of the temperature on its directed tile complexity, or the size of the smallest tile set that, regardless of the order in which tiles bind to the seed-containing assembly, always self-assembles into a unique terminal assembly of tiles that are placed on and only on points of the shape.
Although temperature-1 self-assembly cannot enforce cooperative binding, there is a striking resemblance of its computational and geometric expressiveness in just-barely 3D, where tiles are allowed to be placed at most one step in the third dimension, to that of temperature-2 self-assembly in 2D, with respect to the directed tile complexity of two benchmark shapes: a square and a scaled-up version of any algorithmically specified shape.
Adleman, Cheng, Goel and Huang \cite{AdlemanCGH01} proved, using optimal base conversion, that the directed tile complexity of an $N \times N$ square at temperature 2 in 2D is $O\left( \frac{\log N}{\log \log N} \right)$, matching a corresponding lower bound for all Kolmogorov-random $N$ and all positive temperature values, set by Rothemund and Winfree \cite{RotWin00}. 
Both of these bounds hold for temperature-1 self-assembly in just-barely 3D.
The lower bound is an easy generalization of the latter and the upper bound was established by Furcy, Micka and Summers \cite{jFurcyMickaSummers} via their discovery of a just-barely 3D, optimal encoding construction at temperature 1.
Just-barely 3D, optimal encoding at temperature 1 was inspired by, achieves the same result as, but is drastically different from the 2D optimal encoding at temperature 2 developed by Soloveichik and Winfree \cite{SolWin07}, who proved that the directed tile complexity of a scaled-up version of any algorithmically specified shape $X$ at temperature 2 is $\Theta\left( \frac{K(X)}{\log K(X)} \right)$, where $K(X)$ is the size of the smallest Turing machine that outputs the list of points in $X$. 
This tight bound for temperature-2 self-assembly in 2D was shown to hold for temperature-1 self-assembly in just-barely 3D by Furcy and Summers \cite{FurcyS18}: they combined just-barely 3D optimal encoding at temperature 1 with a modified version of a just-barely 3D, temperature-1 Turing machine simulation by Cook, Fu and Schweller \cite{CookFuSch11}.
Another benchmark shape is the $k \times N$ rectangle, where $k < \frac{\log N}{\log \log N - \log \log \log N}$, making it  ``thin''.  
A thin rectangle is an interesting testbed because its restricted height creates a limited channel through which tiles may propagate information, for example, the current value of a self-assembling counter. 
In fact, Aggarwal, Cheng, Goldwasser, Kao, Moisset de Espan\'{e}s and Schweller \cite{AGKS05g} used an optimal, base-$\left \lceil N^{\frac{1}{k}}\right \rceil$ counter that uniquely self-assembles within the restricted confines of a thin rectangle to derive an upper bound of $O\left( N^{\frac{1}{k}} + k\right)$ on the directed tile complexity of a $k \times N$ thin rectangle at temperature 2 in 2D. 
They then leveraged the limited bandwidth of a thin rectangle in a counting argument for a corresponding lower bound of $\Omega\left( \frac{N^{\frac{1}{k}}}{k} \right)$. 
The previous theory for a square and an algorithmically specified shape would suggest that these thin rectangle bounds should hold at temperature 1 in just-barely 3D.
Yet, we currently do not know if this is the case.
Thus, the power of temperature-1 self-assembly in just-barely 3D resembles that of temperature-2 self-assembly in 2D, with respect to the directed tile complexities of a square and a scaled-up version of any algorithmically specified shape, but not a thin rectangle.

\begin{table}[t!]

\begin{tabular}{|c||c|c||c|c|}\hline
  & \multicolumn{2}{c||}{2D Temperature 2} & \multicolumn{2}{c|}{Just-barely 3D Temperature 1} \\
  & Lower bound & Upper bound & Lower bound & Upper bound \\ \hline

  $N \times N$ Square
  & \multicolumn{2}{c||}{ $\Theta\left(\frac{\log N}{\log\log N}\right)$} 
  & \multicolumn{2}{c|}{Same as 2D Temperature 2}\\ \hline

  Algorithmically-defined shape $X$
  & \multicolumn{2}{c||}{$\Theta\left(\frac{K(X)}{\log K(X)}\right)$}
  & \multicolumn{2}{c|}{Same as 2D Temperature 2}\\ \hline

    \parbox{2in}{\centering $k \times N$ rectangle}
  & $\Omega\left(\frac{N^{\frac{1}{k}}}{k}\right)$
  & $O\left(N^{\frac{1}{k}}+k\right)$
  & $\Omega\left(\frac{N^{\frac{1}{2k}}}{k}\right)$
    & $O\left(N^{\frac{1}{\left\lfloor\frac{k}{3}\right\rfloor}}+\log N\right)$\\ \hline  
\end{tabular}
\caption{\label{tbl:table-1} State-of-the-art directed tile complexity for the
  self-assembly of benchmark shapes in the aTAM, where $K(X)$ is the
  size of the smallest Turing machine that outputs the list of points
  in $X$.}
\end{table}

\begin{table}[t!]
\begin{tabular}{|c||c|c||c|c|}\hline
  & \multicolumn{2}{c||}{2D Temperature 2} & \multicolumn{2}{c|}{Just-barely 3D Temperature 1} \\
  & Lower bound & Upper bound & Lower bound & Upper bound \\ \hline

  \parbox{2in}{\centering $k \times N$ rectangle}
  & N/A
  & N/A 
  & \cellcolor{gray!40} $\Omega\left(N^{\frac{1}{k}}\right)$
  & \cellcolor{gray!20}$O\left(N^{\frac{1}{\left\lfloor\frac{k}{2}\right\rfloor}}+\log N\right)$\\ \hline  
\end{tabular}
\caption{\label{tbl:table-2} Our improved lower and upper bounds on the directed tile
  complexity of rectangles are highlighted in this table as the
  two main contributions of this paper. Note that, for a thin rectangle, the additive terms in this table and Table~\ref{tbl:table-1} are eliminated.}
\end{table}

Motivated by this theoretical discrepancy, we prove new lower and upper bounds on the directed tile complexity of a thin rectangle at temperature 1 in just-barely 3D. See Tables~\ref{tbl:table-1} and~\ref{tbl:table-2} for a quick summary of our results and how they compare with previous state of the art results. 
Our lower bound is:

\begin{theorem}
The directed tile complexity of a $k \times N$ thin rectangle at temperature 1 in just-barely 3D is $\Omega\left(N^{\frac{1}{k}}\right)$.
\end{theorem}

Theorem~\ref{thm:theorem-1} is an asymptotic improvement over the corresponding previous state of the art lower bound:

\begin{theorem*}
\label{thm:current-lower-bound}
The directed tile complexity of a $k \times N$ thin rectangle at temperature 1 in just-barely 3D is $\Omega\left(\frac{N^{\frac{1}{2k}}}{k}\right)$.
\end{theorem*}

Technically, the previous lower bound is not explicitly proved (or even stated) and therefore cannot be referenced, but it can be derived via a straightforward adaptation of the counting argument given in the proof of the lower bound for a thin rectangle for general temperature values in 2D. This proof would basically use a counting argument that upper bounds the number of ways to assign glues (of tiles) to a sequence of fixed locations abutting a plane. The idea is that, if two assignments are similar, in that they, respectively, assign the same glues at the fixed locations in the same order, but off by translation, then it is possible to self-assemble something other than the target rectangle (giving a contradiction). 
In such a Pigeonhole counting argument, since $N$ is fixed at the beginning of the proof, a larger lower bound on the number of types of glues is required in order to avoid a contradiction arising from two similar assignments.
On our way to proving Theorem~\ref{thm:theorem-1}, we prove Lemma~\ref{lem:lemma-2}, which is essentially a new, more powerful type of Window Movie Lemma \cite{WindowMovieLemma} for temperature-1 self-assembly within a just-barely 3D, rectangular region of space. 
We establish our lower bound via a counting argument, but unlike the previous example, our new technical machinery lets us merely upper bound the number of ``sufficiently similar'' ways to assign glues to a fixed set (rather than a sequence) of locations abutting a plane. Intuitively, two assignments are sufficiently similar if, up to translation, they respectively agree on: the set of locations to which glues are assigned, the local order in which certain consecutive pairs of glues appear, and the glues that are assigned to a certain set (of roughly half) of the locations. Our lower bound is also aesthetically pleasing because only a hidden constant term divides ``$N^{\frac{1}{k}}$'', making it roughly the square root of our upper bound, which is:

\begin{theorem}
\label{thm:two}
The directed tile complexity of a $k \times N$ rectangle at temperature 1 in just-barely 3D is\\$O\left(N^{\frac{1}{\left \lfloor \frac{k}{2} \right \rfloor}} + \log N \right)$.
\end{theorem}

Theorem~\ref{thm:two} is an asymptotic improvement over the corresponding previous state of the art upper bound:
\begin{theorem*}[Furcy, Summers and Wendlandt \cite{FurcySummersWendlandtDNA}]
\label{thm:current-upper-bound}
The directed tile complexity of a $k \times N$ rectangle at temperature 1 in just-barely 3D is $O\left(N^{\frac{1}{\left \lfloor \frac{k}{3} \right \rfloor}} + \log N \right)$.
\end{theorem*}

The previous upper bound is based on the self-assembly of a just-barely 3D counter that uniquely self-assembles at temperature 1, but whose base $M$ depends on the dimensions of the target rectangle. 
Moreover, each digit in the previous counter is represented geometrically and in binary within a just-barely 3D region of space comprised of $\Theta(\log M)$ columns and $3$ rows. 
In any kind of construction like this, the number of rows used to represent each digit affects the base of the counter, which, for a thin rectangle, is directly proportional to and the asymptotically-dominating term in the tile complexity.
For example, in the previous construction, the number of rows per digit is $3$, so the base must be set to $\Theta\left(N^{\frac{1}{\left \lfloor \frac{k}{3} \right \rfloor}}\right)$.
Intuitively, ``squeezing'' more digits into the counter for the same rectangle of height $k$ will result in a decrease in the base and therefore the tile complexity. Our construction for Theorem~\ref{thm:two} is based on the self-assembly of a just-barely 3D counter similar to the previous construction, but the geometric structure of our counter is organized according to digit regions, or just-barely 3D regions of space comprised of $\Theta(\log M)$ columns and $4$ rows in which two digits are represented. So, on average, each digit in our counter is represented within a just-barely 3D region of space comprised of $\Theta(\log M)$ columns, but only $2$ rows, resulting in a roughly fifty percent increase in digit density for the same rectangle of height $k$, compared to the counter for the previous result. This increase in digit density is the main reason why the ``$3$'' from the previous upper bound is replaced by a ``$2$'' in Theorem~\ref{thm:two}.

\section{Formal definition of the abstract Tile Assembly Model}
\label{sec:prelims}

\newcommand{\N}{\mathbb{N}}
\newcommand{\Z}{\mathbb{Z}}

\newcommand{\lab}{{\rm label}}
\newcommand{\colorf}[1]{\color(#1)}
\newcommand{\strength}{{\rm str}}
\newcommand{\strengthf}[1]{\strength(#1)}
\newcommand{\bbval}[1]{\left[\!\left[ #1 \right] \! \right]}
\newcommand{\dom}{{\rm dom} \;}
\newcommand{\asmb}{\mathcal{A}}
\newcommand{\asmbt}[2]{\asmb^{#1}_{#2}}
\newcommand{\asmbtt}{\mathcal{A}^\tau_T}
\newcommand{\ste}[2]{#1 \mapsto #2}
\newcommand{\frontier}[3]{{\partial}^{#1}_{#2}{#3}}
\newcommand{\frontiert}[1]{\partial^{\tau}{#1}}
\newcommand{\frontiertt}[1]{\frontier{\tau}{t}{#1}}
\newcommand{\frontiertx}[2]{{\partial}^{\tau}_{#1}{#2}}
\newcommand{\frontiertau}[1]{{\partial}^{\tau}{#1}}
\newcommand{\arrowstett}[2]{#1 \xrightarrow[\tau,T]{1} #2}
\newcommand{\arrowste}[2]{#1 \stackrel{1}{\To} #2}
\newcommand{\arrowtett}[2]{#1 \xrightarrow[\tau,T]{} #2}
\newcommand{\res}[1]{\textrm{res}(#1)}
\newcommand{\termasm}[1]{\mathcal{A}_{\Box}[\mathcal{#1}]}
\newcommand{\prodasm}[1]{\mathcal{A}[\mathcal{#1}]}
\newcommand{\fgg}[1]{fgg^\#_{#1}}
\newcommand{\ftdepth}[2]{\textrm{ft-depth}_{#1}\left(#2\right)}
\newcommand{\str}[1][*]{\textrm{str}_{#1}}
\newcommand{\col}[1]{\textrm{col}_{#1}}
\newcommand{\tmblank}{\llcorner \negthinspace\lrcorner}
\newcommand{\fullgridgraph}{G^\mathrm{f}}\newcommand{\bindinggraph}{G^\mathrm{b}}

\newcommand{\setr}[2]{\left\{\ #1 \ \left|\ #2 \right. \ \right\}}
\newcommand{\setl}[2]{\left\{\ \left. #1 \ \right|\ #2 \ \right\}}

In this section, we briefly sketch a strictly 3D version of Winfree's abstract Tile Assembly Model. 

All logarithms in this paper are base-$2$. Fix an alphabet $\Sigma$.
$\Sigma^*$ is the set of finite strings over $\Sigma$. Let $\Z$, $\Z^+$, and $\N$ denote the set of integers, positive integers, and nonnegative integers, respectively. 

A \emph{grid graph} is an undirected graph $G=(V,E)$, where $V \subset \Z^3$, such that, for all $\left\{\vec{a},\vec{b}\right\} \in E$, $\vec{a} - \vec{b}$ is a $3$-dimensional unit vector. The \emph{full grid graph} of $V$ is the undirected graph $\fullgridgraph_V=(V,E)$,
such that, for all $\vec{x}, \vec{y}\in V$, $\left\{\vec{x},\vec{y}\right\} \in E \iff \| \vec{x} - \vec{y}\| = 1$, i.e., if and only if $\vec{x}$ and $\vec{y}$ are adjacent in the $3$-dimensional integer Cartesian space.

A $3$-dimensional \emph{tile type} is a tuple $t \in (\Sigma^* \times \N)^{6}$, e.g., a unit cube, with six sides, listed in some standardized order, and each side having a \emph{glue} $g \in \Sigma^* \times \N$ consisting of a finite string \emph{label} and a nonnegative integer \emph{strength}. We assume a finite set of tile types, but an infinite number of copies of each tile type, each copy referred to as a \emph{tile}. A \emph{tile set} is a set of  tile types and is usually denoted as $T$.

A {\em configuration} is a (possibly empty) arrangement of tiles on
the integer lattice $\Z^3$, i.e., a partial function $\alpha:\Z^3
\dashrightarrow T$.  Two adjacent tiles in a configuration \emph{bind},
\emph{interact}, or are \emph{attached}, if the glues on their
abutting sides are equal (in both label and strength) and have
positive strength.  Each configuration $\alpha$ induces a
\emph{binding graph} $\bindinggraph_\alpha$, a grid graph whose
vertices are positions occupied by tiles, according to $\alpha$, with
an edge between two vertices if the tiles at those vertices
bind. 

An \emph{assembly} is a connected, non-empty configuration,
i.e., a partial function $\alpha:\Z^3 \dashrightarrow T$ such that
$\fullgridgraph_{\dom \alpha}$ is connected and $\dom \alpha \neq
\emptyset$. Given $\tau\in\Z^+$, $\alpha$ is \emph{$\tau$-stable} if every cut-set
of~$\bindinggraph_\alpha$ has weight at least $\tau$, where the weight
of an edge is the strength of the glue it represents.\footnote{A
  \emph{cut-set} is a subset of edges in a graph which, when removed
  from the graph, produces two or more disconnected subgraphs. The
  \emph{weight} of a cut-set is the sum of the weights of all of the
  edges in the cut-set.} When $\tau$ is clear from context, we say
$\alpha$ is \emph{stable}.  Given two assemblies $\alpha,\beta$, we
say $\alpha$ is a \emph{subassembly} of $\beta$, and we write $\alpha
\sqsubseteq \beta$, if $\dom\alpha \subseteq \dom\beta$ and, for all
points $\vec{p} \in \dom\alpha$, $\alpha(\vec{p}) = \beta(\vec{p})$.

A $3$-dimensional \emph{tile assembly system} (TAS) is a triple $\mathcal{T} =
(T,\sigma,\tau)$, where $T$ is a tile set, $\sigma:\Z^3
\dashrightarrow T$ is the $\tau$-stable, \emph{seed assembly}, with $| \dom\sigma | = 1$
and $\tau\in\Z^+$ is the \emph{temperature}.

Given two $\tau$-stable assemblies $\alpha,\beta$, we write $\alpha
\to_1^{\mathcal{T}} \beta$ if $\alpha \sqsubseteq \beta$ and
$|\dom\beta \backslash \dom\alpha| = 1$. In this case we say $\alpha$
\emph{$\mathcal{T}$-produces $\beta$ in one step}. If $\alpha
\to_1^{\mathcal{T}} \beta$, $ \dom\beta \backslash
\dom\alpha=\{\vec{p}\}$, and $t=\beta(\vec{p})$, we write $\beta =
\alpha + (\vec{p} \mapsto t)$.  The \emph{$\mathcal{T}$-frontier} of
$\alpha$ is the set $\partial^\mathcal{T} \alpha = \bigcup_{\alpha
  \to_1^\mathcal{T} \beta} (\dom\beta \backslash \dom\alpha$), i.e.,
the set of empty locations at which a tile could stably attach to
$\alpha$. The \emph{$t$-frontier} of $\alpha$, denoted
$\partial^\mathcal{T}_t \alpha$, is the subset of
$\partial^\mathcal{T} \alpha$ defined as
$\setr{\vec{p}\in\partial^\mathcal{T} \alpha}{\alpha \to_1^\mathcal{T}
  \beta \text{ and } \beta(\vec{p})=t}.$

Let $\mathcal{A}^T$ denote the set of all assemblies of tiles from $T$, and let $\mathcal{A}^T_{< \infty}$ denote the set of finite assemblies of tiles from $T$.
A sequence of $k\in\Z^+ \cup \{\infty\}$ assemblies $\vec{\alpha} = \left (\alpha_0,\alpha_1,\ldots \right)$ over $\mathcal{A}^T$ is a \emph{$\mathcal{T}$-assembly sequence} if, for all $1 \leq i < k$, $\alpha_{i-1} \to_1^\mathcal{T} \alpha_{i}$. The {\em result} of an assembly sequence $\vec{\alpha}$, denoted as $\textmd{res}(\vec{\alpha})$, is the unique limiting assembly (for a finite sequence, this is the final assembly in the sequence). We write $\alpha \to^\mathcal{T} \beta$, and we say $\alpha$ \emph{$\mathcal{T}$-produces} $\beta$ (in 0 or more steps), if there is a $\mathcal{T}$-assembly sequence $\alpha_0,\alpha_1,\ldots$ of length $k = |\dom\beta \backslash \dom\alpha| + 1$ such that
(1) $\alpha = \alpha_0$,
(2) $\dom\beta = \bigcup_{0 \leq i < k} \dom{\alpha_i}$, and
(3) for all $0 \leq i < k$, $\alpha_{i} \sqsubseteq \beta$.
If $k$ is finite then it is routine to verify that $\beta = \alpha_{k-1}$.

We say $\alpha$ is \emph{$\mathcal{T}$-producible} if $\sigma \to^\mathcal{T} \alpha$, and we write $\prodasm{\mathcal{T}}$ to denote the set of $\mathcal{T}$-producible assemblies. An assembly $\alpha$ is \emph{$\mathcal{T}$-terminal} if $\alpha$ is $\tau$-stable and $\partial^\mathcal{T} \alpha=\emptyset$. 
We write $\termasm{\mathcal{T}} \subseteq \prodasm{\mathcal{T}}$ to denote the set of $\mathcal{T}$-producible, $\mathcal{T}$-terminal assemblies. If $|\termasm{\mathcal{T}}| = 1$ then  $\mathcal{T}$ is said to be {\em directed}.

In general, a $3$-dimensional shape is a set $X \subseteq \mathbb{Z}^3$. We say that a TAS $\mathcal{T}$ \emph{self-assembles} $X$ if, for all $\alpha \in
\termasm{\mathcal{T}}$, $\dom\alpha = X$, i.e., if every terminal
assembly produced by $\mathcal{T}$ places a tile on every point in $X$
and does not place any tiles on points in $\Z^3 \backslash\, X$. We say that a TAS $\mathcal{T}$ \emph{uniquely 
  self-assembles} $X$ if $\termasm{\mathcal{T}} = \{ \alpha \}$ and $\dom\alpha = X$.

In the spirit of \cite{RotWin00}, we define the \emph{tile complexity} of a shape $X$ at temperature $\tau$, denoted by $K^\tau_{SA}(X)$, as the minimum number of distinct tile types of any TAS in which it self-assembles, i.e., $K^\tau_{SA}(X) = \min \left\{ n   \; \left| \; \mathcal{T} = \left(T,\sigma,\tau\right), \left|T\right|=n \textmd{ and } X \textmd{ self-assembles in } \mathcal{T} \right.\right\}$. The \emph{directed tile complexity} of a shape $X$ at temperature $\tau$, denoted by $K^\tau_{USA}(X)$, is the minimum number of distinct tile types of any TAS in which it uniquely self-assembles, i.e., $K^\tau_{USA}(X) = \min \left\{ n   \; \left| \; \mathcal{T} = \left(T,\sigma,\tau\right), \left|T\right|=n \textmd{ and } X \textmd{ uniquely self-assembles in } \mathcal{T} \right.\right\}$.

\section{The lower bound}
\label{sec:impossibility}

In this section, we prove our main impossibility result, namely Theorem~\ref{thm:theorem-1}. For $k,N \in \mathbb{Z}^+$, we say that $R^3_{k,N} \subseteq \mathbb{Z}^3$ is a 3D $k \times N$ \emph{rectangle} if $\{0,1, \ldots,N-1\} \times \{0,1,\ldots,k-1\} \times \{0\} \subseteq R^3_{k,N} \subseteq \{0,1\ldots,N-1\} \times \{0,1\ldots,k
-1\} \times \{0,1\}$. Then, Theorem~\ref{thm:theorem-1} says that $K^1_{USA}\left(R^3_{k,N}\right) = \Omega\left( N^{\frac{1}{k}} \right)$. Our proof of Theorem~\ref{thm:theorem-1} relies on the following unquestionable observation regarding temperature-1 self-assembly.

\begin{observation}
\label{obs:simple}
If $\mathcal{T} = (T,\sigma,1)$ is a directed TAS, in which some shape $X$ self-assembles and $\alpha$ is the unique element of
$\mathcal{A}_{\Box}[\mathcal{T}]$, then, for each simple path $s$ in $G^{\textmd{b}}_{\alpha}$ from the location of $\sigma$ to some location in $X$, there is a unique assembly sequence $\vec{\alpha}$ that follows $s$ by placing tiles on and only on locations in $s$.
\end{observation}

Our proof technique for Theorem~\ref{thm:theorem-1} is based on a Pigeonhole counting argument, justified by novel technical machinery. Basically, we upper bound
the number of ways that glues can be placed between two adjacent just-barely 3D columns of $R^3_{k,N}$ by an assembly sequence that follows a simple
path. Thus, we get a lower bound on the tile
  complexity of a sufficiently large thin rectangle. We
first give some notation that will be used throughout the remainder of
this section. For the sake of consistency, the next paragraph contains definitions that were taken directly from
\cite{WindowMovieLemma}.

A \emph{window} $w$ is a set of edges forming a cut-set of the full grid
graph of $\mathbb{Z}^3$. Given a window $w$ and an assembly $\alpha$, a window that {\em
  intersects} $\alpha$ is a partitioning of $\alpha$ into two
configurations (i.e., after being split into two parts, each part may
or may not be disconnected). In this case we say that the window $w$
cuts the assembly $\alpha$ into two non-overlapping configurations $\alpha_L$
and~$\alpha_R$, satisfying, for all
$\vec{x} \in \dom{\alpha_L}$, $\alpha(\vec{x}) = \alpha_L(\vec{x})$, for
all $\vec{x} \in \dom{\alpha_R}$, $\alpha(\vec{x}) = \alpha_R(\vec{x})$, and
$\alpha(\vec{x})$ is undefined at any point $\vec{x} \in \Z^3
\backslash \left( \dom{\alpha_L} \cup \dom{\alpha_R} \right)$.

Given a window $w$, its translation by a vector $\vec{\Delta}$,
written $ w + \vec{\Delta}$ is simply the translation of each one of
$w$'s elements (edges) by~$\vec{\Delta}$. All windows in this paper
are assumed to be induced by some translation of the
$yz$-plane. Each window is thus uniquely identified by its $x$
  coordinate or, more precisely, its distance from the $x$ axis.

For a window $w$ and an assembly sequence $\vec{\alpha}$, we define a
\emph{glue window movie}~$M$ to be the order of placement, position
and glue type for each glue that appears along the window $w$ in
$\vec{\alpha}$. Given an assembly sequence $\vec{\alpha}$ and a window
$w$, the associated glue window movie is the maximal sequence
$M_{\vec{\alpha},w} = \left(\vec{v}_{1}, g_{1}\right),
\left(\vec{v}_{2}, g_{2}\right), \ldots$ of pairs of grid graph
vertices $\vec{v}_i$ and glues $g_i$, given by the order of appearance
of the glues along window $w$ in the assembly sequence
$\vec{\alpha}$. We write $M_{\vec{\alpha}, w} + \vec{\Delta}$ to
denote the translation by $\vec{\Delta}$ of $M_{\vec{\alpha},w}$,
yielding $\left(\vec{v}_{1}+\vec{\Delta}, g_{1}\right),
\left(\vec{v}_{2}+\vec{\Delta}, g_{2}\right), \ldots$. If $s$ is a
simple path and $\vec{\alpha}$ \emph{follows} $s$ by placing tiles on
all and only the locations that belong to $s$, then the notation
$M_{\vec{\alpha}, w} \upharpoonright s$ denotes the \emph{restricted}
glue window submovie (\emph{restricted to} $s$), which consists of
only those steps of $M_{\vec{\alpha},w}$ that place glues that eventually form
positive-strength bonds at locations belonging to the simple path $s$.

Let $\vec{v}$ denote the location of the starting point of $s$ (i.e.,
the location of $\sigma$). Let $\vec{v}_i$ and $\vec{v}_{i+1}$ denote
two consecutive locations in $M_{\vec{\alpha}, w} \upharpoonright s$
that are located across $w$ from each other. We say that these two
locations define a {\it crossing} of $w$, where a crossing has exactly
one direction: we say that this crossing is {\it away from $\vec{v}$}
(or {\it away from $\sigma$}) if the $x$ coordinates of $\vec{v}$ and
$\vec{v}_i$ are equal or the $x$ coordinate of $\vec{v_{i}}$ is
between the $x$ coordinates of $\vec{v}$ and $\vec{v}_{i+1}$. In
contrast, we say that this crossing is {\it toward $\vec{v}$} (or {\it
  toward $\sigma$}) if the $x$ coordinates of $\vec{v}$ and
$\vec{v}_{i+1}$ are equal or the $x$ coordinate of $\vec{v}_{i+1}$ is
between the $x$ coordinates of $\vec{v}$ and $\vec{v_{i}}$.

See Figure~\ref{fig:glue_movie_window_examples} for
2D examples of $M_{\vec{\alpha},w}$ and
$M_{\vec{\alpha}, w} \upharpoonright s$. In this figure, $\sigma$ is
located west of $w$ and the locations $\vec{v}_1$
and $\vec{v}_2$ form an away crossing, whereas the locations
$\vec{v}_3$ and $\vec{v}_4$ form a crossing toward $\sigma$.

\begin{figure}
    \centering

    \begin{subfigure}[t]{0.2\textwidth}
        \centering
        \includegraphics[width=.75in]{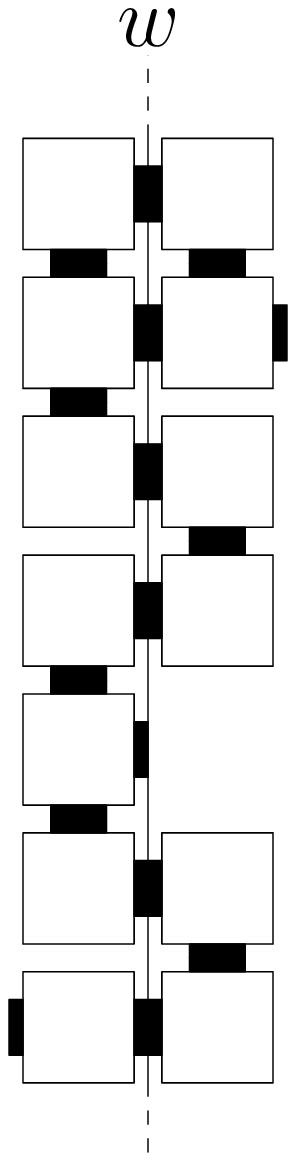}
        \caption{\label{fig:glue_movie_window_example_assembly} A subassembly of $\alpha$ and a window $w$ induced by a translation of the $y$-axis.}
    \end{subfigure}%
    ~
    \begin{subfigure}[t]{0.2\textwidth}
        \centering
        \includegraphics[width=.75in]{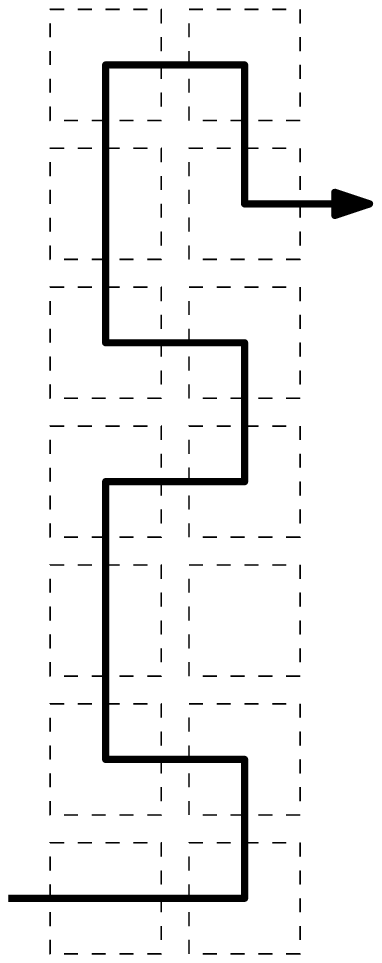}
        \caption{\label{fig:glue_movie_window_example_just_path} A portion of the simple path $s$ through $G^b_{\alpha}$.}
    \end{subfigure}%
   ~
    \begin{subfigure}[t]{0.2\textwidth}
        \centering
        \includegraphics[width=.75in]{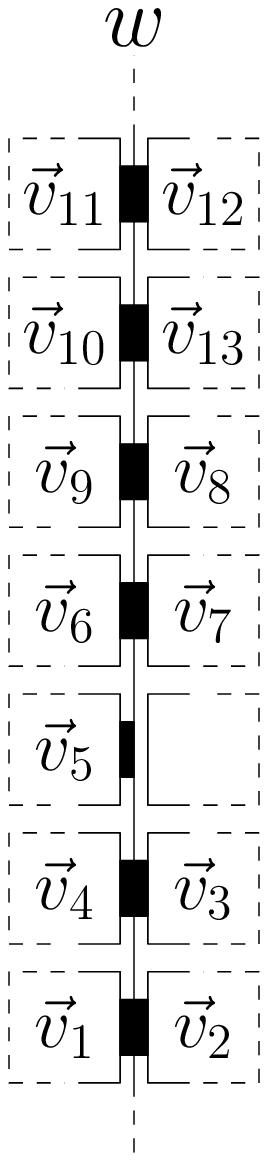}

        \caption{\label{fig:glue_movie_window_example_full} The glue window movie $M_{\vec{\alpha},w}$. }
    \end{subfigure}%
    ~
    \begin{subfigure}[t]{0.2\textwidth}
        \centering
        \includegraphics[width=.75in]{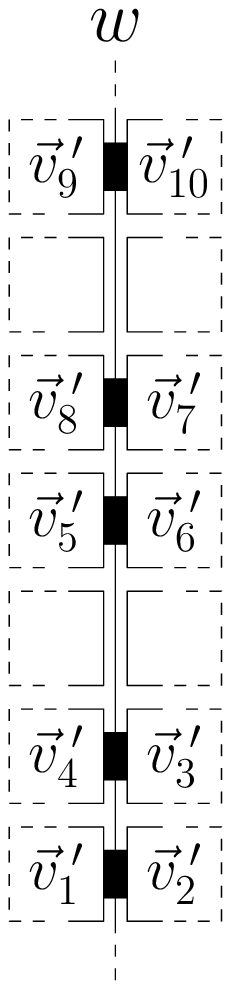}
        \caption{\label{fig:glue_movie_window_example_restricted} The restricted glue window submovie $M_{\vec{\alpha},w} \upharpoonright s$ }
    \end{subfigure}
    \caption{\label{fig:glue_movie_window_examples} An assembly, a simple path, and two types of glue window movies in 2D. }
\end{figure}

We say that two restricted glue window submovies are ``sufficiently
similar'' if they have the same (odd) number of crossings, the same
set of crossing locations (up to horizontal
translation), the same crossing directions at corresponding crossing
locations, and the same glues in corresponding ``away crossing''
locations.

%
%
%
%
\begin{definition}
\label{def:sufficiently-similar}
Assume:
$\mathcal{T}=(T,\sigma, 1)$ is a 3D TAS,
$\alpha \in \mathcal{A}[\mathcal{T}]$,
$s$ is a simple path in $G^{\textmd{b}}_{\alpha}$ starting from the location of $\sigma$,
$\vec{\alpha}$ is a sequence of $\mathcal{T}$-producible assemblies that follows $s$, 
$w$ and $w'$ are windows, such that, $\vec{\Delta} \ne \vec{0}$ is a vector satisfying $w' = w + \vec{\Delta}$,
$e$ and $e'$ are two odd numbers, and
$M = M_{\vec{\alpha}, w} \upharpoonright s =
          \left(\vec{v}_1,g_1\right), \ldots, \left( \vec{v}_{2e},
          g_{2e}\right)$ and $M' = M_{\vec{\alpha}, w'}
            \upharpoonright s = \left(\vec{v}'_1,g'_1\right), \ldots,
            \left( \vec{v}'_{2e'}, g'_{2e'}\right)$ are both non-empty restricted glue window
          submovies.
%
%
%
%
%
%
%
%
We say that $M$ and $M'$ are \emph{sufficiently similar} if the following conditions are satisfied:
\begin{enumerate}
\item \label{lemma-assumption-1} same number of crossings: $e = e'$,

\item \label{lemma-assumption-2} same set of crossing locations (up to translation): $\left\{ \left. \vec{v}_i + \vec{\Delta} \ \right| \ 1 \leq i \leq 2e\right\} = \left\{ \left. \vec{v}'_j \ \right| \ 1 \leq j \leq 2e  \right\}$,
	
	\item \label{lemma-assumption-3} same crossing directions at corresponding crossing locations:\\$\left\{ \left. \vec{v}_{4i-2} + \vec{\Delta} \ \right| \ 1 \leq i \leq \frac{e+1}{2}  \right\} = \left\{ \left. \vec{v}'_{4j-2} \ \right| \ 1 \leq j \leq \frac{e+1}{2} \ \right\}$, and
	
	\item \label{lemma-assumption-4} same glues in corresponding ``away crossing'' locations:\\ for all $1 \leq i,j \leq \frac{e+1}{2}$, if $\vec{v}'_{4j-2} = \vec{v}_{4i-2} + \vec{\Delta}$, then $g'_{4j-2} = g_{4i-3}$.
\end{enumerate}

Note that, since $e$ and $e'$ are both odd,
  the $x$ coordinates of $w$ and $w'$ must both be between the $x$
  coordinates of the end points of $s$.
\end{definition}

See Figure~\ref{fig:half_windows_equivalent_hypothesis} for an example of two restricted glue window submovies that are sufficiently similar. 

Our first technical result says that we must examine only a ``small'' number of distinct restricted glue window submovies in order to find two different ones that are sufficiently similar.
%
%
%
%
\begin{lemma}
\label{lem:lemma-1}
%
%
%
%
Assume:
	 $\mathcal{T}=(T,\sigma, 1)$ is a 3D TAS, 
	 $G$ is the set of all glues in $T$, 
	 $k,N \in \mathbb{Z}^+$,
         $s$ is a simple path starting from
          the location of $\sigma$ such that $s \subseteq R^3_{k,N}$,
	 $\vec{\alpha}$ is a sequence of $\mathcal{T}$-producible
          assemblies that follows $s$,
	 $m \in \mathbb{Z}^+$, 
	 for all $1 \leq l \leq m$, $w_l$ is a
          window,
	 for all $1 \leq l < l' \leq m$, $\vec{\Delta}_{l,l'} \ne \vec{0}$ satisfies $w_{l'} = w_l + \vec{\Delta}_{l,l'}$, and
	 for all $1 \leq l \leq m$, there is an odd $1 \leq e_l <
          2k$ such that $M_{\vec{\alpha}, w_l} \upharpoonright s$
          is a non-empty restricted glue window submovie of length $2e_l$.
%
%
%
%
%
%
%
%
%
%
If $m > |G|^k \cdot k \cdot 16^{k}$, then there exist $1 \leq l < l'
\leq m$ such that $e_l = e_{l'} = e$ and $M_{\vec{\alpha}, w_l}
  \upharpoonright s = \left(\vec{v}_1,g_1\right), \ldots, \left(
  \vec{v}_{2e}, g_{2e}\right)$ and $M_{\vec{\alpha}, w_{l'}}
  \upharpoonright s = \left(\vec{v}'_1,g'_1\right), \ldots, \left(
  \vec{v}'_{2e}, g'_{2e}\right)$ are sufficiently similar non-empty restricted glue
window submovies.

\end{lemma}

The proof idea for Lemma~\ref{lem:lemma-1} goes like this.
We first count the number of ways to choose the set $\left\{  \vec{v}_1, \ldots, \vec{v}_{2e} \right\}$. 
Then, we count the number of ways to choose the set
          $\left\{ \vec{v}_{4i-2} \ \left| \ 1 \leq i \leq
          \frac{e+1}{2} \right. \right\}$. 
Finally, we count the number of ways to choose the sequence $\left( g_{\vec{x}_i} \left| \  i =1,\ldots,\frac{e+1}{2}  \right.  \right)$.
After summing over all odd $e$, we get the indicated lower bound on $m$ that notably neither contains a ``factorial'' term nor a coefficient on the ``$k$'' in the exponent of ``$|G|$''. 
See Section~\ref{sec:appendix-lower-bound} for the full proof of Lemma~\ref{lem:lemma-1}.
Our second technical result is the cornerstone of our lower bound machinery. It basically says that if, for some directed
TAS $\mathcal{T}$, two distinct restricted glue window submovies are
sufficiently similar, then $R^3_{k,N}$ does not self-assemble in $\mathcal{T}$. 
%
%
%
%
\begin{lemma}
\label{lem:lemma-2}
Assume:
	 $\mathcal{T}$ is a directed, 3D TAS, 
	 $k, N \in \mathbb{Z}^+$, 
	 $s \subseteq R^3_{k,N}$ is a simple path from the location of the seed of $\mathcal{T}$ to some location in the furthest extreme column of $R^3_{k,N}$,
	 $\vec{\alpha}$ is a $\mathcal{T}$-assembly sequence that follows $s$, 
	 $w$ and $w'$ are windows, such that, $\vec{\Delta} \ne \vec{0}$ is a vector satisfying $w' = w + \vec{\Delta}$, and 
	 $e$ is an odd number satisfying $1 \leq e < 2k$. 
%
%
%
%
%
%
If $M = M_{\vec{\alpha},w} \upharpoonright s =  \left(\vec{v}_1,g_1\right),\ldots,\left(\vec{v}_{2e},g_{2e}\right)$ and $M' = M_{\vec{\alpha},w'} \upharpoonright s = \left(\vec{v}'_1,g'_1\right),\ldots,\left(\vec{v}'_{2e},g'_{2e}\right)$ are sufficiently similar non-empty restricted glue window submovies, then $R^3_{k,N}$ does not self-assemble in $\mathcal{T}$. 

\end{lemma}

See Figure~\ref{fig:half_windows_equivalent} for a 2D example of Lemma~\ref{lem:lemma-2}.

\begin{figure}[htp]
    \centering
	\begin{subfigure}[t]{\textwidth}
        \centering
        \includegraphics[width=.65\textwidth]{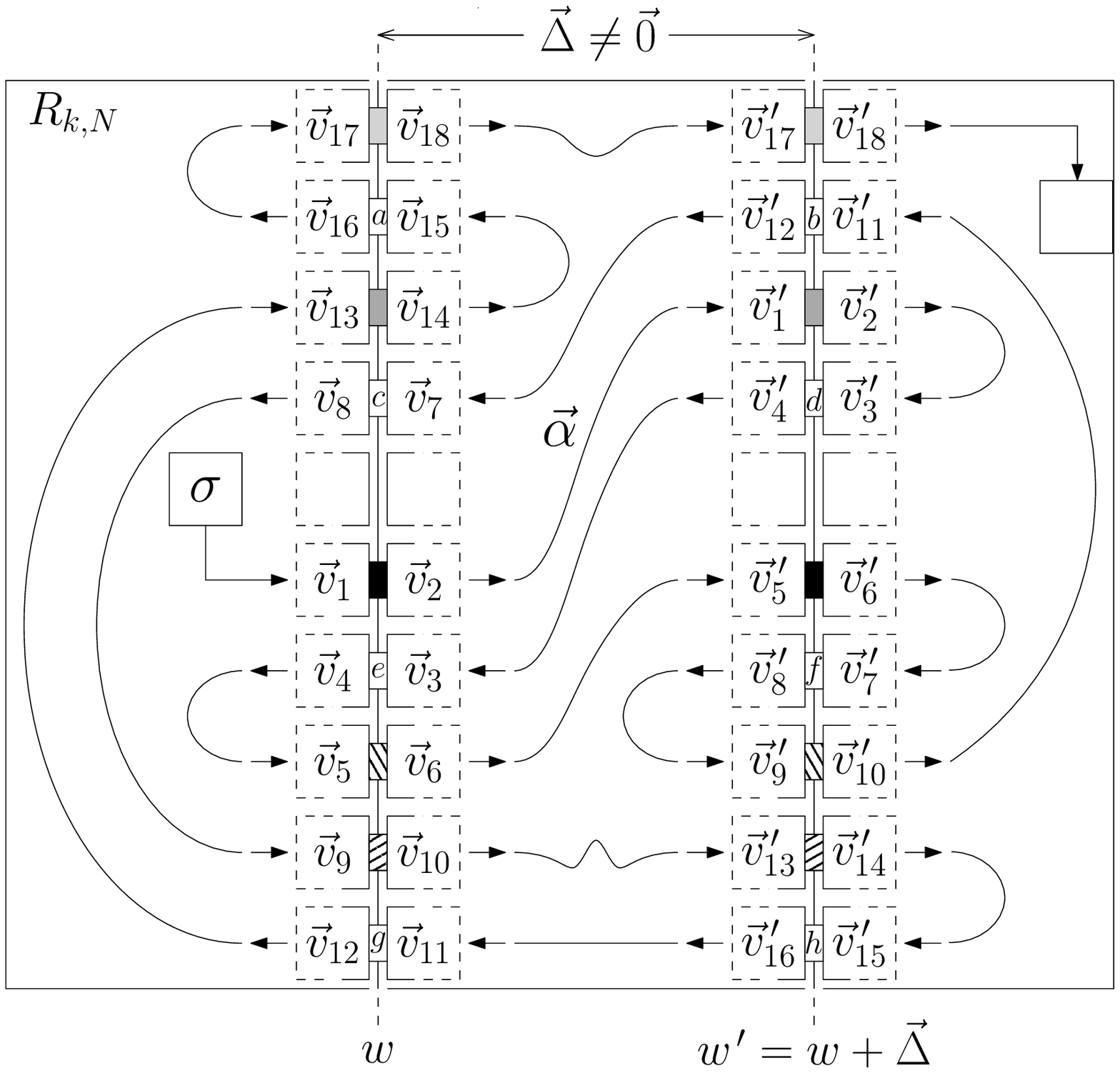}

        \caption{\label{fig:half_windows_equivalent_hypothesis} The
          hypothesis. Note that $\vec{\alpha}$ follows a simple path
          $s$ from the location of $\sigma$ to a location in the
          furthest extreme column. The restricted glue window movies
          are sufficiently similar because their glues are at the same
          locations (up to translation), oriented in the same
          direction (away or toward $\sigma$), and each pair of glues that are
          placed by $\vec{\alpha}$ at an ``away crossing'' of
            one of the windows is equal to its translated
          counterpart in the other window, e.g., the two topmost glues
          that touch $w$ and $w'$ are both light gray. The same
          constraint holds for all glue pairs shown with a solid shade
          of gray or a striped pattern. On the other hand, the glues
          adjacent to $w'$ that are placed by $\vec{\alpha}$ at
            a ``toward crossing'', for example $g'_{11}$
          and $g'_{12}$, are decorated with a letter in order to
          represent the fact that we do not assume that these glues
          are equal to their translated counterparts that touch $w$
          (i.e., $g_{15}$ and $g_{16}$).  }
    \end{subfigure}

    \begin{subfigure}[t]{\textwidth}
        \centering
        \includegraphics[width=.65\textwidth]{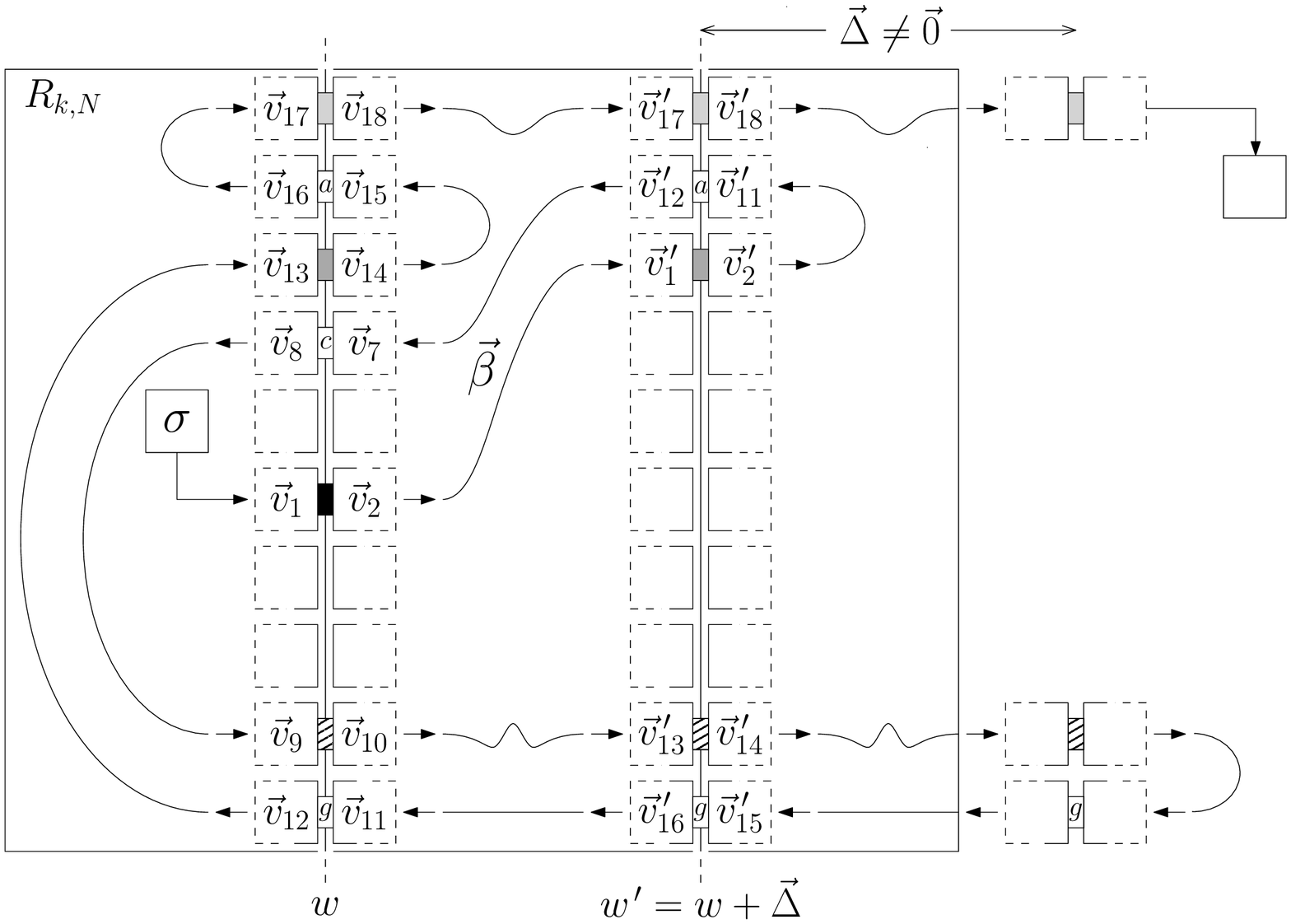}
        \caption{\label{fig:half_windows_equivalent_conclusion}The
          conclusion.  Given the fact that $\mathcal{T}$ is directed
          and the way $\vec{\beta}$ is defined, every pair of glues
          that touch $w$ must be equal to the corresponding pair of
          glues that touch $w'$ (if any). Thus, e.g., the glue pairs
          labelled $b$ and $h$ in
          part~(\subref{fig:half_windows_equivalent_hypothesis}) must
          really be equal to the glue pairs $a$ and $g$, respectively.
          After $\vec{\beta}$ places a tile at location
          $\vec{v}^{\,\prime}_{17}$, it will mimic how $\vec{\alpha}$
          got from $\vec{v}_{18}$ to the tile in the extreme column of
          $R_{k,N}$, as depicted in
          part~(\subref{fig:half_windows_equivalent_hypothesis}). Since
          $\vec{\Delta} \ne \vec{0}$, this always results in at least
          one tile placement outside of $R_{k,N}$. In this example,
          $\beta$ also happens to exit $R_{k,N}$ earlier in its
          assembly sequence, i.e., in the sub-path from
          $\vec{v}^{\,\prime}_{14}$ to $\vec{v}^{\,\prime}_{15}$. }
    \end{subfigure}

 \caption{\label{fig:half_windows_equivalent} A 2D example of the
   hypothesis and conclusion of Lemma~\ref{lem:lemma-2} for $k = 10$
   and $e = 9$. Since the example is 2D, we use $R_{k,N} = \{0,
   \ldots, N - 1\} \times \{0, \ldots, k - 1\} $, rather than
   $R^3_{k,N}$.  }

    \end{figure}

We now give some notation that will be useful for our discussion of the proof of Lemma~\ref{lem:lemma-2}. The definitions and notation in the following paragraph are inspired by notation that first appeared in \cite{WindowMovieLemma}.

For a $\mathcal{T}$-assembly sequence $\vec{\alpha} = (\alpha_i \mid 0 \leq i <
l)$, we write $\left| \vec{\alpha} \right| = l$. 
We write $\vec{\alpha}[i]$ to denote $\vec{x} \mapsto t$, where $\vec{x}$
and~$t$ are such that $\alpha_{i+1} = \alpha_i + \left(\vec{x} \mapsto
t\right)$.
We write $\vec{\alpha}[i] + \vec{\Delta}$, for
some vector $\vec{\Delta}$, to denote
$\left(\vec{x}+\vec{\Delta}\right) \mapsto t$. 
If $\alpha_{i+1} = \alpha_i +
\left(\vec{x} \mapsto t\right)$, then we write
$Pos\left(\vec{\alpha}[i]\right) = \vec{x}$ and
$Tile\left(\vec{\alpha}[i]\right) = t$.
Assuming $\left|
\vec{\alpha} \right| > 0$, the notation $\vec{\alpha} = \vec{\alpha} + \left(\vec{x} \mapsto
t\right)$ denotes a \emph{tile placement step}, namely the sequence of configurations $\left(\alpha_i \mid 0 \leq i < l + 1 \right)$, where $\alpha_{l}$ is the configuration satisfying, $\alpha_l\left( \vec{x} \right) = t$ and for all $\vec{y} \ne \vec{x}$, $\alpha_l\left(\vec{y}\right) = \alpha_{l-1}\left( \vec{y} \right)$.
Note that the ``$+$'' in a tile placement step is different from the
``$+$'' used in the notation ``$\beta = \alpha + \left(\vec{p} \mapsto
t\right)$''.
However, since the former operates on an assembly sequence, it should be clear from the context which operator is being invoked. 
The definition of a tile placement step does not require that the sequence of configurations be a $\mathcal{T}$-assembly sequence.
After all, the tile placement step $\vec{\alpha} = \vec{\alpha} + \left(\vec{x} \mapsto
t\right)$ could be attempting to place a tile at a location that is not even adjacent to (a location in the domain of) $\alpha_{l-1}$.
Or, it could be attempting to place a tile at a location that is in the domain of $\alpha_{l-1}$, which means a tile has already been placed at $\vec{x}$. 
So we say that  such a tile placement step is \emph{correct} if $\left(\alpha_i \mid 0 \leq i < l + 1 \right)$ is a $\mathcal{T}$-assembly sequence.
If $\left| \vec{\alpha} \right| = 0$, then $\vec{\alpha} =
\vec{\alpha} + \left(\vec{x} \mapsto t \right)$ results in the
$\mathcal{T}$-assembly sequence $(\alpha_0)$, where $\alpha_0$ is the
assembly such that $\alpha_0\left(\vec{x}\right) = t$ and
$\alpha_0\left( \vec{y} \right)$ is undefined at all other locations
$\vec{y} \neq \vec{x}$.

In Figure~\ref{alg:beta-prime}, we define an algorithm that uses
$\vec{\alpha}$ to construct a new assembly sequence $\vec{\beta}$ such
that the tile placement steps by $\vec{\beta}$ on the far side of
$w'$ from the seed mimic a (possibly strict) subset of the tile
placements by $\vec{\alpha}$ on the far side of $w$ from the seed.
When $\vec{\beta}$ is on the near side of $w'$ to the seed, it mimics $\vec{\alpha}$, although $\vec{\beta}$ does not necessarily mimic every tile placement by $\vec{\alpha}$ on the near side of $w'$ to the seed. 
When $\vec{\beta}$ crosses $w'$, going away from the seed, by placing
tiles at $\vec{v}'_{4j-3}$ and $\vec{v}'_{4j-2}$ in this order, then
the tile it places at $\vec{v}'_{4j-2}$ is of the same type as the
tile that $\vec{\alpha}$ places at $\vec{v}_{4i-2} = \vec{v}'_{4j-2} -
\vec{\Delta}$.
After $\vec{\beta}$ crosses $w'$ by placing a tile at $\vec{v}'_{4j-2}$, $\vec{\beta}$ places tiles that $\vec{\alpha}$ places along $s$ from $\vec{v}_{4i-2}$ to $\vec{v}_{4i-1}$, but the tiles $\vec{\beta}$ places are translated to the far side of $w'$ from the seed. 
When $\vec{\beta}$ is about to cross $w'$, going toward the seed, by placing a tile at $\vec{v}'_{4j-1}$, then, since $\mathcal{T}$ is directed, the type of tile that it places at this location is equal to the type of tile that $\vec{\alpha}$ places at $\vec{v}'_{4j-1}$. 
This means that $\vec{\beta}$ may continue to follow $s$ but starting from $\vec{v}'_{4j}$. 
Eventually, $\vec{\beta}$ will finish crossing $w'$ going away from the seed for the last time by placing a tile at $\vec{v}_{2e}+\vec{\Delta}$. 
Then, $\vec{\beta}$ places tiles that $\vec{\alpha}$ places along $s$, starting from $\vec{v}_{2e}$, but the tiles that $\vec{\beta}$ places are translated to the far side of $w'$ from the seed. 
Since $\vec{\Delta} \ne \vec{0}$, $\vec{\beta}$ will ultimately place a tile that is not in $R^3_{k,N}$. 
%


\begin{figure}[h!]
\input{lemma-2-algorithm}
\caption{\label{alg:beta-prime} The algorithm for $\vec{\beta}$. Here, the variable ``$k$'' has no relation to the ``$k$'' used in $R^3_{k,N}$.}
\end{figure}

We illustrate the behavior of this algorithm in Figure~\ref{fig:algorithm_trace}, where we
  apply it to the assembly sequence $\vec{\alpha}$ shown in
  Figure~\ref{fig:half_windows_equivalent}(a).

We must show that all of the tile
placement steps executed by the algorithm for $\vec{\beta}$ are
correct. 
In addition, we must also prove that the tile placement steps executed by the algorithm for $\vec{\beta}$ place tiles along a simple path.
Let $\vec{\alpha} = \vec{\alpha} + \left( \vec{x}' \mapsto t' \right)$
and $\vec{\alpha} = \vec{\alpha} + \left( \vec{x} \mapsto t \right)$
be consecutive tile placement steps executed by the algorithm for
$\vec{\beta}$ and assume that the former is either the first tile placement step executed or it is correct.
To show that the latter is correct, we will show that:
\begin{enumerate}[label=\alph*]
	\item \label{adjacent-valid-binding} the tile configuration 
          that consists of $t$ placed at $\vec{x}$ and $t'$ placed at
          $\vec{x}'$ is a $1$-stable assembly (not necessarily $\mathcal{T}$-producible) whose domain
          consists of two locations, and
	\item \label{adjacent-valid-empty-location} the location	
          $\vec{x}$ is unoccupied before $\vec{\alpha} = \vec{\alpha} + \left( \vec{x} \mapsto t \right)$ is executed.
\end{enumerate}
The previous two conditions constitute a slightly stronger notion of ``correctness'' for a tile placement step, which we will call \emph{adjacently correct}. 
After all, the two previous conditions imply that $\vec{\alpha} = \vec{\alpha} + \left( \vec{x} \mapsto t \right)$ is correct, but if $\vec{\alpha} = \vec{\alpha} + \left( \vec{x} \mapsto t \right)$ is correct, then condition~\ref{adjacent-valid-empty-location} must hold but condition~\ref{adjacent-valid-binding} need not because $\vec{x}$ does not have to be adjacent to $\vec{x}'$.  
It suffices to prove that every tile placement step executed by the algorithm for $\vec{\beta}$ is adjacently correct.
As a result, $\vec{\beta}$, like $\vec{\alpha}$, will place tiles along a simple path.


\begin{figure}
    \centering

    \begin{subfigure}[t]{0.475\textwidth}
        \centering
        \includegraphics[width=\textwidth]{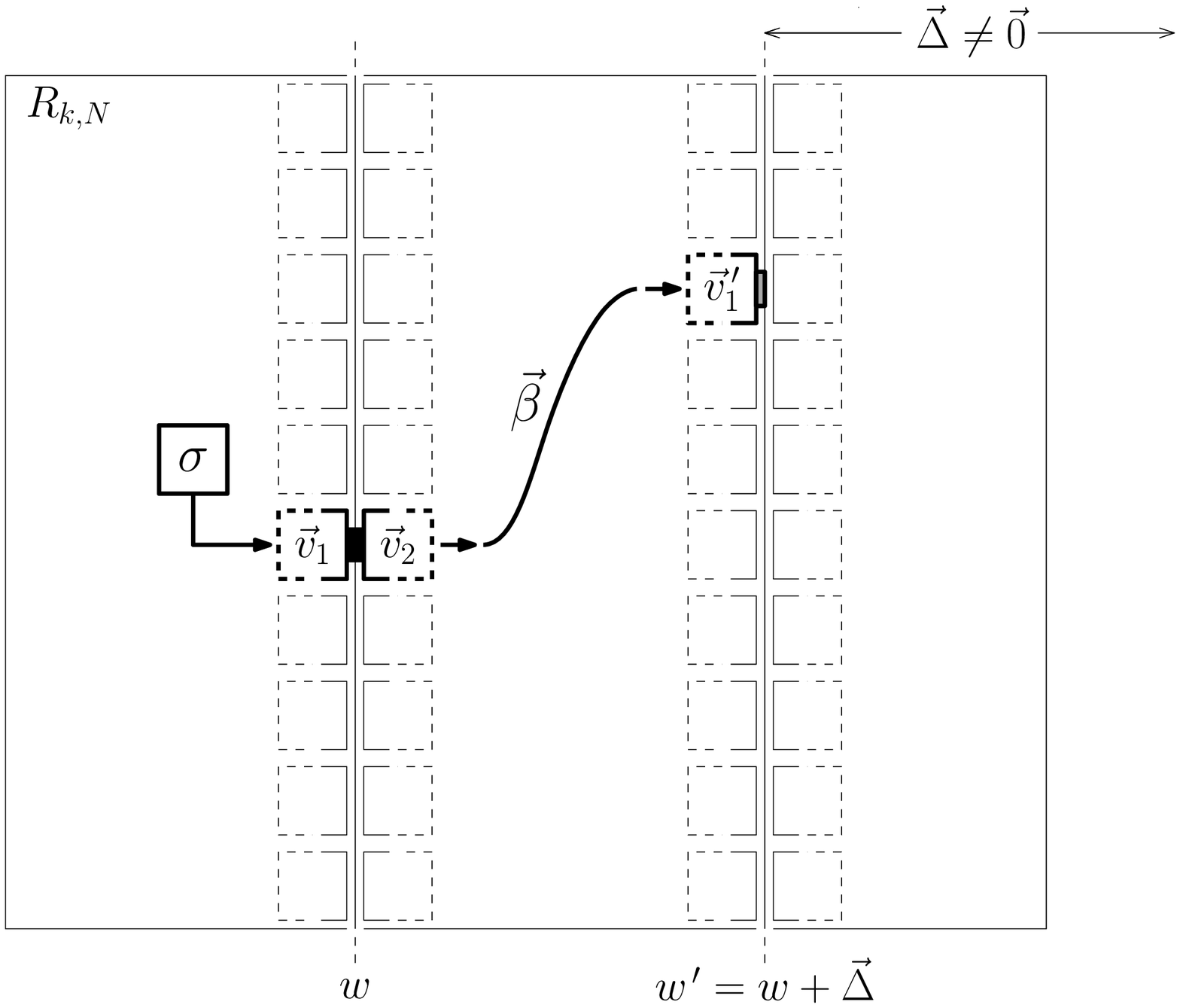}
        \caption{\label{fig:algorithm_trace_1} Right after Loop~1 has
          completed: The $\vec{\alpha}$ sub-path from $\sigma$ to
          $\vec{v}^{\,\prime}_1$ was used to initialize
          $\vec{\beta}$.}
    \end{subfigure}%
    ~
    \begin{subfigure}[t]{0.475\textwidth}
        \centering
        \includegraphics[width=\textwidth]{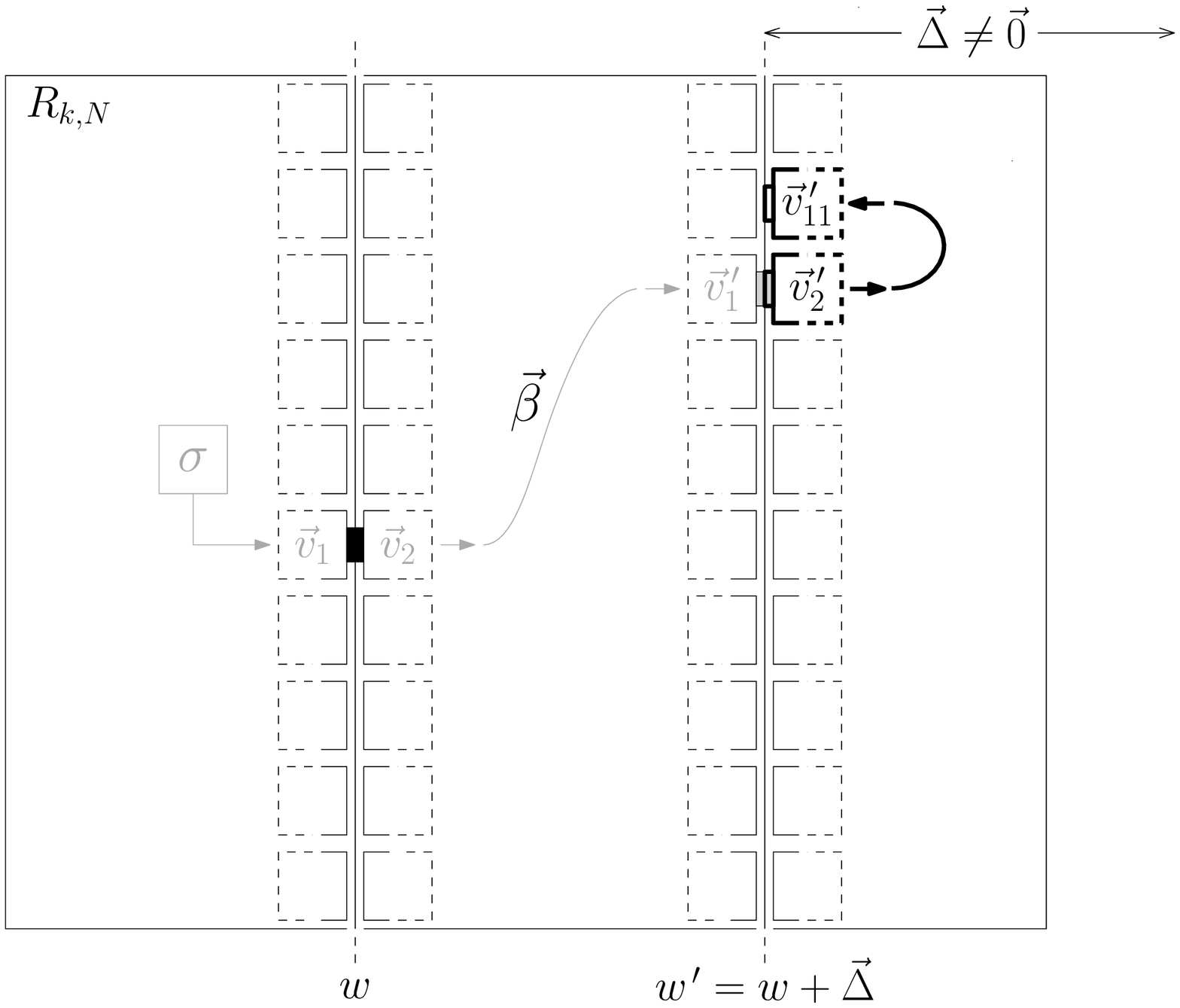}
        \caption{\label{fig:algorithm_trace_2} Right after Loop~2a has
          completed for the first time: The $\vec{\alpha}$ sub-path
          from $\vec{v}_{14}$ to $\vec{v}_{15}$ was translated by
          $\vec{\Delta}$ and appended to $\vec{\beta}$.}
    \end{subfigure}%

    \begin{subfigure}[t]{0.475\textwidth}
        \centering
        \includegraphics[width=\textwidth]{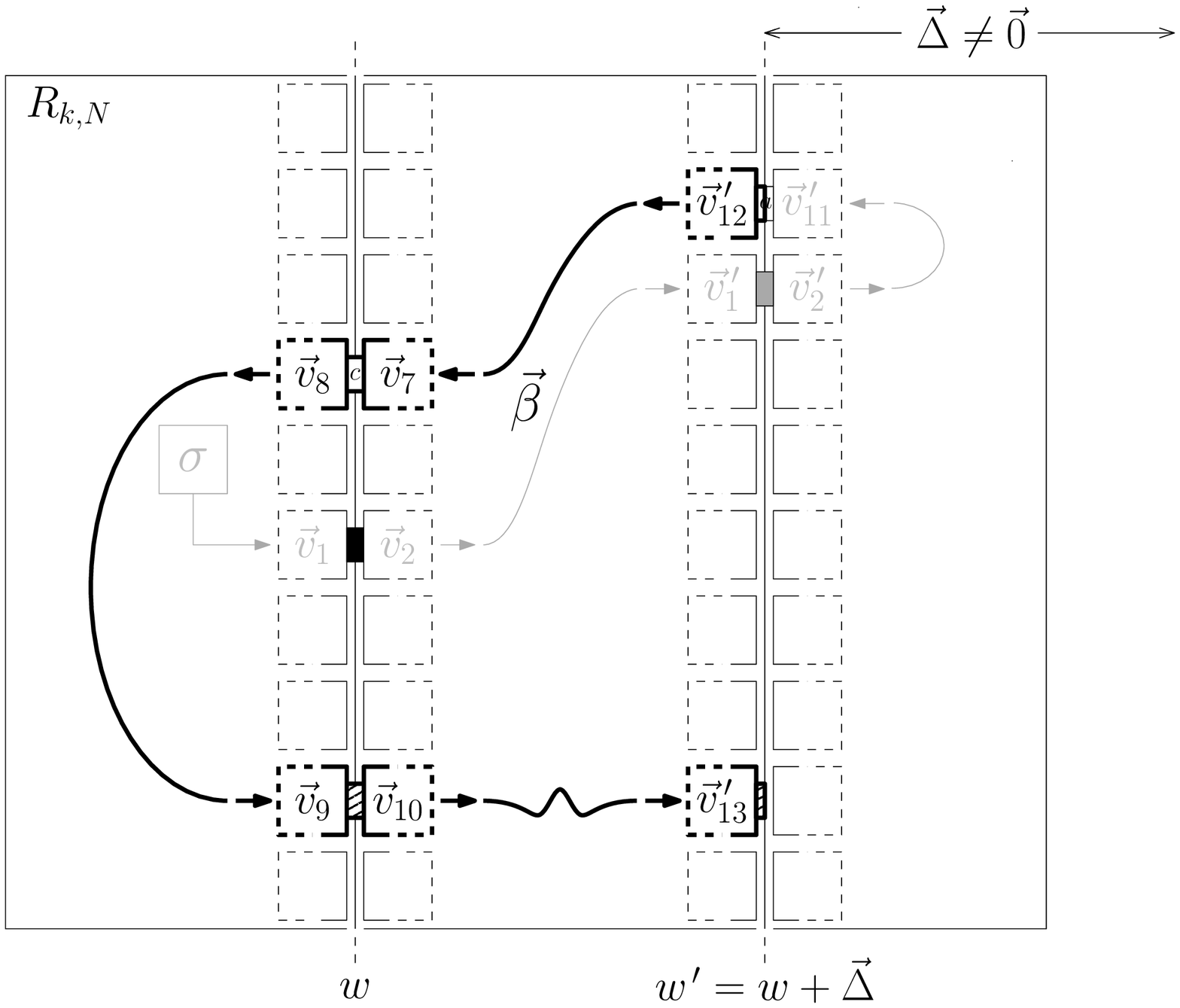}
        \caption{\label{fig:algorithm_trace_3} Right after Loop~2b has
          completed for the first time: The $\vec{\alpha}$ sub-path
          from $\vec{v}^{\,\prime}_{12}$ to $\vec{v}^{\,\prime}_{13}$
          was appended to $\vec{\beta}$.}
    \end{subfigure}%
    ~
    \begin{subfigure}[t]{0.475\textwidth}
        \centering
        \includegraphics[width=\textwidth]{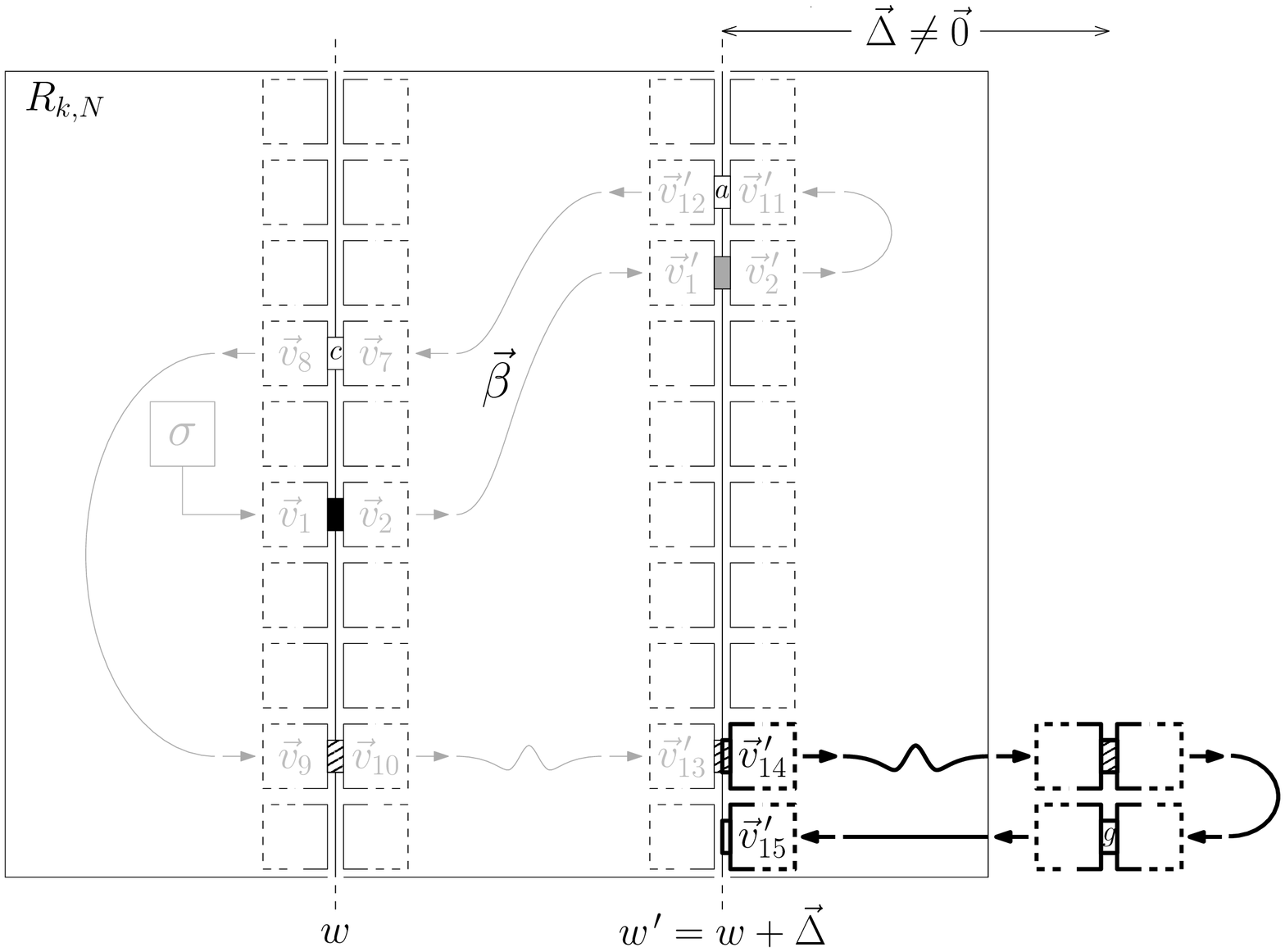}
        \caption{\label{fig:algorithm_trace_4} Right after Loop~2a has
          completed for the second time: The $\vec{\alpha}$ sub-path
          from $\vec{v}_{10}$ to $\vec{v}_{11}$ was translated by
          $\vec{\Delta}$ and appended to $\vec{\beta}$.}
    \end{subfigure}%

    \begin{subfigure}[t]{0.475\textwidth}
      \centering
      \includegraphics[width=\textwidth]{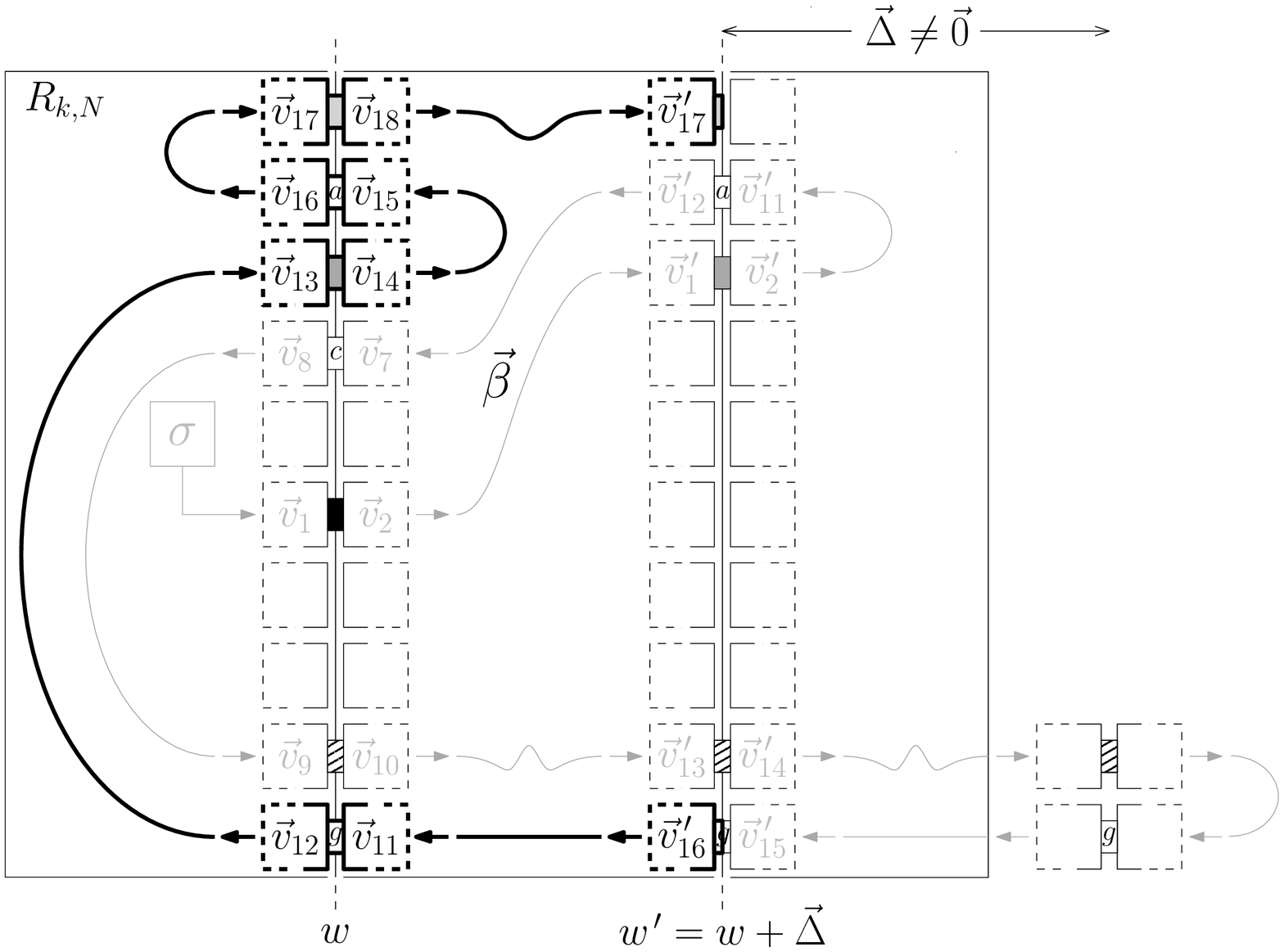}
        \caption{\label{fig:algorithm_trace_5} Right after Loop~2b has
          completed for the second time: The $\vec{\alpha}$ sub-path
          from $\vec{v}^{\,\prime}_{16}$ to $\vec{v}^{\,\prime}_{17}$
          was appended to $\vec{\beta}$.}
    \end{subfigure}%
    ~
    \begin{subfigure}[t]{0.475\textwidth}
        \centering
        \includegraphics[width=\textwidth]{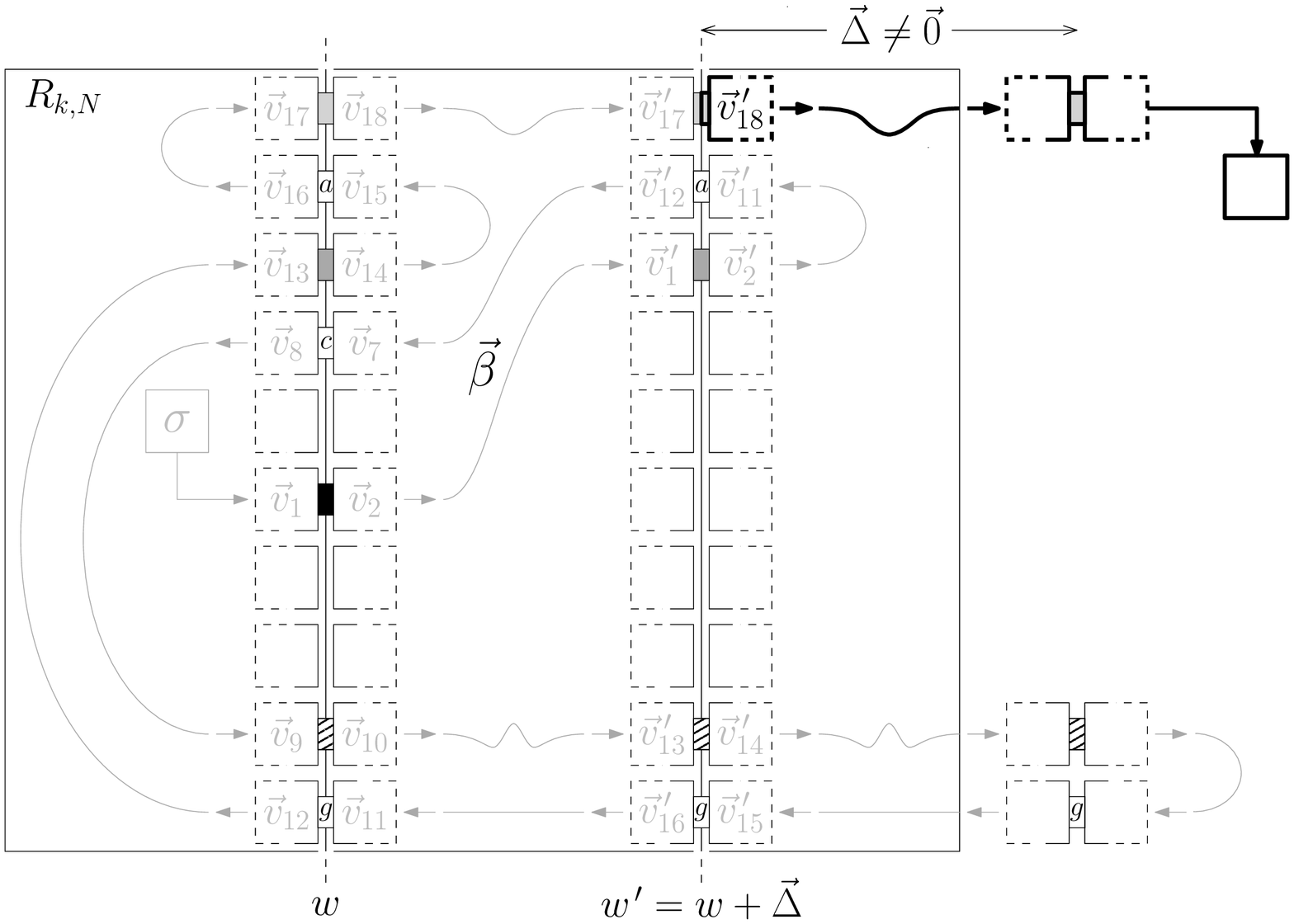}
        \caption{\label{fig:algorithm_trace_6} Right after Loop~3 (and
          the algorithm) has completed: The $\vec{\alpha}$ suffix
          starting with $\vec{v}_{18}$ was translated by
          $\vec{\Delta}$ and appended to $\vec{\beta}$.}
    \end{subfigure}%

    \caption{\label{fig:algorithm_trace} The trace of the algorithm shown
      in Figure~\ref{alg:beta-prime} when applied to the assembly
      sequence $\vec{\alpha}$ shown in
      Figure~\ref{fig:half_windows_equivalent}(a). In each sub-figure,
      the new sub-path is bolded and is a continuation of the sub-path
      in the previous one. The last sub-figure above shows the same
      assembly sequence $\vec{\beta}$ depicted in
      Figure~\ref{fig:half_windows_equivalent}(b).}
\end{figure}

See Section~\ref{sec:appendix-lower-bound} for the full proof of Lemma~\ref{lem:lemma-2}. We now have the necessary machinery for our lower bound, which is the following.

\addtocounter{theorem}{-2}
%
%
%
%
\begin{theorem}
\label{thm:theorem-1}
$K^1_{USA}\left(R^3_{k,N}\right)=\Omega\left( N^{\frac{1}{k}}\right)$.
\end{theorem}

The proof idea for Theorem~\ref{thm:theorem-1} is as follows. Assume $\mathcal{T}=(T,\sigma,1)$ is a directed, 3D TAS in which $R^3_{k,N}$ self-assembles. By Lemmas~\ref{lem:lemma-1} and~\ref{lem:lemma-2}, $N \leq 3 \cdot |G|^k  \cdot k \cdot 16^k$, where $G$ is the set of all glues in $T$. This means that $|T| \geq \frac{N^{\frac{1}{k}}}{576} = \Omega\left(N^{\frac{1}{k}}\right)$. See Section~\ref{sec:appendix-lower-bound} for the full proof of Theorem~\ref{thm:theorem-1}.


Theorem~\ref{thm:theorem-1} says that temperature-1 self-assembly in just-barely 3D is no more powerful than temperature-2 self-assembly in 2D. Interestingly, our lower bound for $K^1_{USA}\left(R^3_{k,N}\right)$ matches the lower bound for $K^1_{SA}\left(R^3_{k,N}\right)$ by Furcy, Summers and Wendlandt \cite{FurcySummersWendlandtDNA} but our bound is much more interesting than theirs because ours is roughly the square root of the best known upper bound, to which we turn our attention.

\section{The upper bound}
\label{sec:thin-rectangle-construction}

In this section, we give a construction that outputs aTAS in which a sufficiently large rectangle (of any height $k\geq 2$) $R^3_{k,N}$ uniquely self-assembles, testifying to our upper bound, which is roughly the square of our lower bound. 

\addtocounter{theorem}{0}

\begin{theorem}
\label{thm:main_positive_theorem}
$K^1_{USA}\left( R^3_{k,N} \right) = O\left( N^{\frac{1}{\left \lfloor \frac{k}{2} \right \rfloor}} + \log N \right)$.
\end{theorem}

Assume that $k > 3$, otherwise the construction is trivial. We use a counter whose base depends on the dimensions of the target rectangle. Let $w = \left\lfloor \frac{k}{2} \right\rfloor$ be the width (number of digits) of the counter. The base of the counter is $M = \left \lceil \left(\frac{N}{31} \right)^{\frac{1}{w}} \right\rceil$. The value of each digit is represented in binary, using a series of $m = \left\lceil \log M \right \rceil$ {\it bit bumps} that protrude from a horizontal line of tiles. Each bit bump geometrically encodes one bit as a corresponding assembly of tiles.

A novel and noteworthy feature of our construction is the organization of the digits of the counter into pairs of digits, where each pair of digits is contained within a rectangular \emph{digit region}. We say that a digit region is a \emph{general} digit region if its dimensions are four rows by $l = 9m + 22$ columns. If $k \mod 4 = 0$, then each general digit region, of which there are $\frac{w}{2}$, contains two digits. We will use a \emph{special} digit region with two rows and $l$ columns to handle the case where $k \mod 4 = 2$. Going forward, we will refer to a general digit region as simply a digit region. Throughout this section, we will assume $k \mod 4 = 0$. 

\begin{figure}[!h]
	\centering
	\includegraphics[width=.6 \textwidth ]{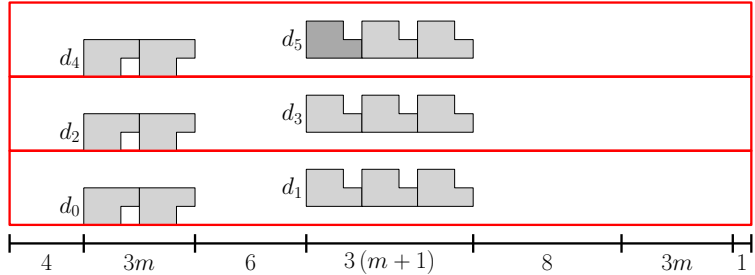}			\caption{\label{fig:Overview_digit_regions_dimensions}  This example shows how the digits that comprise a value of the counter are organized into digit regions. The next value of the counter would have a similar organization, to the east of the current value (see Figure~\ref{fig:Overview_digit_regions}). In this example, since $k = 12$, the value of the counter has six digits, $d_0$ through $d_5$, the latter being the most significant digit and the former the least significant. Even (odd) digits have even (odd) subscripts. In this example, each digit is encoded using two bits. The extra bit for odd digits indicates whether that digit is the most significant digit. Note that the least significant (westernmost) bit of $d_5$ is darkened to indicate that its value is 1, because $d_ 5$ is the most significant digit. Although this is a specific example, the general dimensions are given. Note that we include an ``extra'' $\Theta\left( m \right)$ columns in a general digit region in order to contain the most significant digit of the counter within a ``special'' digit region comprised of only two rows and $l$ columns, in the case where $k \mod 4 = 2$. }
\end{figure}

The westernmost digit within a general digit region is \emph{even}, and its bit bumps face toward the south. The easternmost digit is \emph{odd}, and its bit bumps face toward the north. The westernmost bit of each odd digit encodes whether that digit is the most significant digit that is contained in a general digit region. Figure~\ref{fig:Overview_digit_regions_dimensions} shows a high-level overview of how the digits (that comprise a value) of the counter are organized into digit regions.

A \emph{gadget}, referred to by a name like {\tt Gadget}, is a group of tiles that perform a specific task as they self-assemble. Each gadget, except for the seed-containing gadget, has one \emph{input} glue and at least one \emph{output} glue. For each gadget, the placement of the input and output glues can be inferred from the way the new gadgets bind to the assembly shown in the preceding figure. Glues internal to the gadget are configured to ensure unique self-assembly within the gadget. The strength of every glue is 1. If a glue contains some information $x$, this means that the glue label has a structure that contains the encoding of $x$, according to some fixed, standard encoding scheme.

We initialize the counter to start at a certain initial value $s$, padded out to $w$ digits, with leading 0s. In order to choose the initial value, let $n = \left\lfloor \frac{N}{l} \right\rfloor - 1$ be the number of increment steps. Then, we set $s = M^w - n$ to be the initial value. Then, once $s$ is set, a tile assembly representation of $s$ self-assembles via a series of gadgets. Figure~\ref{fig:Initial_overview_full_initial_value} shows a fully assembled example for $s=333332$.

\begin{figure}[!h]
	\centering
	\includegraphics[width=\textwidth ]{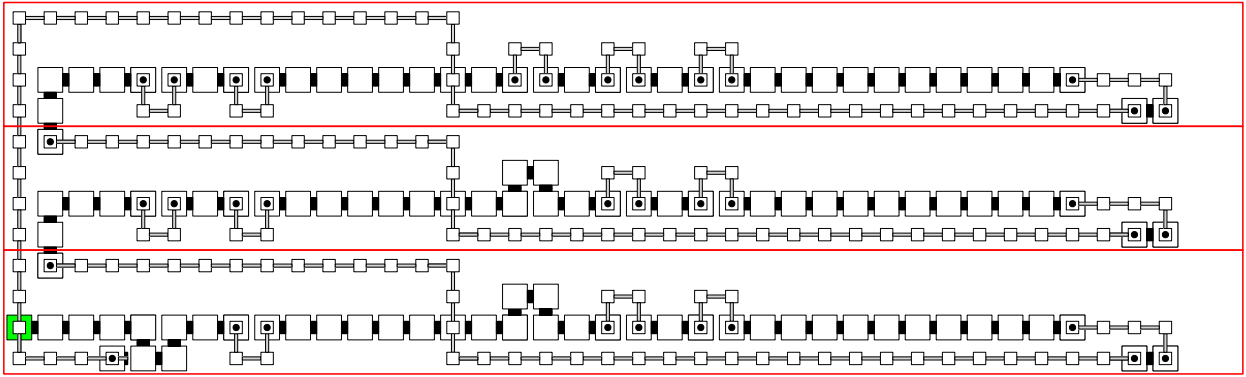}
	\caption{\label{fig:Initial_overview_full_initial_value} A fully assembled example of the initial value of the counter for $s=333332$. Following standard presentation conventions for just-barely 3D tile self-assembly, we use big squares to represent tiles placed in the $z=0$ plane and small squares to represent tiles placed in the $z=1$ plane. A glue between a $z=0$ tile and $z=1$ tile is denoted as a small disk. Glues between $z=0$ tiles are denoted as thick lines. Glues between $z=1$ tiles are denoted as thin lines. The leftmost tile in the $z=0$ plane is the seed tile. }
\end{figure}

After the initial value of the counter self-assembles (see Figures~\ref{fig:Initial_overview_seed_start} through~\ref{fig:Initial_overview_seed_to_next_significant_digit_region} in Section~\ref{sec:appendix-upper-bound-initial-value}), the counter undergoes a series of increment operations. Each increment operation increments the value of the counter by one. The counter counts up to the highest possible value, as determined by its base and the number of digits, increments once more to roll over to 0, and then stops. Finally, one could use $O(N \mod l) = O(l)$ filler tiles to fill in the remaining columns of the rectangle (we actually never explicitly specify this trivial step in our construction). Figure~\ref{fig:Overview_digit_regions} shows a high-level, artificial example of the behavior of the counter in terms of its increment steps.

\begin{figure}[!h]
	\centering
	\includegraphics[width=\textwidth]{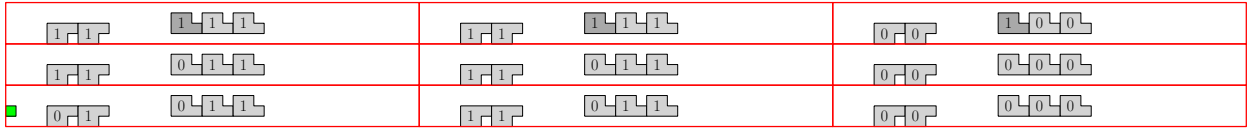}
	\caption{\label{fig:Overview_digit_regions}  In this artificial example, $M=4$ and $s=333332$. The counter increments through $333333$ and rolls over to $000000$ before stopping. Each digit region is outlined. For each digit region (other than a digit region that contains a digit of the final value of the counter), the digit region to its east is the \emph{corresponding} digit region.  }
\end{figure}

The basic idea of the general self-assembly algorithm for incrementing the value of the counter is to read an even digit in the current digit region, write its result in the corresponding digit region, come back to the current digit region and read the odd digit, write its result in the corresponding digit region. Then, do the same thing in the digit region in which the next two most significant digits are contained and stop after the most significant digit was read and the result was written.

The trick is to read each digit from the current digit region and write the respective result in the corresponding digit region without having to hard-code into the glues of the tiles both $\Theta\left( m \right)$ bits (representing the binary representation of the value of the digit that was just read), as well as the relative location along a path whose length is $\Theta(l)$. To accomplish this, we use gadgets whose names are prefixed with {\tt Repeating\_after\_} and {\tt Stopper\_after\_}.

In Figures~\ref{fig:General_overview_read_non_MSB} through~\ref{fig:General_overview_start_digit_region}, we create the gadgets that implement the general self-assembly algorithm that increments the value of the counter. Figures~\ref{fig:General_overview_read_non_MSB} through~\ref{fig:General_overview_start_digit_region} also show an example assembly sequence, where, unless specified otherwise, each figure continues the sequence from the resulting assembly in the previously-numbered figure, unless explicitly stated otherwise.

\begin{figure}[!h]
	\centering
	\includegraphics[width=\textwidth]{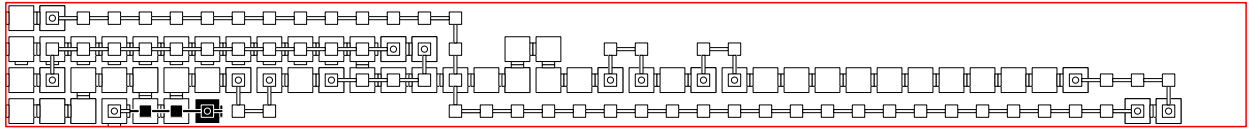}
	\caption{\label{fig:General_overview_read_non_MSB}  A {\tt Read\_non\_MSB} gadget is shown here. A {\tt Read\_non\_MSB} gadget reads the value of a bit that is not the most significant bit of a digit. The depicted gadget is a {\tt Read\_non\_MSB\_0} gadget. 
	For every digit region except the southernmost one, the {\tt Start\_digit\_region} gadget in Figure~\ref{fig:General_overview_start_digit_region} exposes output glues in both the $z=0$ and $z=1$ planes, from which only the correct {\tt Read\_non\_MSB} gadget self-assembles. 
	Otherwise, for the southernmost digit region, a {\tt Reset\_read\_even\_digit} gadget is used (see Figure~\ref{fig:General_overview_return_read_even_digit}).
	The last tile in a {\tt Read\_non\_MSB} gadget guesses the value of the next bit. The input glue of a {\tt Read\_non\_MSB} gadget being created here contains a binary string $x \in \{0, 1\}^{i}$, for $1 \leq i < m$, where the rightmost bit is the bit that this gadget reads.  The two output glues contain $x$ followed by the value of the next bit that will be read, i.e., $x0$ and $x1$. The input glues of the {\tt Read\_non\_MSB} gadgets being created here also contain a bit set to 0, which is the parity of the digit (even or odd) whose bits are being read, which, in this case, is even. This will allow us to use the general {\tt Read\_non\_MSB} gadgets to create specific gadgets to read the bits of both even and odd digits, respectively. In general, we create $O(1)$ {\tt Read\_non\_MSB} gadgets for each $x \in \{0, 1\}^i$, for $1 \leq i < m$, contributing $O(M)$ tile types.   }
\end{figure}

\begin{figure}[!h]
	\centering
	\includegraphics[width=\textwidth]{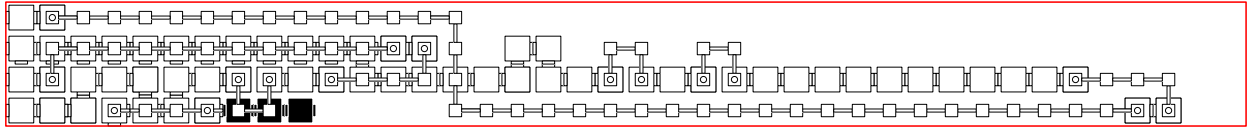}
	\caption{\label{fig:General_overview_read_MSB}  A {\tt Read\_MSB} gadget is shown here.    A {\tt Read\_MSB} gadget reads the value of the most significant bit of a digit. The input glue of a {\tt Read\_MSB} gadget being created here contains a binary string $x \in \{0,1\}^{m}$ and its output glue will also contain $x$. The rightmost bit of $x$ is the most significant bit of the digit that was just read. The input glues of the {\tt Read\_MSB} gadgets being created here also contain a bit set to 0, which is the parity of the digit (even or odd) whose bits were read, which, in this case, is even. This will allow us to use the general {\tt Read\_MSB} gadgets to create specific gadgets to read the MSBs of both even and odd digits, respectively. In general, we create $O(1)$ {\tt Read\_MSB} gadgets for each $x \in \{0, 1\}^{m}$, contributing $O(M)$ tile types.    }
\end{figure}

\begin{figure}[!h]
	\centering
	\includegraphics[width=\textwidth]{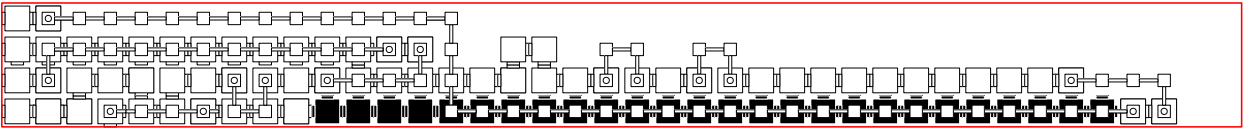}
	\caption{\label{fig:General_overview_repeating_after_even_digit}    A path of {\tt Repeating\_after\_even\_digit} gadgets is shown here. The \texttt{Repeating\_after\_even\_digit} gadget initiates the attachment of another \texttt{Repeating\_after\_even\_digit} tile, continuing the self-assembly of a path of repeating tiles toward and ultimately blocked from continuing by a {\tt Stopper\_after\_odd\_digit} gadget. We use {\tt Repeating\_after\_even\_digit} (and {\tt Repeating\_after\_odd\_digit}) gadgets to propagate $\Theta(m)$ bits along an arbitrarily long path of tiles, without also having the glues of the tiles along the path contain the relative location of each tile within the path. All gadgets whose name starts with {\tt Repeating\_after} essentially ``forget'' where they are and self-assemble in a line, until they cannot. In general, we create $O(1)$ {\tt Repeating\_after\_even\_digit} gadgets for each $x \in \{0,1\}^{m}$, contributing $O(M)$ tile types.   }
\end{figure}

\begin{figure}[!h]
	\centering
	\includegraphics[width=\textwidth]{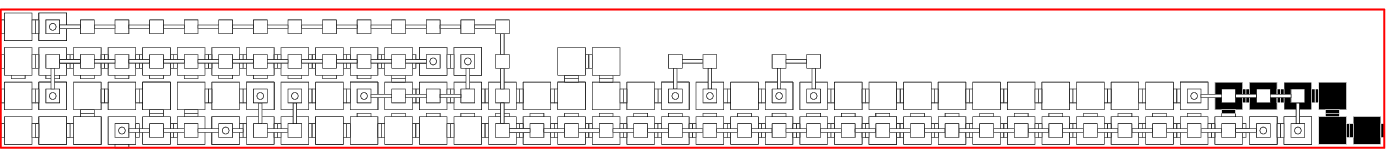} \\
	\includegraphics[width=\textwidth]{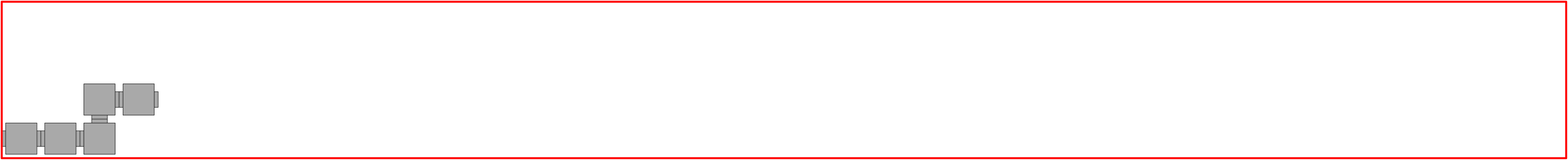}
	\caption{\label{fig:General_overview_at_stopper_after_odd_digit} An {\tt At\_stopper\_after\_odd\_digit} gadget is shown here. It has a fixed size. The {\tt At\_stopper\_after\_odd\_digit} gadget spans two adjacent digit regions. The black portion is in the current digit region and the gray portion is in the corresponding digit region. In general, if a gadget spans the current and corresponding digit regions, then the portion in the latter is depicted in gray and the former in black. The self-assembly of an {\tt At\_stopper\_after\_odd\_digit} gadget is initiated by the north-facing glue of the last {\tt Repeating\_after\_even\_digit} gadget to attach in the path in Figure~\ref{fig:General_overview_repeating_after_even_digit}. If $x \in \{0,1\}^{m}$ and $c \in \{0,1\}$ are contained in the output glue of the latter, where $c=1$ indicates the presence of an arithmetic carry and $c=0$ otherwise, then the output glue of the former contains the $m$-bit binary representation of $\left( x + c \right) \mod M$. If $\left( x + c \right) \mod M = 0$, then $c=1$ is contained in the output glue of the gadgets being created here. In general, we create $O(1)$ {\tt At\_stopper\_after\_odd\_digit} gadgets for each $x\in \{0,1\}^{m}$, contributing $O(M)$ tile types.  }
\end{figure}

\begin{figure}[!h]
	\centering
	\includegraphics[width=\textwidth]{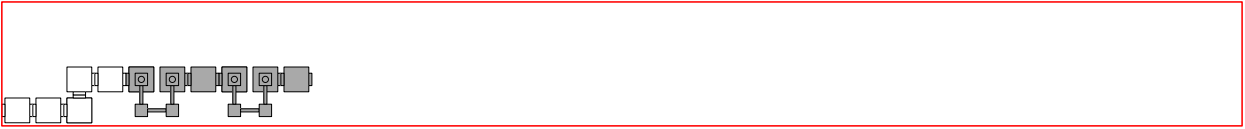}
	\caption{\label{fig:General_overview_write_even_digit}  A series of two {\tt Write\_even\_digit} gadgets is shown here. Intuitively, {\tt Write} gadgets ``undo'' what the {\tt Read} gadgets do. To that end, the input glue of a {\tt Write\_even\_digit} gadget being created here contains a binary string  $bx$, where $b \in \{0,1\}$, and $x \in \{0,1\}^{i}$, for $0 \leq i < m$, its output glue contains $x$, and the corresponding bit bump for $b$ self-assembles. In general, we create $O(1)$ {\tt Write\_even\_digit} gadgets for each $x \in \{0,1\}^i$, for $0 \leq i < m$, contributing $O(M)$ tile types.
	}
\end{figure}

\begin{figure}[!h]
	\centering
	\includegraphics[width=\textwidth]{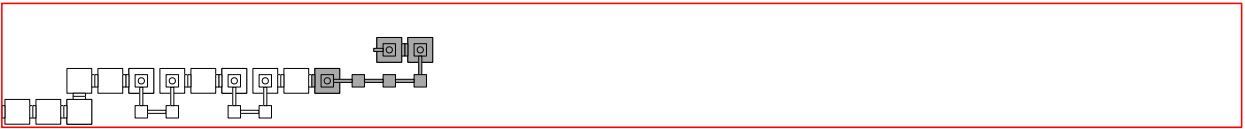}
	\caption{\label{fig:General_overview_stopper_after_even_digit}       A {\tt Stopper\_after\_even\_digit} gadget is shown here.   A \texttt{Stopper\_after\_even\_digit} gadget is used to stop a subsequent path of {\tt Repeating\_after\_odd\_digit} gadgets that will be propagating the value of an odd digit as they self-assemble. In general, we create $O(1)$ {\tt Stopper\_after\_even\_digit} gadgets, contributing $O(1)$ tile types.  }
\end{figure}

\begin{figure}[!h]
	\centering
	\includegraphics[width=\textwidth]{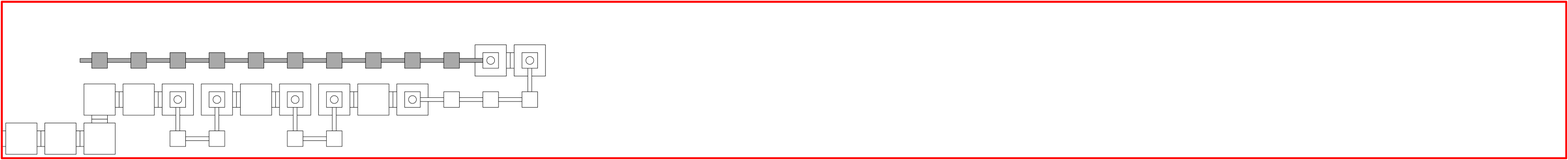}
	\caption{\label{fig:General_overview_single_tile_opposite_0}   A path of \texttt{Single\_tile\_opposite} gadgets is shown here. We create $O(1)$  \texttt{Single\_tile\_opposite} gadgets for each location in the general version of the depicted path of length $3m + 4$, contributing $O( m)$ tile types. }
\end{figure}

\begin{figure}[!h]
	\centering
	\includegraphics[width=\textwidth]{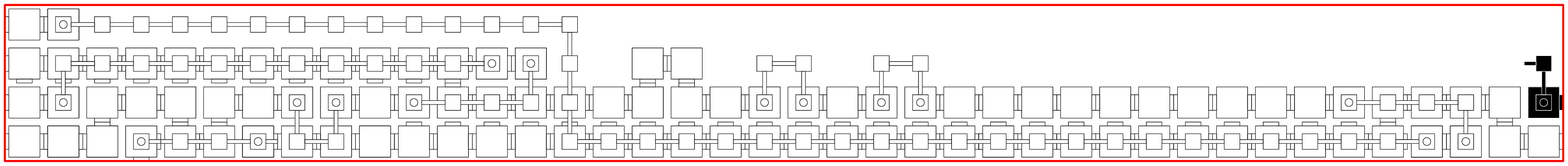}
	\includegraphics[width=\textwidth]{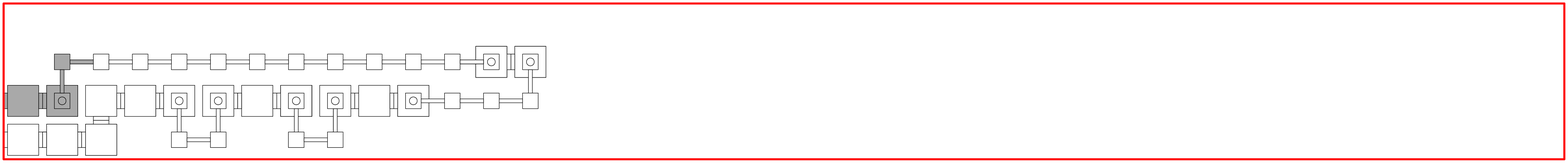}
	\caption{\label{fig:General_overview_between_digit_regions}  A {\tt Between\_digit\_regions} gadget, originating in the corresponding digit region (bottom) and terminating back in the current digit region (top), is shown here. In general, we create $O(1)$ {\tt Between\_digit\_regions}, contributing $O(1)$ tile types.  }
\end{figure}

\begin{figure}[!h]
	\centering
	\includegraphics[width=\textwidth]{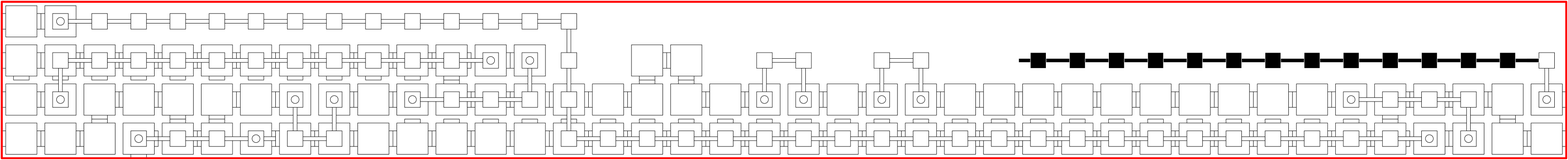}
	\caption{\label{fig:General_overview_single_tile_opposite_1}  A path of \texttt{Single\_tile\_opposite} gadgets is shown here. We create $O(1)$  \texttt{Single\_tile\_opposite} gadgets for each location in the general version of the depicted path of length $3m + 7$, contributing $O(m)$ tile types.           }
\end{figure}

\begin{figure}[!h]
	\centering
	\includegraphics[width=\textwidth]{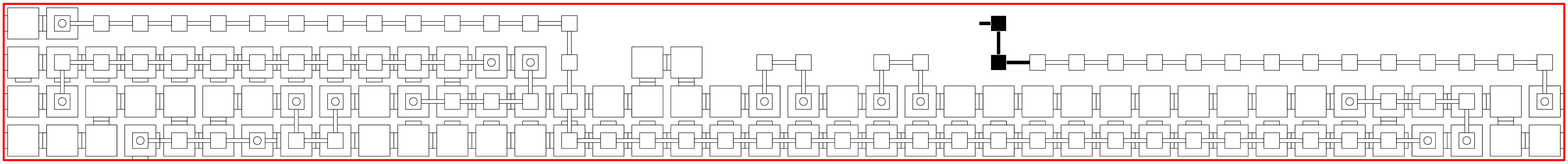}
	\caption{\label{fig:General_overview_at_MSB_of_odd_digit}     An {\tt At\_MSB\_of\_odd\_digit} gadget is shown here. In general, we create $O(1)$ {\tt At\_MSB\_of\_odd\_digit} gadgets, contributing $O(1)$ tile types.   }
\end{figure}

\begin{figure}[!h]
	\centering
	\includegraphics[width=\textwidth]{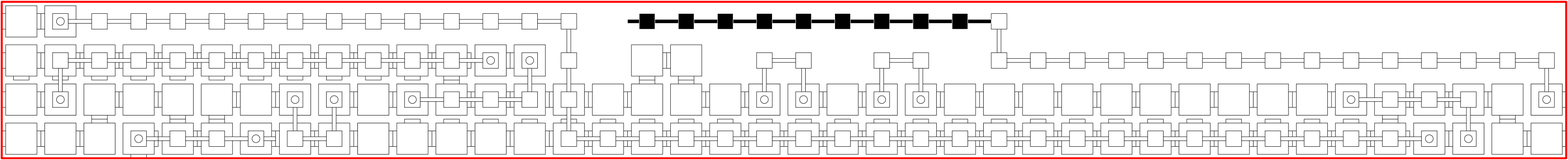}
	\caption{\label{fig:General_overview_single_tile_opposite_2}   A path of \texttt{Single\_tile\_opposite} gadgets is shown here. We create $O(1)$ \texttt{Single\_tile\_opposite} gadgets for each location in the general version of the depicted path of length $3 \left(m + 1\right)$, contributing $O(m)$ tile types.        }
\end{figure}

\begin{figure}[!h]
	\centering
	\includegraphics[width=\textwidth]{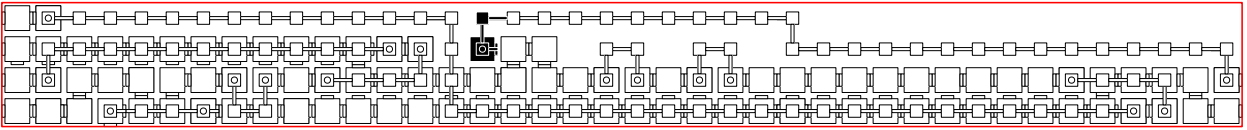}
	\caption{\label{fig:General_overview_start_read_odd_digit}  A {\tt Start\_read\_odd\_digit} gadget is shown here. 
	A {\tt Start\_read\_odd\_digit}, like a {\tt Start\_digit\_region} gadget from Figure~\ref{fig:General_overview_start_digit_region} does for an even digit, ``guesses'' the value of the most significant digit indicator bit in an odd digit by exposing output glues in both the $z=0$ and $z=1$ planes. 
	In general, we create $O(1)$ {\tt Start\_read\_odd\_digit} gadgets, contributing $O(1)$ tile types. 
	The {\tt Read\_} gadgets for the odd digit that attach to this gadget (and are depicted in Figure~\ref{fig:General_overview_repeating_after_odd_digit}) are created in a similar manner and contribute the same number of tile types as the {\tt Read\_} gadgets for the even digits (see Figures~\ref{fig:General_overview_read_non_MSB} and~\ref{fig:General_overview_read_MSB}), except the input glues of the former gadgets contain a bit set to 1, which is the parity of the digit whose bits were being read.}
\end{figure}

\begin{figure}[!h]
	\centering
	\includegraphics[width=\textwidth]{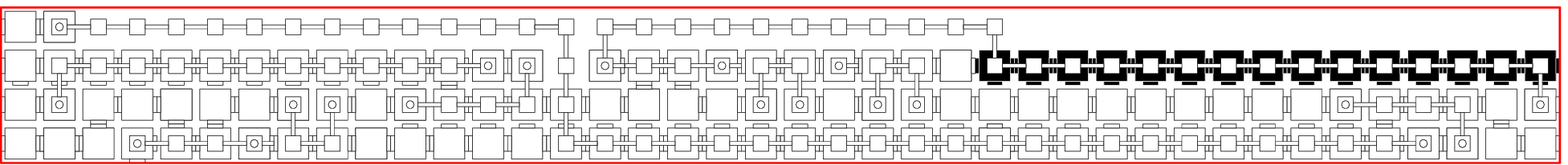}
	\includegraphics[width=\textwidth]{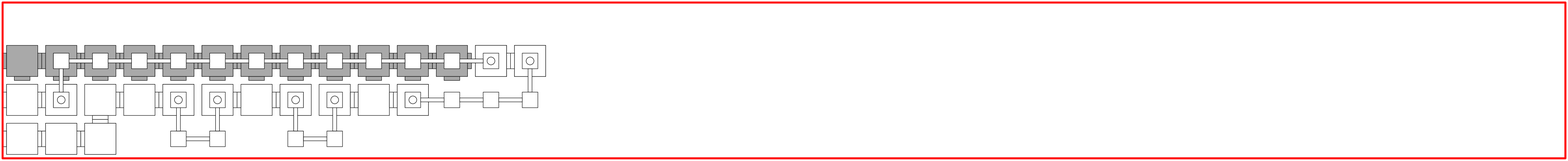}
	\caption{\label{fig:General_overview_repeating_after_odd_digit}       A path of {\tt Repeating\_after\_odd\_digit} gadgets is shown here.  This path is eventually hindered by a {\tt Stopper\_after\_even\_digit} gadget (see Figure~\ref{fig:General_overview_stopper_after_even_digit}). In general, we create $O(1)$ {\tt Repeating\_after\_odd\_digit} gadgets for each $x\in \{0,1\}^{m+1}$, contributing $O(M)$ tile types.  }
\end{figure}

\begin{figure}[!h]
	\centering
	\includegraphics[width=\textwidth]{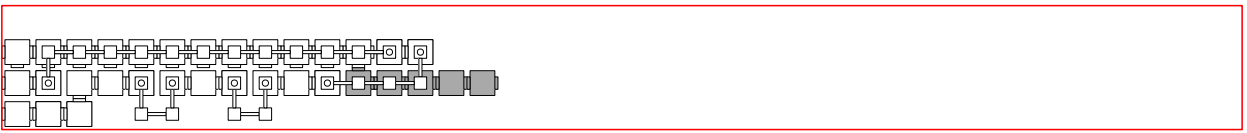}
	\caption{\label{fig:General_overview_at_stopper_after_even_digit}   An {\tt At\_stopper\_after\_even\_digit} gadget is shown here. It has a fixed size. The south-facing glue of the last {\tt Repeating\_after\_odd\_digit} gadget to attach in the path shown in Figure~\ref{fig:General_overview_repeating_after_odd_digit} will initiate the self-assembly of an {\tt At\_stopper\_after\_even\_digit} gadget. If, for $b \in \{0,1\}$ and $x \in \{0,1\}^m$, $bx$ and $c \in \{0,1\}$ are contained in the output glue of the former, where $c=1$ indicates the presence of an arithmetic carry and $c=0$ otherwise, then the output glue of the latter contains the result of prepending $b$ to the $m$-bit binary representation of $\left( x + c \right) \mod M$. If $\left( x + c \right) \mod M = 0$, then $c=1$ is contained in the output glue of the gadgets being created here. In general, we create $O(1)$ {\tt At\_stopper\_after\_even\_digit} gadgets for each $bx \in \{0,1\}^{m + 1}$, contributing $O(M)$ tile types.}
\end{figure}

\begin{figure}[!h]
	\centering
	\includegraphics[width=\textwidth]{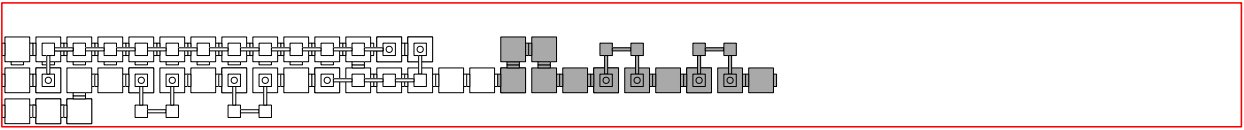}
	\caption{\label{fig:General_overview_write_odd_digit}  A series of three {\tt Write\_odd\_digit} gadgets is shown here.  The input glue of a {\tt Write\_odd\_digit} gadget being created here contains a binary string  $bx$, where $b \in \{0,1\}$, and $x \in \{0,1\}^{i}$, for $0 \leq i \leq m$, its output glue contains $x$, and the corresponding bit bump for $b$ will self-assemble. In general, we create $O(1)$ {\tt Write\_odd\_digit} gadgets for each $x \in \{0,1\}^i$, for $0 \leq i \leq m$, contributing $O(M)$ tile types. }
\end{figure}

\begin{figure}[!h]
	\centering
	\includegraphics[width=\textwidth]{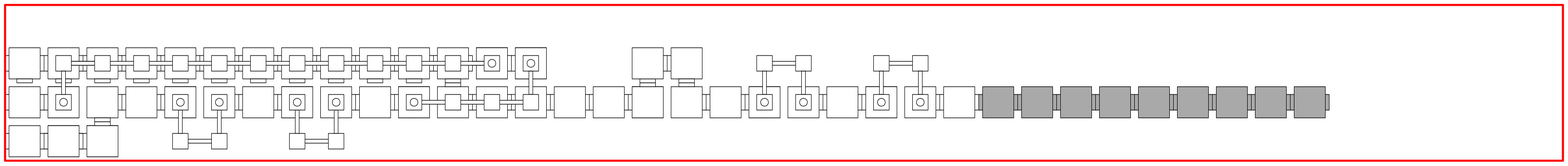}
	\caption{\label{fig:General_overview_single_tile} A path of {\tt Single\_tile} gadgets is shown here. We create $O(1)$ {\tt Single\_tile} gadgets for each location in the general version of the depicted path of length $3m + 3$, contributing $O(m)$ tile types. }
\end{figure}

\begin{figure}[!h]
	\centering
	\includegraphics[width=\textwidth]{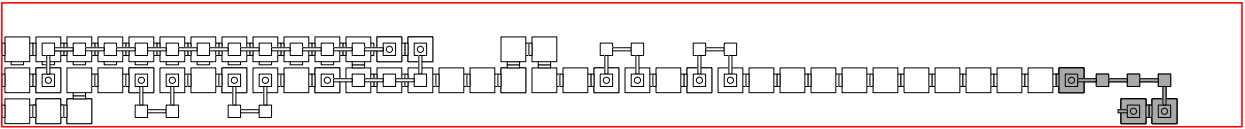}
	\caption{\label{fig:General_overview_stopper_after_odd_digit}  A {\tt Stopper\_after\_odd\_digit} gadget is shown here. 
	The \texttt{Stopper\_after\_odd\_digit} gadget geometrically marks the location of (a constant distance from) the easternmost edge of the \emph{current} digit region in which it self-assembles. It will ultimately block a subsequent path of repeating tiles that will be propagating the value of an even digit from the current digit region to the corresponding adjacent digit region for the next value of the counter.
	In general, we create $O(1)$ {\tt Stopper\_after\_odd\_digit} gadgets, contributing $O(1)$ tile types. }
\end{figure}

\begin{figure}[!h]
	\centering
	\includegraphics[width=\textwidth]{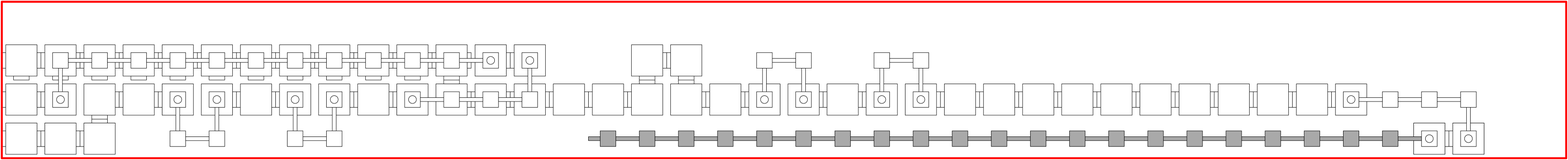}
	\caption{\label{fig:General_overview_single_tile_opposite_3}  A path of {\tt Single\_tile\_opposite} gadgets is shown here. We create $O(1)$ \texttt{Single\_tile\_opposite} gadgets for each location in the general version of the depicted path of length $1+3\left(m + 1\right) + \left(8 + 3m - 4 \right) + 1$, contributing $O(m)$ tile types.           }
\end{figure}

\begin{figure}[!h]
	\centering
	\includegraphics[width=\textwidth]{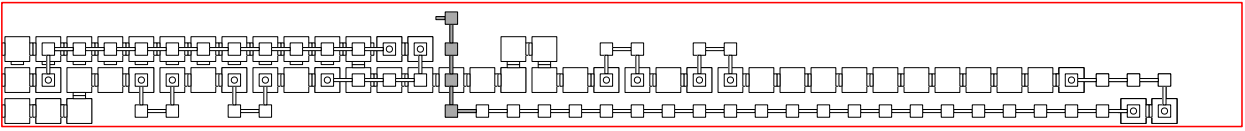}
	\caption{\label{fig:General_overview_between_digits}    A {\tt Between\_digits} gadget is shown here.  In general, we create $O(1)$ {\tt Between\_digits} gadgets, contributing $O(1)$ tile types.  }
\end{figure}

\begin{figure}[!h]
	\centering
	\includegraphics[width=\textwidth]{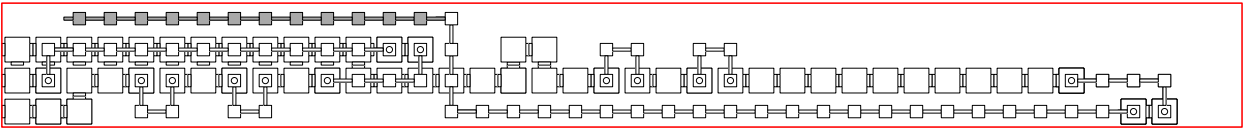}
	\caption{\label{fig:General_overview_single_tile_opposite_4}  A path of {\tt Single\_tile\_opposite} gadgets is shown here. 
	If the odd digit that just self-assembled is the most significant digit and the value of the counter did not roll over to 0, then turn the corner (see the {\tt Reset\_turn\_corner} gadget in Figure~\ref{fig:General_overview_return_turn_corner}) and return to execute another increment operation. If the value of the counter rolled over to 0, which could have been detected by the last {\tt Write\_odd\_digit} gadget created in Figure~\ref{fig:General_overview_write_odd_digit}, then no further increment operations are executed.
	We create $O(1)$  \texttt{Single\_tile\_opposite} gadgets for each location in the general version of the depicted path of length $3 m + 6$, contributing $O(m)$ tile types.         }
\end{figure}

\begin{figure}[!h]
	\centering
	\includegraphics[width=\textwidth]{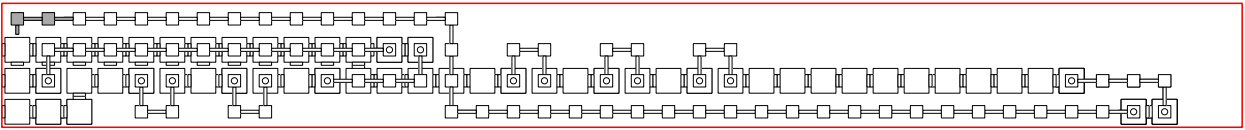}
	\caption{\label{fig:General_overview_return_turn_corner}  A {\tt Reset\_turn\_corner} gadget is shown here. In general, we create one {\tt Reset\_turn\_corner} gadget, contributing $O(1)$ tile types.}
\end{figure}

\begin{figure}[!h]
		\centering
		\includegraphics[width=\textwidth ]{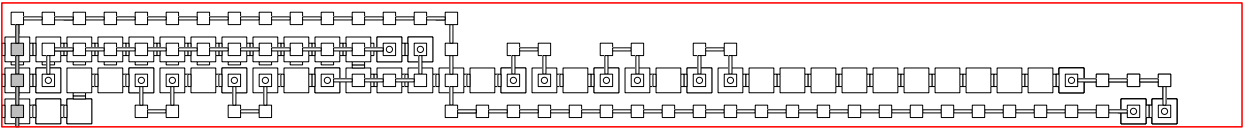}
		\caption{\label{fig:General_overview_return_single_tile} A (beginning portion of a) path of \texttt{Reset\_single\_tile} gadgets is shown here. Note that a {\tt Reset\_single\_tile} gadget is comprised of a single tile whose input glue is always north-facing, and whose output glue is always south-facing. We create one {\tt Reset\_single\_tile} gadget for each location in the general version of the depicted path of length $k - 2$, contributing $O(k)$ tile types.   }
	\end{figure}

\begin{figure}[!h]
		\centering
		\includegraphics[width=\textwidth ]{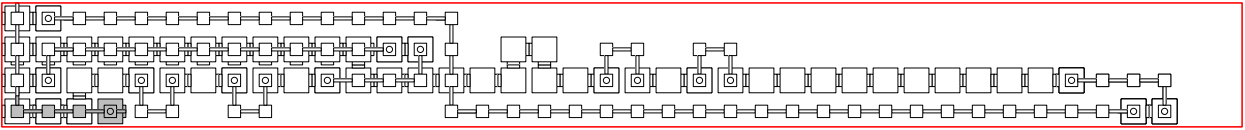}
		\caption{\label{fig:General_overview_return_read_even_digit} The \texttt{Reset\_read\_even\_digit} gadget is shown here. The \texttt{Reset\_read\_even\_digit} gadget initiates the execution of the next increment step. The {\tt Reset\_read\_even\_digit} gadget ``guesses'' the value of the first bit in the least significant digit by exposing output glues in both the $z=0$ and $z=1$ planes. 
		In general, we create one {\tt Reset\_read\_even\_digit} gadget, contributing $O(1)$ tile types.   }
	\end{figure}

\begin{figure}[!h]
	\centering
	\includegraphics[width=\textwidth]{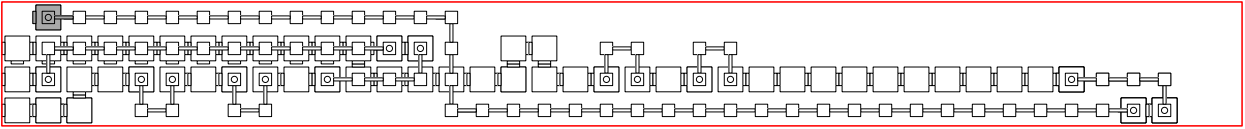}
	\caption{\label{fig:General_overview_z_1_to_z_0}  A {\tt Z1\_to\_z0} gadget is shown here. If the odd digit that just self-assembled is not the most significant digit, then proceed back to the current digit region and ultimately into the digit region in which the next two most significant digits are contained. This gadget transitions the path of {\tt Single\_tile\_opposite} gadgets from Figure~\ref{fig:General_overview_single_tile_opposite_4} from the $z=1$ plane to the $z=0$ plane in order to ensure a clear path in the $z=1$ plane for a subsequent path of {\tt Reset\_single\_tile} gadgets (see Figure~\ref{fig:General_overview_return_single_tile}). In general, we create $O(1)$ {\tt Z1\_to\_z0} gadgets, contributing $O(1)$ tile types.}
\end{figure}

\clearpage

\begin{figure}[!h]
	\centering
	\includegraphics[width=\textwidth]{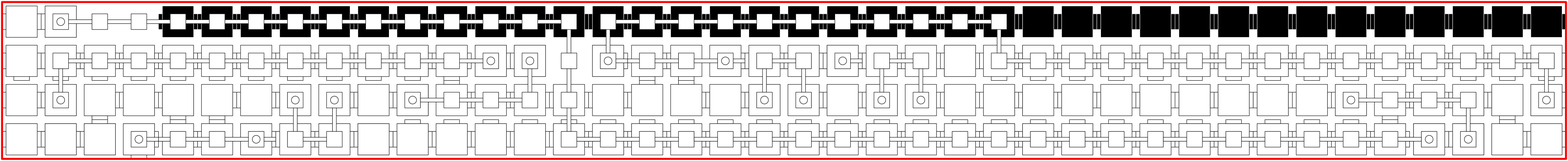}
	\includegraphics[width=\textwidth]{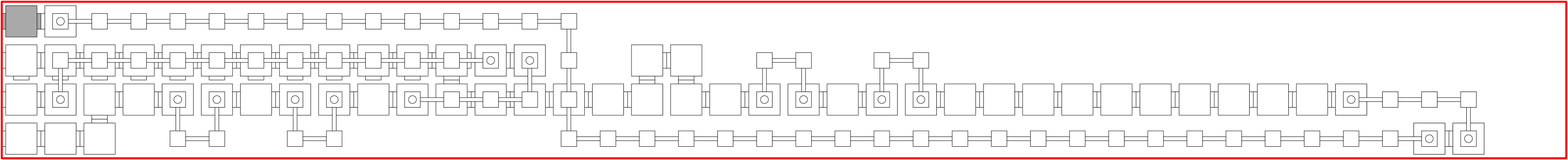}
	\caption{\label{fig:General_overview_single_tile_opposite_5}  A path of {\tt Single\_tile\_opposite} gadgets is shown here.
	A {\tt Start\_digit\_region} gadget (created in Figure~\ref{fig:General_overview_start_digit_region}) binds to the west-facing glue of the last {\tt Single\_tile\_opposite} in the depicted path. We create $O(1)$  \texttt{Single\_tile\_opposite} gadgets for each location in the general version of the depicted path of length $3m + 6 + 3\left( m + 1 \right) + 8 + 3m + 1 + 1$, contributing $O(m)$ tile types.        }
\end{figure}

\begin{figure}[!h]
	\centering
	\includegraphics[width=\textwidth]{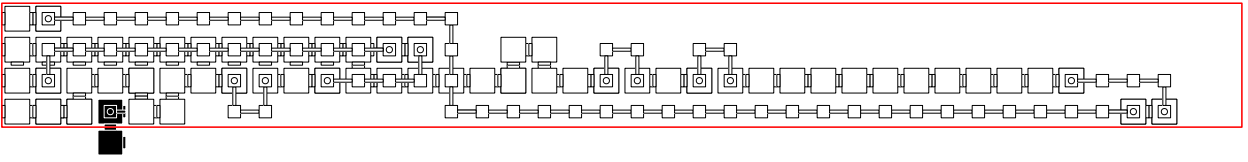}
	\caption{\label{fig:General_overview_start_digit_region} A {\tt Start\_digit\_region} gadget is shown here.   A {\tt Start\_digit\_region} gadget ``guesses'' the value of the least significant bit in a non-least significant even digit by exposing output glues in both the $z=0$ and $z=1$ planes. The value of the least significant bit in the least significant digit is guessed by the {\tt Reset\_read\_even\_digit} gadget created in Figure~\ref{fig:General_overview_return_read_even_digit}. In general, we create $O(1)$ {\tt Start\_digit\_region} gadgets, contributing $O(1)$ tile types.         }
\end{figure}

If $N$ is assumed to be sufficiently large, then the total number of tile types contributed by all the gadgets that were created in Figures~\ref{fig:General_overview_read_non_MSB} through~\ref{fig:General_overview_start_digit_region} when $k \mod 4 = 0$, is $O(M + m + k)$. Moreover, the total number of tile types contributed by all the gadgets that we use to self-assemble the initial value is $O(km)$ (see Figures~\ref{fig:Initial_overview_seed_start} through~\ref{fig:Initial_overview_seed_to_next_significant_digit_region} in Section~\ref{sec:appendix-upper-bound-initial-value}). Note that $k m = k \left\lceil \log \left \lceil \left( \frac{N}{31} \right)^{\frac{1}{w}} \right \rceil \right \rceil = k \left \lceil \log \left(\frac{N}{31}\right)^{\frac{1}{w}} \right \rceil \leq 2k \log N^{\frac{1}{w}} = \frac{2k}{w}\log N = \frac{2k}{\left\lfloor \frac{k}{2}\right \rfloor}\log N \leq \frac{2k}{\frac{k}{3}}\log N = O(\log N)$ and $M = \left\lceil \left( \frac{N}{31} \right)^{\frac{1}{w}}\right\rceil  = \left\lceil \left( \frac{N}{31} \right)^{\frac{1}{\left \lfloor \frac{k}{2} \right \rfloor }}\right\rceil = O\left( N^{\frac{1}{\left \lfloor \frac{k}{2} \right\rfloor}} \right)$. Thus, the size of the tile set output by our construction, when $k \mod 4 = 0$, is $O\left( N^{\frac{1}{\left\lfloor \frac{k}{2} \right \rfloor}} + \log N \right)$. Observe that, if $R^3_{k,N}$ is a thin rectangle, then $k < \frac{\log N}{\log \log N - \log \log \log N} < \frac{\log N}{\log \log N - \frac{1}{3}\log \log N} = \frac{\frac{3}{2}\log N}{\log \log N}$, and we have $\log N = 2^{\log \log N} = \left( N^{\frac{1}{\log N}}\right)^{\log \log N} = N^{\frac{\frac{3}{2}}{\frac{\frac{3}{2} \log N}{\log \log N}}}=O\left( N^{\frac{\frac{3}{2}}{k}} \right) = O\left( N^{\frac{1}{\frac{k}{2}}} \right) \ = \ O\left( N^{\frac{1}{\left \lfloor \frac{k}{2} \right \rfloor}} \right)$.

The case of $k \mod 4 = 2$ can be handled similarly, using a special case digit region in which the most significant digit is represented using two rows and $l$ columns (see Section~\ref{sec:appendix-upper-bound-special}).

The cases where $k \mod 4 \in \{1,3\}$ can be handled by using $O(1)$ tiles that self-assemble into an additional row.

The full details for our construction, in which all cases are handled, can be found in Section~\ref{sec:appendix-upper-bound-full}.

\section{Future work}
\label{sec:conclusion}
In this paper, we gave improved bounds on $K^1_{USA}\left(R^3_{k,N}\right)$. Specifically, our upper bound, $O\left( N^{\frac{1}{\left\lfloor \frac{k}{2} \right\rfloor}} + \log N\right)$, is roughly the square of our lower bound, $\Omega\left( N^{\frac{1}{k}}\right)$. However, questions still remain, upon which we feel future work should be based. Is it the case that either $K^1_{SA}\left(R^3_{k,N}\right)$ or $K^1_{USA}\left(R^3_{k,N}\right)$ is equal to $O\left(N^{\frac{1}{k}} + k\right)$? If not, then what are tight bounds for $K^1_{SA}\left(R^3_{k,N}\right)$ and  $K^1_{USA}\left(R^3_{k,N}\right)$?

\bibliographystyle{amsplain}
\bibliography{tam}

\clearpage

\appendix

\section{Lower bound appendix}
\label{sec:appendix-lower-bound}

This section contains all the proofs related to our lower bound.

%
%
%
%
\begin{lemma*}

\end{lemma*}

\begin{proof}

Let $e$ be a fixed odd number such that $1 \leq e < 2k$. 
Let $w$ be any window such that $M = M_{\vec{\alpha},w}~\upharpoonright~s = \left(\vec{v}_1,g_1\right), \ldots, \left(\vec{v}_{2e},g_{2e}\right)$ is a non-empty restricted glue window submovie. 
We will assume that $e$ represents the number of times that
$\vec{\alpha}$ crosses $w$ (going either away from or toward the seed)
as it follows $s$.
Here, $e$ can be at most $2k-1$ because $w$ is a translation of the $yz$-plane and $s \subseteq R^3_{k,N}$.
\begin{enumerate}
	\item First, we count the number of ways to choose the set $\left\{  \vec{v}_1, \ldots, \vec{v}_{2e} \right\}$, or the set of locations of $M$. Clearly, there are $\binom{4k}{2e}$ ways to choose a subset of $2e$ locations from a set comprised of $4k$ locations. However, for $M = M_{\vec{\alpha}, w}$, since $\vec{\alpha}$ follows a simple path, it suffices to count the number of ways to choose the set $\left\{  \vec{v}_{2i-1} \ \left| \ i=1, \ldots, e \right. \right\}$. This is because, once we choose a location of $M$, the location that is adjacent to the chosen location but on the opposite side of $w$ is determined. There are $\binom{2k}{e}$ ways to choose a subset of $e$ elements from a set comprised of $2k$ elements. 
	
	\item Next, we count the number of ways to choose the set
          $\left\{ \vec{v}_{4i-2} \ \left| \ 1 \leq i \leq
          \frac{e+1}{2} \right. \right\}$. Intuitively, this is the
          set of locations at which $\vec{\alpha}$ finishes crossing
          $w$ going away from the seed. Observe that each chosen
          location $\vec{v}$ of $M$ is either on the far side of $w$
          from the seed or on the near side of $w$ to the seed.
	Furthermore, $\vec{v}$ is paired up with a different chosen location $\vec{v}'$ of $M$ in the sense that $\vec{v}$ is adjacent to but on the opposite side of $w$ from $\vec{v}'$.
	Thus, choose $\frac{e+1}{2}$ locations from the set comprised of the $e$ chosen locations that are on the far side of $w$ from the seed. For each $1 \leq i \leq \frac{e+1}{2}$, there is a unique element in this set that will be assigned to $\vec{v}_{4i-2}$. There are $\binom{e}{(e+1)/2}$ ways to choose $\frac{e+1}{2}$ elements from a set comprised of $e$ elements. 
	
	\item Finally, observe that each location $\vec{x} \in \left\{  \vec{v}_{4i-2} \ \left| \ 1 \leq i \leq \frac{e+1}{2}  \right. \right\}$ is associated with some glue $g_{\vec{x}}$ in $M$. Such an association is represented by the pair $\left(\vec{x}, g_{\vec{x}}\right)$.
	Assume that $\left\{ \left. \vec{v}_{4i-2} + \vec{\Delta} \ \right| \ i = 1, \ldots, \frac{e+1}{2}  \right\} = \left\{ \vec{x}_1, \cdots, \vec{x}_{\frac{e+1}{2}} \right\}$, where the $\vec{x}$ locations are listed in lexicographical order. 
	In this last step, we count the number of ways to choose the sequence $\left( g_{\vec{x}_i} \left| \  i =1,\ldots,\frac{e+1}{2}  \right.  \right)$.
	Since the sequence is comprised of $\frac{e+1}{2}$ glues, and each glue can be assigned in one of $|G|$ possible ways, there are $|G|^{\frac{e+1}{2}}$ ways to choose the sequence.

\end{enumerate}

By the above counting procedure, for all $e = 1, \ldots, 2k - 1$, if $M_{\vec{\alpha},w} \upharpoonright s = \left(\vec{v}_1,g_1\right), \ldots, \left(\vec{v}_{2e},g_{2e}\right)$, then the number of ways to choose the sets $\left\{ \vec{v}_1, \ldots, \vec{v}_{2e} \right\}$ and $\left\{ \left. \vec{v}_{4i-2} \ \right| \ 1 \leq i \leq \frac{e+1}{2}  \right\}$ and the sequence $\left( g_{\vec{x}_i}  \left| \  i =1,\ldots,\frac{e+1}{2}  \right.  \right)$ is less than or equal to $\displaystyle \sum_{\substack{1\leq e < 2k \\ e\ \mathrm{odd}}}\left( \binom{2k}{e} \binom{e}{(e+1)/2}   |G|^{\frac{e+1}{2}} \right)$. Then, we have

\begin{equation*}
  \begin{aligned}
   \sum_{\substack{1\leq e < 2k \\ e\ \mathrm{odd}}}\left( \binom{2k}{e} \binom{e}{(e+1)/2}   |G|^{\frac{e+1}{2}} \right)  & \leq \sum_{\substack{1\leq e < 2k \\ e\ \mathrm{odd}}}\left( \binom{2k}{k} \binom{2k}{k}   |G|^{\frac{2k-1+1}{2}} \right) \\
    & \leq |G|^k \ \sum_{\substack{1\leq e < 2k \\ e\ \mathrm{odd}}} \left(2^{2k} \right)^2 =  |G|^k \ \sum_{\substack{1\leq e < 2k \\ e\ \mathrm{odd}}} 2^{4k} = |G|^k \cdot k \cdot 16^{k} .  \\
   \end{aligned}
\end{equation*}

Thus, if $ m > |G|^k \cdot k \cdot 16^{k}$, then there are two numbers $1 \leq l < l' \leq m$, such that, for $e = e_l = e_{l'}$, $M = M_{\vec{\alpha},w_l} \upharpoonright s = \left(\vec{v}_1,g_1\right), \ldots, \left(\vec{v}_{2e}, g_{2e}\right)$ and $M' = M_{\vec{\alpha},w_{l'}} \upharpoonright s = \left(\vec{v}'_1,g'_1\right),\ldots, \left(\vec{v}'_{2e},g'_{2e}\right)$ are non-empty restricted glue window submovies 
satisfying the following conditions:
\begin{enumerate}
	\item $\left\{ \left. \vec{v}_i + \vec{\Delta}_{l,l'} \ \right| \ 1 \leq i \leq 2e \right\} = \left\{ \left. \vec{v}'_j \ \right| \ 1 \leq j \leq 2e  \right\}$, and
	
	\item $\left\{ \left. \vec{v}_{4i-2} + \vec{\Delta}_{l,l'} \ \right| \ 1 \leq i \leq \frac{e+1}{2}  \right\} = \left\{ \vec{v}'_{4j-2} \ \left| \ 1 \leq j \leq \frac{e+1}{2} \right. \right\}$ and
	
	\item for all $1 \leq i,j \leq \frac{e+1}{2}$, if $\vec{v}'_{4j-2} = \vec{v}_{4i-2}+\vec{\Delta}_{l,l'}$, then $g'_{4j-2} = g_{4i-2}$.
\end{enumerate}

Note that, since $M$ and $M'$ are both restricted to $s$, we have, for all $1 \leq i \leq \frac{e+1}{2}$, $g_{4i-2} = g_{4i-3}$. 
This means that, for all $1 \leq i, j \leq \frac{e+1}{2}$, if $\vec{v}'_{4j-2} = \vec{v}_{4i-2} + \vec{\Delta}_{l,l'}$, then $g'_{4j-2} = g_{4i-3}$, and it follows that $M$ and $M'$ are sufficiently similar.

\end{proof}

%
%
%
%
\begin{lemma*}
Assume:
	 $\mathcal{T}$ is a directed, 3D TAS, 
	 $k, N \in \mathbb{Z}^+$, 
	 $s \subseteq R^3_{k,N}$ is a simple path from the location of the seed of $\mathcal{T}$ to some location in the furthest extreme column of $R^3_{k,N}$,
	 $\vec{\alpha}$ is a $\mathcal{T}$-assembly sequence that follows $s$, 
	 $w$ and $w'$ are windows, such that, $\vec{\Delta} \ne \vec{0}$ is a vector satisfying $w' = w + \vec{\Delta}$, and 
	 $e$ is an odd number satisfying $1 \leq e < 2k$. 
%
%
%
%
%
%
If $M = M_{\vec{\alpha},w} \upharpoonright s =  \left(\vec{v}_1,g_1\right),\ldots,\left(\vec{v}_{2e},g_{2e}\right)$ and $M' = M_{\vec{\alpha},w'} \upharpoonright s = \left(\vec{v}'_1,g'_1\right),\ldots,\left(\vec{v}'_{2e},g'_{2e}\right)$ are sufficiently similar non-empty restricted glue window submovies, then $R^3_{k,N}$ does not self-assemble in $\mathcal{T}$. 

\end{lemma*}

\begin{figure}[h!]
\input{lemma-2-algorithm}
\caption{The algorithm for $\vec{\beta}$. Here, the variable ``$k$'' has no relation to the ``$k$'' used in $R^3_{k,N}$.}
\end{figure}

\begin{proof}[Overview of our proof of Lemma~\ref{lem:lemma-2}] Our
  correctness proof for algorithm $\vec{\beta}$ breaks down into sub-proofs
  \#1, \#2, and \#3 that show that all of the tile placement steps
  performed by Loops~1,~2, and~3, respectively, are 
  adjacently correct. The first and third sub-proofs are relatively
  straightforward since each one of Loops~1 and~3 places tiles on only
  one side of $w'$ while mimicking a prefix or suffix of
  $\vec{\alpha}$, respectively. Sub-proof \#2 makes up the bulk of our
  correctness proof because Loop~2 contains two nested loops that
  alternate placing tiles on either side of $w'$. To prove the
  correctness of Loop~2, we will define a 6-part invariant for it, and
  prove, in turn, the initialization, maintenance, and termination
  properties of this invariant. Establishing the initialization
  property will be straightforward. For the maintenance property, we
  will first prove that the first four parts of the invariant still
  hold at the end of Loop~2a (i.e., on Line~11). Second, we will prove
  that the first four parts of the invariant still hold at the end of
  Loop~2b (i.e., on Line~16). Third, we will complete the maintenance
  proof by showing that the last two parts of the invariant also hold
  right after Line~16 is executed. Finally, we will wrap up sub-proof \#2
  with a proof of the termination property of the invariant.

We now define some notation needed to state our Loop~2 invariant. If and when
the algorithm enters  Loop~2, let $m$ be an integer such that $1 \leq m \leq
\frac{e+1}{2}$.
The variable $m$ will count the iterations of Loop~2. 
For $1 \leq l \leq m$, define $j_l$ to be the value of $j$ prior to
iteration $l$ of Loop~2.
Likewise, for $1 \leq l \leq m$, define $j'_l$ to be the value of $j'$ after Line~11 executes during
iteration $l$ of Loop~2.
We say that $j_m$ is the value of $j$ in the algorithm for $\vec{\beta}$ prior to the \emph{current}
iteration of Loop~2.
When it is clear from the context,  we will simply use ``$j$'' in place of ``$j_{m}$'' and ``$j'$'' in place of ``$j'_m$''. We define the following {\bf Loop~2 invariant}:\medskip

\fbox{\begin{minipage}{0.92\linewidth}
Prior to each iteration $m$ of Loop~2:
\begin{enumerate}
  \item \label{invariant-1} all previous tile placement steps executed
    by the algorithm for $\vec{\beta}$ are adjacently correct,
  \item \label{invariant-2} all tiles placed by $\vec{\beta}$ on
    locations on the far side of $w'$ from the seed are placed by tile
    placement steps executed by Loop~2a,
  \item \label{invariant-3} all tiles placed by $\vec{\beta}$ on
    locations on the near side of $w'$ to the seed are placed by tile
    placement steps executed by Loop~1 or Loop~2b,
  \item \label{invariant-4} if $m > 1$, then for all $1 \leq l < m$,
    $j_l \ne j_m$,
  \item \label{invariant-5} the location at which $\vec{\beta}$ last
    placed a tile (say $t$) is $\vec{v}'_{4j_m-3}$, and
  \item \label{invariant-6} the glue of $t$ that touches $w'$ is $g'_{4j_m-3}$.
\end{enumerate}
  \end{minipage}}
\bigskip

\end{proof}

\begin{proof}

  $ $ 

\noindent  \underline{Sub-proof \#1:}
  
Since $\vec{\alpha}$ is a $\mathcal{T}$-assembly sequence that follows
a simple path, the tile placement steps in Loop~1 are adjacently
correct and only place tiles that are on the near side of $w'$ to the
seed.
Note that Loop~1 terminates with $Pos\left( \vec{\alpha}\left[ k
  \right] \right) = \vec{v}'_{2}$.
By the definition of $M'$, $\vec{v}'_2$ is the first location at which
$\vec{\alpha}$ places a tile on the far side of $w'$ from the seed.

\noindent \underline{Sub-proof \#2 - Loop~2 invariant initialization}

Just before the first iteration of Loop~2, $m=1$ and all prior tile
placements have been completed within Loop~1. Parts~\ref{invariant-1}
and~\ref{invariant-3} of the invariant follow directly from Sub-proof
\#1.  Part~\ref{invariant-2} of the invariant is true since no tiles
have been placed yet on the far side of $w'$. Part~\ref{invariant-4}
of the invariant holds since $m=1$. Part~\ref{invariant-5} of the
invariant holds because $j_m = j_1 = 1$, $\vec{v}'_{4j_m-3} =
\vec{v}'_1$, and the location at which $\vec{\beta}$ last placed a
tile $t$ is $\vec{v}'_1$, that is, the location that precedes
$\vec{v}'_2$ in $\vec{\alpha}$. Finally, part~\ref{invariant-6} of the
invariant holds because $t$ is the tile that $\vec{\alpha}$ placed at
$\vec{v}'_1$ and, by the definition of $M'$, the glue of $t$ that
touches $w'$ is $g'_{4j_m-3}=g'_1$.

\noindent \underline{Sub-proof \#2 - Loop~2 invariant maintenance}

On Line~5, if $j$ is such that $\vec{v}'_{4j-2} = \vec{v}_{2e}+\vec{\Delta}$,
then Loop~2 terminates.
So, let $j$ be such that $\vec{v}'_{4j-2} \ne
\vec{v}_{2e}+\vec{\Delta}$ and assume that the Loop~2 invariant holds.
We will first prove (by induction) that parts~\ref{invariant-1}
through~\ref{invariant-4} of the invariant still hold when Loop~2a
terminates.

Within the current iteration of Loop~2, Line~7 sets $k$ to a value
such that $Pos\left( \vec{\alpha}\left[ k \right] \right) =
\vec{v}_{4i-2}$, where $i$ is such that $4i-2$ is the index of
$\vec{v}'_{4j-2}-\vec{\Delta}$ in $M$.
For the base step of the induction, consider the tile placement step
executed in the first iteration of Loop~2a. To establish the
  first part of the invariant, we now prove that this tile placement
  step is adjacently correct. First, we prove that it places a tile
  that binds to the last tile placed by the algorithm.
Intuitively, this tile placement step is where $\vec{\beta}$ finishes
crossing from the near side of $w'$ to the seed over to the far
side. 
Formally, we have:
\begin{eqnarray*} 
\vec{\beta} & = & \vec{\beta} + \left(  \vec{\alpha}\left[ k \right] +\vec{\Delta} \right) \\
                    & =  & \vec{\beta} + \left( \left(   Pos\left( \vec{\alpha}\left[ k \right] \right) + \vec{\Delta} \right) \mapsto Tile\left(  \vec{\alpha}\left[ k \right] \right)\right)  \\
                    & = & \vec{\beta} + \left( \left(    \vec{v}_{4i-2}  + \vec{\Delta} \right) \mapsto Tile\left(  \vec{\alpha}\left[ k \right] \right)\right)  \\
                     & = & \vec{\beta} +   \left( \left( \vec{v}'_{4j-2} - \vec{\Delta} + \vec{\Delta} \right)  \mapsto Tile\left(  \vec{\alpha}\left[ k \right] \right) \right) \\
                     & = & \vec{\beta} +  \left( \vec{v}'_{4j-2}  \mapsto Tile\left(  \vec{\alpha}\left[ k \right] \right) \right),
\end{eqnarray*}
where the second-to-last equality follows from Line~6 in the algorithm
for $\vec{\beta}$. This, together with part~\ref{invariant-5}
  of the invariant, shows that the location of this tile placement
  step is adjacent to, and on the opposite side of $w'$ from,
  $\vec{v}'_{4j-3}$.
We now prove that the tile $t = Tile\left(\vec{\alpha}[k]\right)$ this
step places at $\vec{v}'_{4j-2}$ does bind to the tile $t'$ that the
algorithm just placed at $\vec{v}'_{4j-3}$. By part~\ref{invariant-6}
of the invariant, the glue of $t'$ that touches $w'$ is $g'_{4j-3}$,
which, according to the follow reasoning, must be equal to the glue of
$t$ that touches $w'$.
\begin{itemize}
\item Since $\vec{\alpha}$ follows the simple path $s$ and $M'$ is restricted to $s$, $g'_{4j-3} = g'_{4j-2}$.
\item By part~\ref{lemma-assumption-4} of sufficiently similar, $g'_{4j-2} = g_{4i-3}$.
\item Since $\vec{\alpha}$ follows a simple path and $M$ is restricted to $s$, $g_{4i-3} = g_{4i-2}$.
\item Since $t$ is the type of tile that $\vec{\alpha}$ placed at
  $\vec{v}_{4i-2}$ and the glue of $t$ that touches $w$ is $g_{4i-2}$,
  the previous chain of equalities imply that the glue of $t'$
  that touches $w'$ is equal to the glue of $t$ that touches $w'$.
\end{itemize}

%
Having shown that $t$ binds to $t'$, we now prove that $\vec{\beta}$
has not already placed a tile at $Pos\left(\vec{\alpha}[k]\right) +
\vec{\Delta} = \vec{v}'_{4j-2}$ before the tile placement step in the
first iteration of Loop~2a is executed.

According to part~\ref{invariant-2} of the invariant, all locations on the
far side of $w'$ from the seed at which tiles are placed by
$\vec{\beta}$ are filled by tile placement steps executed by Loop~2a.
Since $\vec{v}'_{4j-2}$ is on the far side of $w'$ from the seed, we
only need to consider tile placement steps that place tiles at
locations that are on the far side of $w'$ from the seed.
Since we are assuming that $\vec{\beta} = \vec{\beta} +
\left( \vec{\alpha}\left[ k \right]+ \vec{\Delta} \right)$ is the tile
placement step executed in the first iteration of Loop~2a, we know
that any already completed tile placement step $\vec{\beta} = \vec{\beta} +
\left( \vec{\alpha}\left[ k' \right]+ \vec{\Delta} \right)$, for $0 \leq k'  < \left| \vec{\alpha} \right|$, is
executed in some iteration of Loop~2a but in a
past iteration of Loop~2.
Define $index_{\vec{\alpha}}\left( \vec{x} \right)$ to be the
value of $n$ such that $Pos\left(\vec{\alpha}[n]\right) = \vec{x}$.
Define the rule $f(j) = i$ such that $4i - 2$ is the index of
$\vec{v}'_{4j-2}-\vec{\Delta}$ in $M$.
Note that $f$ is a valid function because, by part~\ref{lemma-assumption-3} of sufficiently similar, we have $\left\{ \vec{v}_{4i-2} + \vec{\Delta} \ \left| \ 1
\leq i \leq \frac{e+1}{2} \right. \right\} = \left\{ \vec{v}'_{4j-2}
\ \left| \ 1 \leq j \leq \frac{e+1}{2} \right. \right\}$.
Moreover, $f$ is injective, because, intuitively, two different
locations in $M'$ cannot translate with the same $\vec{\Delta}$ to the
same location in $M$.
Formally, assume that $f(a) = f(b)$ and let $c$ be such that $4c-2$ is the index of $\vec{v}'_{4a-2} - \vec{\Delta}$ in $M$ and let $d$ be such that $4d-2$ is the index of $\vec{v}'_{4b-2} - \vec{\Delta}$ in $M$. 
Since we are assuming $f(a) = f(b)$, then we have $c=d$. 
This means that $4c-2 = 4d-2$ is the index of $\vec{v}'_{4a-2} - \vec{\Delta}$ in $M$. 
Likewise, $4c-2 = 4d-2$ is the index of $\vec{v}'_{4b-2} - \vec{\Delta}$ in $M$. 
Then we have $\vec{v}'_{4a-2} = \vec{v}_{4c-2} + \vec{\Delta}$, and $\vec{v}_{4c-2} + \vec{\Delta} = \vec{v}'_{4b-2}$.
In other words, we have $\vec{v}'_{4a-2} = \vec{v}'_{4b-2}$, which implies that $a=b$ and it follows that $f$ is injective.
For all $1 \leq l \leq m$, define $i_l$ to be the value of $i$
computed in Line~6.
In other words, $i_l = f\left(j_l\right)$ and $i_m$ is the value of $i$ computed in Line~6 during
the current iteration of Loop~2.
Observe that $k$ (on Line~7) satisfies

\begin{equation}\label{loop2a-inequality-on-k}
  index_{\vec{\alpha}}\left(\vec{v}_{4i_{m}-2}\right) \leq k < index_{\vec{\alpha}}\left( \vec{v}_{4i_{m}} \right)
\end{equation}

\noindent because $\vec{\beta} = \vec{\beta} +
\left( \vec{\alpha}\left[ k \right]+ \vec{\Delta} \right)$ is the tile placement step executed in the first iteration of Loop~2a, and, for some $1 \leq l < m$, $k'$ satisfies

\begin{equation}\label{loop2a-inequality-on-kprime}
  index_{\vec{\alpha}}\left(\vec{v}_{4i_{l}-2}\right) \leq k' < index_{\vec{\alpha}}\left( \vec{v}_{4i_{l}} \right).
\end{equation}

\noindent because $\vec{\beta} = \vec{\beta} +
\left( \vec{\alpha}\left[ k' \right]+ \vec{\Delta} \right)$ is some tile placement step executed in Loop~2a but in a past iteration of Loop~2.
In fact, in the first iteration of Loop~2a, $k = index_{\vec{\alpha}}\left(\vec{v}_{4i_m-2}\right) < index_{\vec{\alpha}}\left(\vec{v}_{4i_m}\right)$. 
By part~\ref{lemma-assumption-4} of the invariant, for all $1 \leq l < m$, $j_m \ne j_l$. 
Since $f$ is injective, it follows that, for all $1 \leq l < m$, $i_m = f\left( j_m\right) \ne f\left( j_l \right) = i_l$.
Then we have three cases to consider. 
Case 1, where $\vec{v}_{4i_m} = \vec{v}_{4i_l-2}$, is impossible,
since these two locations are on opposite sides of $w$.
In case 2, where $index_{\vec{\alpha}}(\vec{v}_{4i_m}) <
index_{\vec{\alpha}}(\vec{v}_{4i_l-2})$, we have:

\centerline{$index_{\vec{\alpha}}\left(\vec{v}_{4i_m-2}\right) \leq k < index_{\vec{\alpha}}\left( \vec{v}_{4i_m} \right) < index_{\vec{\alpha}}\left(\vec{v}_{4i_{l}-2}\right) \leq k' < index_{\vec{\alpha}}\left( \vec{v}_{4i_{l}} \right)$.}

\noindent Finally, in case 3, where $index_{\vec{\alpha}}(\vec{v}_{4i_m}) > index_{\vec{\alpha}}(\vec{v}_{4i_l-2})$, we have:	

\centerline{$index_{\vec{\alpha}}\left(\vec{v}_{4i_{l}-2}\right) \leq k' < index_{\vec{\alpha}}\left( \vec{v}_{4i_{l}} \right) < index_{\vec{\alpha}}\left(\vec{v}_{4i_m-2}\right) \leq k < index_{\vec{\alpha}}\left( \vec{v}_{4i_m} \right)$.}

\noindent In all possible cases, $k \ne k'$. 
Thus, since $\vec{\alpha}$ follows a simple path, $Pos\left( \vec{\alpha}[k] \right)
\ne Pos\left( \vec{\alpha}[k'] \right)$, which implies $Pos\left(
\vec{\alpha}[k] \right) + \vec{\Delta} \ne Pos\left( \vec{\alpha}[k']
\right) + \vec{\Delta}$. 
Therefore, $Pos\left(\vec{\alpha}[k]\right) + \vec{\Delta}$ is empty prior to the execution of $\vec{\beta} = \vec{\beta} +
\left( \vec{\alpha}\left[ k \right]+ \vec{\Delta} \right)$, i.e., no previous tile placement step placed a tile at that location before the first iteration of Loop~2a.
This means that the tile placement step $\vec{\beta} = \vec{\beta} +
\left( \vec{\alpha}\left[ k \right] + \vec{\Delta} \right)$ is
adjacently correct.
This concludes the proof of correctness for the first iteration of
Loop~2a (base step).

We now show (inductive step) that the rest of the tile placement steps
executed in Loop~2a within the current iteration of Loop~2 are adjacently
correct.
Let $\vec{\beta} = \vec{\beta} + \left( \vec{\alpha}[k] +
\vec{\Delta} \right)$ be a tile placement step executed in some (but
not the first) iteration of Loop~2a and assume that all
tile placement steps executed in past iterations of Loop~2a are
adjacently correct and place tiles at locations that are on the far side of $w'$
from the seed (inductive hypothesis).
In particular, assume that the tile placement step $\vec{\beta} =
\vec{\beta} + \left( \vec{\alpha}[k-1] + \vec{\Delta} \right)$
executed in Loop~2a, for the current iteration
of Loop~2, is adjacently correct and places a tile at a
location that is on the far side of $w'$ from the seed.
Since $\vec{\alpha}$ follows a simple path,
$Pos\left(\vec{\alpha}[k]\right)$ is adjacent to
$Pos\left(\vec{\alpha}[k-1]\right)$ and the configuration consisting
of a tile of type $Tile\left(\vec{\alpha}[k]\right)$ placed at
$Pos\left(\vec{\alpha}[k]\right)$ and a tile of type
$Tile\left(\vec{\alpha}[k-1]\right)$ placed at
$Pos\left(\vec{\alpha}[k-1]\right)$ is stable.
This means that $Pos\left(\vec{\alpha}[k]\right) + \vec{\Delta}$ is
adjacent to $Pos\left(\vec{\alpha}[k-1]\right) + \vec{\Delta}$ and the
configuration consisting of a tile of type
$Tile\left(\vec{\alpha}[k]\right)$ placed at
$Pos\left(\vec{\alpha}[k]\right) + \vec{\Delta}$ and a tile of type
$Tile\left(\vec{\alpha}[k-1]\right)$ placed at
$Pos\left(\vec{\alpha}[k-1]\right) + \vec{\Delta}$ is stable, thus proving part~\ref{adjacent-valid-binding} of adjacently correct.
Now, when proving part~\ref{adjacent-valid-empty-location}, two cases arise.
The first case is where $\vec{\beta} = \vec{\beta} + \left( \vec{\alpha}\left[ k' \right] +
\vec{\Delta} \right)$ is executed in a past iteration of Loop~2a in the
current iteration of Loop~2.
Here, we have $k \ne k'$ because, within Loop~2a, we are merely
translating a segment of $\vec{\alpha}$, which follows a simple path.
The second case is where $\vec{\beta} = \vec{\beta} + \left( \vec{\alpha}\left[
  k' \right] + \vec{\Delta} \right)$ is executed in a past iteration
of Loop~2.
Here, using reasoning that is similar to the one we used to establish
the correctness of the first iteration of Loop~2a based on
inequalities~(\ref{loop2a-inequality-on-k}) and~(\ref{loop2a-inequality-on-kprime}) above, we have $k \ne k'$.
In both cases, $k \ne k'$ implies $Pos\left( \vec{\alpha}[k] \right)
\ne Pos\left( \vec{\alpha}[k'] \right)$, which means that the location
of the current tile placement step in Loop~2a is different from the
location of any previous tile placement step that was executed in
Loop~2a.  
It follows that $Pos\left( \vec{\alpha}[k] \right) +
\vec{\Delta} \ne Pos\left( \vec{\alpha}[k'] \right) + \vec{\Delta}$.
This proves part~\ref{adjacent-valid-empty-location}, and therefore,
the tile placement step
$\vec{\beta} = \vec{\beta} + \left( \vec{\alpha}\left[ k \right] +
\vec{\Delta} \right)$ is adjacently correct. 
This concludes our proof that
part~\ref{invariant-1} of the invariant holds at the end of Loop~2a.

Since Loop~2a mimics the portion of $\vec{\alpha}$
between (and including) the points $\vec{v}_{4i_m-2}$ and
$\vec{v}_{4i_m-1}$ which, by definition of $M$, is on the far side of $w$ from
the seed, it follows that $Pos\left(\vec{\alpha}[k]\right)$ is on the
far side of $w$ from the seed during every iteration of Loop~2a.
This means that $Pos\left(\vec{\alpha}[k]\right) + \vec{\Delta}$ is on
the far side of $w'$ from the seed during every iteration of Loop~2a
and thus part~\ref{invariant-2} of the invariant holds at the end of
Loop~2a. 
For the same reason, part~\ref{invariant-3} of the invariant
also holds at that point. 
Finally, part~\ref{invariant-4} of the
invariant trivially holds since Loop~2a does not update $j$.
This concludes our proof that the first four parts of the invariant
hold when Loop~2a terminates. 
We will now prove that these four parts still
hold when Loop~2b terminates.

Loop~2b ``picks up'' where Loop~2a ``left off''.
Note that  Loop~2a terminates with $Pos\left(\vec{\alpha}[k]\right) =
\vec{v}_{4i_m}$, with the last tile being placed at
$\vec{v}_{4i_m-1}+\vec{\Delta}$.
Define the rule $g(i) = j$ such that $4j$ is the index of
$\vec{v}_{4i}+\vec{\Delta}$ in $M'$.
Note that $g$, like $f$, is a valid function because, by parts~\ref{lemma-assumption-2}~\&~\ref{lemma-assumption-3} of sufficiently similar, we have $\left\{ \vec{v}_{4i} + \vec{\Delta}
\ \left| \ 1 \leq i \leq \frac{e-1}{2} \right. \right\} = \left\{
\vec{v}'_{4j} \ \left| \ 1 \leq j \leq \frac{e-1}{2}
\right. \right\}$.
Similarly, $g$, like $f$, is injective.
Line~11 sets the value of $j'_m$ to be such that $4j'_m$ is the index of
$\vec{v}_{4i_m}+\vec{\Delta}$ in $M'$.
In other words, Line~11 computes $j'_m = g\left( i_m \right)$ and
Line~12 sets the value of $k$ such that $Pos\left( \vec{\alpha}\left[
  k \right] \right) = \vec{v}'_{4j'_m}$.
Intuitively, $\vec{v}_{4i_m}+\vec{\Delta} = \vec{v}'_{4j'_m}$ is the
location at which $\vec{\beta}$ finishes crossing from the far side of
$w'$ from the seed back to the near side.
Recall that Loop~2a ``left off'' by placing a
tile (in its last iteration) at the location $\vec{v}_{4i_m-1} +
\vec{\Delta} = \vec{v}'_{4j'_m-1}$.

Now, for the base step of the induction we use to prove that
part~\ref{invariant-1} of the invariant holds after Loop~2b, consider the tile
placement step executed in the first iteration of Loop~2b.  
Formally, we have:
\begin{eqnarray*}
	\vec{\beta} & = & \vec{\beta} + \vec{\alpha}\left[ k \right] \\
		            & = & \vec{\beta} + \left( Pos\left( \vec{\alpha}\left[ k \right]  \right) \mapsto Tile\left( \vec{\alpha}\left[ k \right] \right) \right) \\
		            & = & \vec{\beta} + \left( \vec{v}'_{4j'_m} \mapsto Tile\left(\vec{\alpha}\left[ k \right] \right) \right).
\end{eqnarray*}
Thus, the tile placement step executed in the first iteration of
Loop~2b will place a tile at $\vec{v}'_{4j'_m}$,
which, by the definition of $M'$, is adjacent to but on the opposite
side of $w'$ from $\vec{v}'_{4j'_m-1}$.
Since $\mathcal{T}$ is directed, the type of tile that $\vec{\beta}$
places at $\vec{v}'_{4j'_m-1}$ during the final iteration of Loop~2a
 must be the same as the type of the tile that
$\vec{\alpha}$ places at $\vec{v}'_{4j'_m-1}$.
This is the only place in the proof where we use the fact that
$\mathcal{T}$ is directed.
By the definition of the tile placement step executed in the first
iteration of Loop~2b, the type of tile that
$\vec{\beta}$ places at $\vec{v}'_{4j'_m}$ is the same as the type of tile that
$\vec{\alpha}$ places at $\vec{v}'_{4j'_m}$.
This means that the glue of the tile that $\vec{\beta}$ places at
$\vec{v}'_{4j'_m}$ and that touches $w'$ is equal to the glue of the tile
that $\vec{\beta}$ places at $\vec{v}'_{4j'_m-1}$ and that touches $w'$.
This proves part~\ref{adjacent-valid-binding} of adjacently correct for $\vec{\beta} = \vec{\beta} + \vec{\alpha}[k]$. 
So, in order to show that $\vec{\beta} = \vec{\beta} + \vec{\alpha}[k]
$ is adjacently correct, it suffices to show that $\vec{\beta}$ has not already
placed a tile at $Pos\left(\vec{\alpha}[k]\right)$, i.e., part~\ref{adjacent-valid-empty-location} of adjacently correct.

By part~\ref{invariant-3} of the invariant, all tiles placed by
$\vec{\beta}$ on the near side of $w'$ to the seed result from tile
placement steps belonging to either Loop~1 or Loop~2b.
Since $\vec{v}'_{4j'_m}$ is on the near side of $w'$ to the seed, we
only need to consider tile placement steps in the algorithm for
$\vec{\beta}$ that place tiles at locations that are on the near side
of $w'$ to the seed.
Since we are assuming that $\vec{\beta} = \vec{\beta} +
\vec{\alpha}\left[ k \right] $ is the tile placement step executed in
the first iteration of Loop~2b, we must consider two cases for any
already completed tile placement step $\vec{\beta} = \vec{\beta} +
\vec{\alpha}\left[ k' \right]$ with $0
\leq k' < \left| \vec{\alpha} \right|$.
In the case where $\vec{\beta} = \vec{\beta} + \vec{\alpha}\left[ k'
  \right]$ is executed in some iteration of Loop~1 (before the first
iteration of Loop~2), we have $k' < index_{\vec{\alpha}}\left(
\vec{v}'_{2} \right)$ and
$index_{\vec{\alpha}}\left(\vec{v}'_{4}\right) \leq k$.
In this case, $index_{\vec{\alpha}}\left( \vec{v}'_{2} \right) <
index_{\vec{\alpha}}\left(\vec{v}'_{4}\right)$ implies $k' \ne k$.
In the second case, namely when $\vec{\beta} = \vec{\beta} +
\vec{\alpha}\left[ k' \right]$ is executed in some iteration of
Loop~2b but in a past iteration of Loop~2, $k$ satisfies

\begin{equation}\label{loop2b-inequality-on-k}
index_{\vec{\alpha}}\left(\vec{v}'_{4
  j'_{m}}\right) \leq k <
index_{\vec{\alpha}}\left(\vec{v}'_{4j'_{m}+2}\right)
\end{equation}

\noindent because $\vec{\beta} = \vec{\beta} +
\vec{\alpha}\left[ k \right]$ is the tile placement step executed in the first iteration of Loop~2b and, for some $1 \leq l < m$, $k'$ satisfies

\begin{equation}\label{loop2b-inequality-on-kprime}
index_{\vec{\alpha}}\left(\vec{v}'_{4j'_{l}}\right)
 \leq k' < index_{\vec{\alpha}}\left(\vec{v}'_{4j'_{l}+2}\right).
\end{equation}
In order to show that $k \ne k'$, since $\vec{\alpha}$ follows a simple path, it suffices to show that, for all $1
\leq l < m$, $j'_{m} \ne j'_l$.
By part~\ref{invariant-4} of the invariant, for all $1 \leq l < m$, $j_m \ne j_l$. 
By definition, for all $1 \leq l \leq m$, $i_l = f\left( j_l \right)$.
Since $f$ is injective, we have, for all $1 \leq l < m$, $i_m \ne i_l$. 
Since $g$ is injective, we have, for all $1 \leq l < m$, $g\left( i_m
\right) \ne g\left( i_l \right)$.
By definition, for all $1 \leq l \leq m$, $j'_l = g\left( i_l \right)$. 
Then, we have, for all $1 \leq l < m$, $j'_{m} = g\left(
i_{m} \right) \ne g\left( i_l \right) = j'_l$.
So, in all cases, we have $k\ne k'$, which implies that $Pos\left(
\vec{\alpha}[k] \right) \ne Pos\left( \vec{\alpha}[k'] \right)$.
This means that part~\ref{adjacent-valid-empty-location} is satisfied and therefore the tile placement step $\vec{\beta} = \vec{\beta} +
\vec{\alpha}\left[ k \right]$ is adjacently correct.
This concludes the proof of correctness for the first iteration of
Loop~2b (base step).

We now show (inductive step) that the rest of the tile placement steps
executed in Loop~2b and within the current iteration of Loop~2 are adjacently correct.
So, let $\vec{\beta} = \vec{\beta} + \vec{\alpha}[k]$ be a tile
placement step executed in some (but not the first) iteration of 
Loop~2b and assume all tile placement steps
executed in past iterations of Loop~2b are adjacently correct and place tiles at locations that are on the near side of $w'$
to the seed.
In particular, the tile placement step $\vec{\beta} = \vec{\beta} +
\vec{\alpha}[k-1] $ executed in Loop~2b, for
the current iteration of Loop~2, is adjacently correct
and places a tile at a location that is on the near side of $w'$ to
the seed.
Since $\vec{\alpha}$ follows a simple path,
$Pos\left(\vec{\alpha}[k]\right)$ is adjacent to
$Pos\left(\vec{\alpha}[k-1]\right)$ and the configuration consisting
of a tile of type $Tile\left(\vec{\alpha}[k]\right)$ placed at
$Pos\left(\vec{\alpha}[k]\right)$ and a tile of type
$Tile\left(\vec{\alpha}[k-1]\right)$ placed at
$Pos\left(\vec{\alpha}[k-1]\right)$ is stable, thus proving part~\ref{adjacent-valid-binding} of adjacently correct.
Here, using reasoning that is similar to the one we used to establish
the correctness of the first iteration of Loop~2b based on
inequalities~(\ref{loop2b-inequality-on-k}) and~(\ref{loop2b-inequality-on-kprime}) above, we have $k \ne k'$.
This means that $Pos\left( \vec{\alpha}[k] \right) \ne Pos\left(
\vec{\alpha}[k'] \right)$, thereby satisfying
part~\ref{adjacent-valid-empty-location}.
It follows that the tile placement step $\vec{\beta} = \vec{\beta} +
\vec{\alpha}\left[ k \right]$ is adjacently correct.
This concludes our proof that part~\ref{invariant-1} of the invariant
holds when Loop~2b terminates.

Since $\vec{\alpha}$ follows a simple path and the portion of
$\vec{\alpha}$ between (and including) $\vec{v}'_{4j'_m}$ and
$\vec{v}'_{4j'_m+1}$ is, by definition of $M'$, on the near side of
$w'$ to the seed, $Pos\left(\vec{\alpha}[k]\right)$ is on the near
side of $w'$ to the seed during every iteration of Loop~2b.
This means that all tile placement steps executed by Loop~2b only
place tiles at locations on the near side of $w'$ to the seed. 
This
concludes our proof that parts~\ref{invariant-2} and~\ref{invariant-3}
of the invariant hold when Loop~2b terminates.

We will now show that, for all $1
\leq l \leq m$, $j_{m+1} \ne j_l$.
We already showed above that, for all $1
\leq l < m$, $j'_{m} \ne j'_l$.
Then we have, for all $1 \leq l < m$, $j'_{m} + 1 \ne j'_l + 1$.
Since Line~16 computes the value of $j$ for the next iteration of
Loop~2 to be the value of $j' + 1$, we infer, for all $1 \leq l < m$, $j_{m+1}
\ne j_{l+1}$, or, equivalently, for all $2 \leq l \leq m$, $j_{m+1} \ne
j_l$.
Since $j_1 = 1$ and $j'$ (computed on Line~11) cannot be equal to 0,
we have $j_{m+1} \ne j_1$.
It follows that, for all $1 \leq l \leq m$, $j_{m+1} \ne j_l$.
This concludes our proof that part~\ref{invariant-4} of the invariant
holds when Loop~2b terminates.

Note that Line~11 computes $j'_m = g\left( i_m \right)$ and Line~12
sets $k$ to a value such that $Pos\left(\vec{\alpha}[k]\right) =
\vec{v}'_{4j'_m}$.
Subsequently, Loop~2b terminates with $Pos\left(\vec{\alpha}[k]\right) =
\vec{v}'_{4j'_{m}+2}$.
This means that the location of the tile placed during the last
iteration of Loop~2b, which is also the location at which $\vec{\beta}$ last
placed a tile during this iteration of Loop~2, and thus right before
the next iteration of Loop~2, is $\vec{v}'_{4j'_{m}+1}$.
Since Line~16 computes $j_{m+1} = j'_m + 1$, we have
$\vec{v}'_{4j'_{m}+1} = \vec{v}'_{4 \left( j_{m+1} - 1 \right) + 1} =
\vec{v}'_{4 j_{m+1} - 3}$. 
This concludes our proof that
part~\ref{invariant-5} of the invariant holds when Loop~2b terminates.

Let $t$ be the tile that $\vec{\beta}$ placed at location $\vec{v}'_{4j_{m+1} - 3}$. 
Since Loop~2b simply copies the portion of $\vec{\alpha}$
between (and including) the points $\vec{v}'_{4j'_m}$ and
$\vec{v}'_{4j'_m+1}$, the glue of $t$ that touches $w'$ is
$g'_{4j'_{m}+1} = g'_{4j_{m+1} - 3}$. 
This concludes our proof that
part~\ref{invariant-6} of the invariant holds when Loop~2b terminates.

In conclusion, all six parts of our invariant hold when Loop~2b
terminates. 
Since no tile placements are performed during the current
iteration of Loop~2 after Loop~2b terminates, the invariant holds
when iteration $m$ terminates and thus prior to iteration $m+1$ of Loop~2.
This concludes our maintenance proof for Loop~2.

\noindent \underline{Sub-proof \#2 - Loop~2 invariant termination}

Note that Loop~2 terminates when the location at which $\vec{\beta}$
will next place a tile is $\vec{v}_{2e}+\vec{\Delta}$.
By part~\ref{invariant-4} of the invariant, prior to each iteration of Loop~2, for all
$1 \leq l < m$, $j_m \ne j_l$.
Since $|M| = \left| M' \right| < \infty$, Loop~2 must eventually
terminate with $Pos\left(\vec{\alpha}[k]\right) =$ $\vec{v}_{2e}+\vec{\Delta}$.

\noindent \underline{Sub-proof \#3:}

The reasoning that we used to show that all of the tile placement steps
executed by Loop~2a are adjacently correct and only place
tiles at locations on the far side of $w'$ from the seed can be
adapted to show that all of the tile placement steps executed by Loop~3
 are adjacently correct and only place tiles at locations
on the far side of $w'$ from the seed.
Moreover, Loop~3 will terminate because
$\left|\vec{\alpha}\right| < \infty$.

Thus, every tile placement step executed by the algorithm for $\vec{\beta}$ is adjacently correct. 
Since $s \subseteq R^3_{k,N}$ is a path from the location of the seed
of $\mathcal{T}$ to some location in the furthest extreme column of
$R^3_{k,N}$ and $\vec{\Delta} \ne \vec{0}$, it follows that, during Loop~2a and/or Loop~3,
$\vec{\beta}$ places at least one tile at a location that is not in
$R^3_{k,N}$.
In other words, $R^3_{k,N}$ does not self-assemble in $\mathcal{T}$.
\end{proof}

The following result combines Lemmas~\ref{lem:lemma-1} and~\ref{lem:lemma-2}. 
%
%
%
%
\begin{lemma*}
\label{lem:lemma-3}
Assume:
	 $\mathcal{T}=(T,\sigma, 1)$ is a 3D TAS, 
	 $G$ is the set of all glues in $T$, 
	 $k,N \in \mathbb{Z}^+$,
	 $s \subseteq R^3_{k,N}$ is a simple path from the location of $\sigma$ to some location in the furthest extreme column of $R^3_{k,N}$,
	 $\vec{\alpha}$ is a $\mathcal{T}$-assembly sequence that follows $s$, 
	 $m \in \mathbb{Z}^+$, 
         for all $1 \leq l \leq m$, $w_l$ is a
          window,
	 for all $1 \leq l < l' \leq m$, $\vec{\Delta}_{l,l'} \ne \vec{0}$ satisfies $w_{l'} = w_l + \vec{\Delta}_{l,l'}$, and
	 for all $1 \leq l \leq m$, there is an odd $1 \leq e_l <
          2k$ such that $M_{\vec{\alpha}, w_l} \upharpoonright s$ is a
          non-empty restricted glue window submovie of length
            $2e_l$.
%
%
%
%
%
%
%
%
%
%
If $m > |G|^k \cdot k \cdot 16^{k}$, then $R^3_{k,N}$ does not self-assemble in $\mathcal{T}$.
\end{lemma*}

\begin{proof}
The hypothesis of Lemma~\ref{lem:lemma-1} is satisfied. So there exist
$1 \leq l < l' \leq m$ such that $e = e_l = e_{l'}$ and
$M_{\vec{\alpha}, w_l} \upharpoonright s = \left(\vec{v}_1,g_1\right),
\ldots, \left( \vec{v}_{2e}, g_{2e}\right)$ and $M_{\vec{\alpha},
  w_{l'}} \upharpoonright s = \left(\vec{v}'_1,g'_1\right), \ldots,
\left( \vec{v}'_{2e}, g'_{2e}\right)$ are
sufficiently similar non-empty restricted glue window submovies. Thus,
the hypothesis of Lemma~\ref{lem:lemma-2} is satisfied. It follows
that $R^3_{k,N}$ does not self-assemble in $\mathcal{T}$.

\end{proof}

\addtocounter{theorem}{-2}
%
%
%
%
\begin{theorem}
$K^1_{USA}\left(R^3_{k,N}\right)=\Omega\left( N^{\frac{1}{k}}\right)$.
\end{theorem}

\begin{proof}
Assume that $\mathcal{T} = (T,\sigma,1)$ is a directed, 3D TAS in
  which $R^3_{k,N}$ self-assembles, with $\alpha \in
  \mathcal{A}_{\Box}\left[ \mathcal{T} \right]$ satisfying
  $\textmd{dom } \alpha = R^3_{k,N}$.
Let $s$ be a simple path in $G^{\textmd{b}}_{\alpha}$ from the location of $\sigma$ (the seed) to some location in the furthest extreme (westernmost or easternmost) column of $R^3_{k,N}$ in either the $z=0$ or $z=1$ plane.  
%
By Observation~\ref{obs:simple}, there is an assembly sequence $\vec{\alpha}$ that follows $s$. 
Assume $N \geq 3$.
Since $s$ is a simple path from the location of the seed to some
location in the furthest extreme column of $R^3_{k,N}$, there is some
positive integer $m \geq \left \lfloor \frac{N}{2} \right \rfloor \geq
\frac{N}{3}$ such that, for all $1 \leq l \leq m$, $w_l$ is a window
that cuts $R^3_{k,N}$, for all $1 \leq l < l' \leq m$, there exists
$\vec{\Delta}_{l,l'} \ne \vec{0}$ satisfying $w_{l'} = w_l +
\vec{\Delta}_{l,l'}$, and for each $1 \leq l \leq m$, there exists a
corresponding odd number $1 \leq e_l < 2k$ such that
$M_{\vec{\alpha},w_l} \upharpoonright s$ is a
non-empty restricted glue window submovie of length $2e_l$.
Since $R^3_{k,N}$ self-assembles in $\mathcal{T}$, (the contrapositive of) Lemma~\ref{lem:lemma-3} says that $m \leq |G|^k \cdot k \cdot 16^k$. 
We also know that $\frac{N}{3} \leq m$, which means that $\frac{N}{3} \leq |G|^k \cdot k \cdot 16^k$. Thus, we have $N \leq 3 \cdot |G|^k \cdot k \cdot 16^k$ and it follows that $|T| \geq \frac{|G|}{6} \geq \frac{1}{6} \frac{N^{\frac{1}{k}}}{\left(3 \cdot k \cdot 16^k \right)^{\frac{1}{k}}} \geq \frac{1}{6} \frac{N^{\frac{1}{k}}}{\left(3^k \cdot 2^k \cdot 16^k \right)^{\frac{1}{k}}} = \frac{1}{6} \frac{N^{\frac{1}{k}}}{96} = \Omega\left(  N^{\frac{1}{k}} \right)$.
\end{proof}

\section{Upper bound appendix}
\label{sec:appendix-upper-bound}

This section contains the remaining details of our upper bound. 

\subsection{Initial value gadgets for $k \mod 4 = 0$}
\label{sec:appendix-upper-bound-initial-value}

In Figures~\ref{fig:Initial_overview_seed_start} through~\ref{fig:Initial_overview_seed_to_next_significant_digit_region}, we create the gadgets that self-assemble the initial value $s$ of the counter when $k \mod 4 = 0$. We will assume that $d_{w-1}, \ldots, d_0$ are the base-$M$ digits of $s$, where $d_{w-1}$ is the most significant digit and $d_0$ is the least significant digit.

Figures~\ref{fig:Initial_overview_seed_start} through~\ref{fig:Initial_overview_seed_to_next_significant_digit_region} also show an example assembly sequence, where, in general, each figure continues the sequence from the resulting assembly in the previously-numbered figure, unless explicitly stated otherwise. In each figure, the black tiles belong to the gadget that is currently self-assembling, starting from the black tile that connects to a white (or the seed) tile. Figure~\ref{fig:Initial_overview_full_initial_value} shows a fully assembled example of the initial value of the counter. 
	\begin{figure}[!h]
		\centering
		\includegraphics[width=\textwidth ]{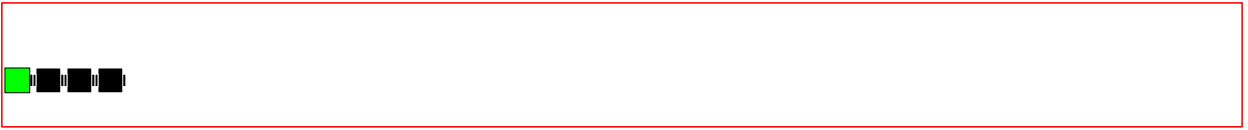}
		\caption{\label{fig:Initial_overview_seed_start} The {\tt Seed\_start} gadget is shown here. It is the only gadget that does not have an input glue. The westernmost tile in the {\tt Seed\_start} gadget is the seed tile type.  In general, we create one {\tt Seed\_start} gadget, contributing $O(1)$ tile types. }
	\end{figure}

	\begin{figure}[!h]
		\centering
		\includegraphics[width=\textwidth ]{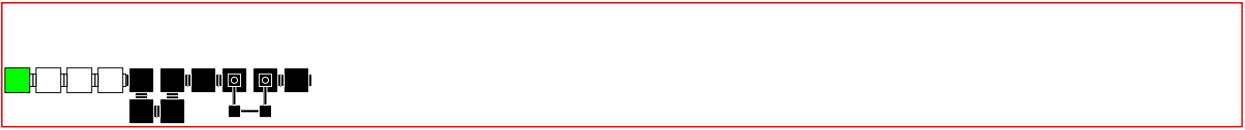}
		\caption{\label{fig:Initial_overview_write_even_digits} A series of two \texttt{Write\_even\_digit}  gadgets is shown here. Each bit of an even digit is represented by a corresponding \texttt{Write\_even\_digit} gadget having its bump in the plane $z=0$ (resp., $z=1$) if the bit being represented is 0 (resp., 1). The gadgets depicted here are: \texttt{Write\_even\_digit\_0} and \texttt{Write\_even\_digit\_1}, encoding the binary string 01, which we, in our construction, will interpret as the binary representation of the base-10 value 2. That is, the westernmost bit in a digit is its least significant bit. Since the digit region in this example contains the seed tile, 2 is the value of the least significant digit. In general, we create a series of {\tt Write\_even\_digit} gadgets for each digit $d_{i}$, where $i < w - 1$ is an even number, contributing $O(k m)$ tile types.}
	\end{figure}

	\begin{figure}[!h]
		\centering
		\includegraphics[width=\textwidth ]{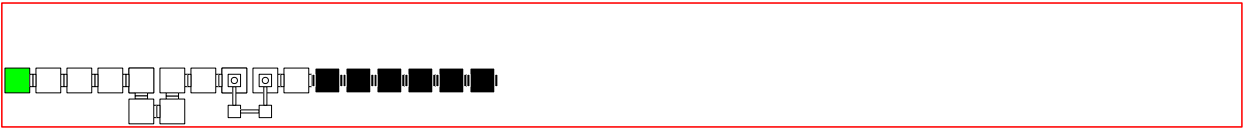}
		\caption{\label{fig:Initial_overview_seed_even_digit_to_odd_digit} A {\tt Seed\_even\_digit\_to\_odd\_digit} gadget is shown here. In general, we create one  {\tt Seed\_even\_digit\_to\_odd\_digit} gadget for each digit region of the initial value, contributing $O(k)$ tile types. }
	\end{figure}

	\begin{figure}[!h]
		\centering
		\includegraphics[width=\textwidth ]{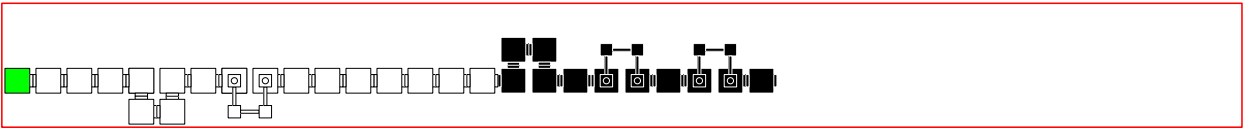}
		\caption{\label{fig:Initial_overview_write_odd_digits} A series of three \texttt{Write\_odd\_digit} gadgets is shown here. Each bit of an odd digit is represented by a corresponding \texttt{Write\_odd\_digit} gadget. A \texttt{Write\_odd\_digit} gadget is similar to its \texttt{Write\_even\_digit} counterpart, except the bit bumps of the latter face to the south while those of the former face to the north. An odd digit has an additional (westernmost) bit indicating whether the digit is the most significant digit. If $k \mod 4 = 2$, then this extra bit indicates whether the digit is the second-most significant digit, or the most significant digit contained in a (general) digit region. 
		The gadgets depicted here, from west to east are: \texttt{Write\_odd\_digit\_0}, \texttt{Write\_odd\_digit\_1} and \texttt{Write\_odd\_digit\_1}, encoding the binary string 011, which we, in our construction, will interpret as the binary representation of the base-10 value 3, and this digit is not the most significant digit, as indicated by its most significant digit indicator bit having the value 0. In general, we create a series of {\tt Write\_odd\_digit} gadgets for each digit $d_{i}$, where $i \leq m - 1$ is an odd number, contributing $O(k m)$ tile types.}
		\end{figure}

	\begin{figure}[!h]
		\centering
		\includegraphics[width=\textwidth ]{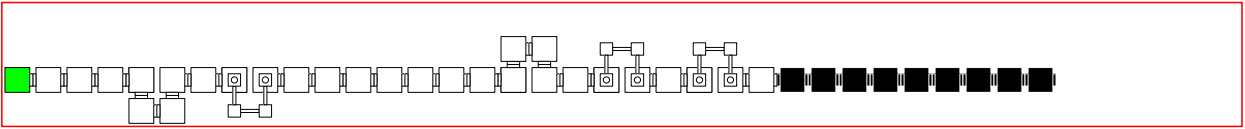}
		\caption{\label{fig:Initial_overview_single_tile_0} A path of \texttt{Single\_tile} gadgets is shown here. The non-constant length of this path, which has to stop exactly six tiles from the easternmost edge of the digit region, facilitates a special case digit region that contains the most significant digit when $k \mod 4 = 2$. Note that a {\tt Single\_tile} gadget is comprised of a single tile whose input glue is always north-facing, and whose output glue is always south-facing. We create one {\tt Single\_tile} gadget for each location in the general version of the depicted path of length $\left(8+3 m + 1\right) - (2+4)$, for each digit region of the initial value, contributing $O(k m)$ tile types.   }
	\end{figure}

	\begin{figure}[!h]
		\centering
		\includegraphics[width=\textwidth ]{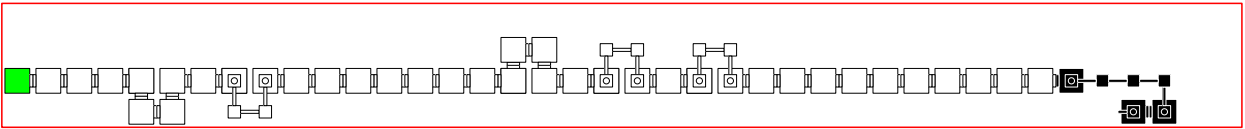}
		\caption{\label{fig:Initial_overview_stopper_after_odd_digit} A \texttt{Stopper\_after\_odd\_digit} gadget is shown here.
		The {\tt Stopper\_after\_odd\_digit} gadgets being created here serve the same purpose as the {\tt Stopper\_after\_odd\_digit} gadgets that were created in Figure~\ref{fig:General_overview_stopper_after_odd_digit}, but here, we create separate {\tt Stopper\_after\_odd\_digit} gadgets for each digit region.
		In general, we create one \texttt{Stopper\_after\_odd\_digit} gadget for each digit region of the initial value, contributing $O(k)$ tile types.}
	\end{figure}

	\begin{figure}[!h]
		\centering
		\includegraphics[width=\textwidth ]{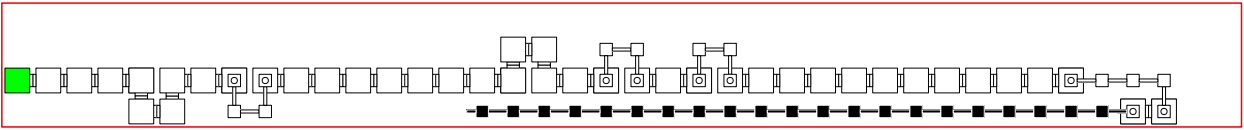}
		\caption{\label{fig:Initial_overview_single_tile_opposite_0} A path of \texttt{Single\_tile\_opposite} gadgets is shown here. Note that a {\tt Single\_tile\_opposite} gadget is comprised of a single tile whose input glue is always north-facing, and whose output glue is always south-facing. We create one \texttt{Single\_tile\_opposite} gadget for each location in the general version of the depicted path of length $6m + 9$, for each digit region of the initial value, contributing $O(k m)$ tile types.  }
	\end{figure}

	\begin{figure}[!h]
		\centering
		\includegraphics[width=\textwidth ]{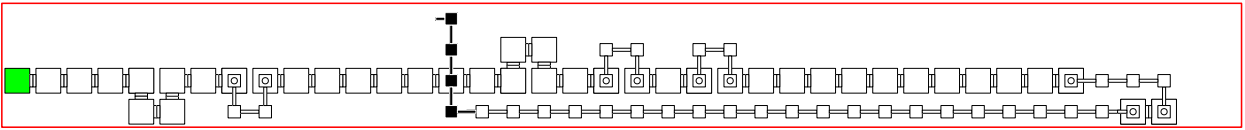}
		\caption{\label{fig:Initial_overview_between_digits} A \texttt{Between\_digits} gadget is shown here. In general, we create one \texttt{Between\_digits} gadget for each digit region of the initial value, contributing $O(k)$ tile types. }
	\end{figure}

	\begin{figure}[!h]
		\centering
		\includegraphics[width=\textwidth ]{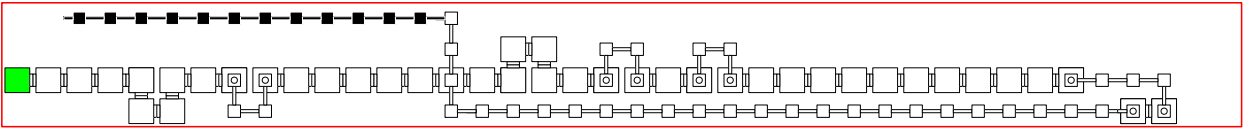}
		\caption{\label{fig:Initial_overview_single_tile_opposite_1}  A path of \texttt{Single\_tile\_opposite} gadgets is shown here. We create one  \texttt{Single\_tile\_opposite} gadget for each location in the general version of the depicted path of length $3m+ 6$, for each digit region of the initial value, contributing $O(k m)$ tile types. Note that, if the most significant digit indicator bit is 1 (it is 0 in the depicted example), then the {\tt Reset\_turn\_corner} gadget created in Figure~\ref{fig:General_overview_return_turn_corner} would attach to the last {\tt Single\_tile\_opposite} gadget in the depicted path.}
	\end{figure}

	\begin{figure}[!h]
		\centering
		\includegraphics[width=\textwidth ]{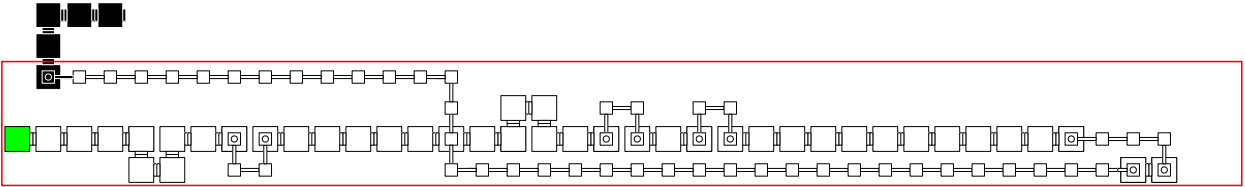}
		\caption{\label{fig:Initial_overview_seed_to_next_significant_digit_region} A \texttt{Seed\_to\_next\_significant\_digit\_region} gadget is shown here. In this example, the digit region does not contain the most significant digit, so a \texttt{Seed\_to\_next\_significant\_digit\_region} gadget self-assembles into the digit region in which the next two most significant digits are contained. The next gadget to self-assemble after a {\tt Seed\_to\_next\_significant\_digit\_region} would be a {\tt Write\_even\_digit} gadget (see Figure~\ref{fig:Initial_overview_write_even_digits}). In general, we create one  \texttt{Seed\_to\_next\_digit\_region} gadget for each digit region that does not contain the most significant digit, contributing $O(k)$ tile types.       }
	\end{figure}


	\clearpage
	
\subsection{All gadgets for $k \mod 4 = 2$}
\label{sec:appendix-upper-bound-special}

%
%
We will now consider the case where $k \mod 4 = 2$. For this case, it suffices to encode the most significant counter digit using only two rows. To that end, we will use a \emph{special} case digit region, which is a digit region whose dimensions are two rows by $l$ columns, that contains one (most significant) even digit. Figure~\ref{fig:Special_digit_regions_dimensions} shows a high-level overview of how the digits (that comprise a value) of the counter are partitioned into digit regions when $k \mod 4 = 2$.

\begin{figure}[!h]
	\centering
	\includegraphics[width=.6 \textwidth ]{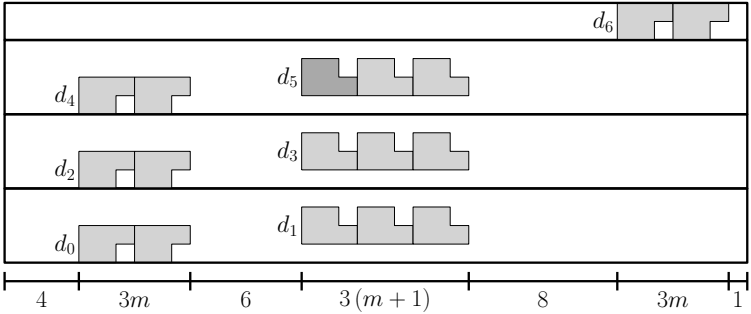}			\caption{\label{fig:Special_digit_regions_dimensions}  This example shows how the digits (that comprise a value) of the counter are partitioned into digit regions when $k \mod 4 = 2$. Recall that we include an ``extra'' $\Theta\left( m \right)$ columns in a general digit region (see Figure~\ref{fig:Overview_digit_regions_dimensions}). We do this to accommodate the most significant (even) digit of the counter in a special case digit region. Notice that we set the least significant indicator bit of the second-most significant digit to 1.}
\end{figure}

Assume the existence of all the gadgets that were created in Figures~\ref{fig:Initial_overview_seed_start} through~\ref{fig:Initial_overview_seed_to_next_significant_digit_region} and Figure~\ref{fig:General_overview_return_read_even_digit}. In Figures~\ref{fig:Special_initial_overview_seed_to_most_significant_digit_region} through~\ref{fig:Special_initial_overview_return_single_tile}, we create the gadgets that self-assemble the initial value of the counter, when $k \mod 4 = 2$. Figures~\ref{fig:Special_initial_overview_seed_to_most_significant_digit_region} through~\ref{fig:Special_initial_overview_return_single_tile} also show an example assembly sequence, where, unless specified otherwise, each figure continues the sequence from the resulting assembly in the previously-numbered figure. A fully assembled example of the initial value of the counter, when $k \mod 4 = 2$, is shown in Figure~\ref{fig:Special_initial_overview_full_initial_value}.

\begin{figure}[!h]
	\centering
	\includegraphics[width=\textwidth]{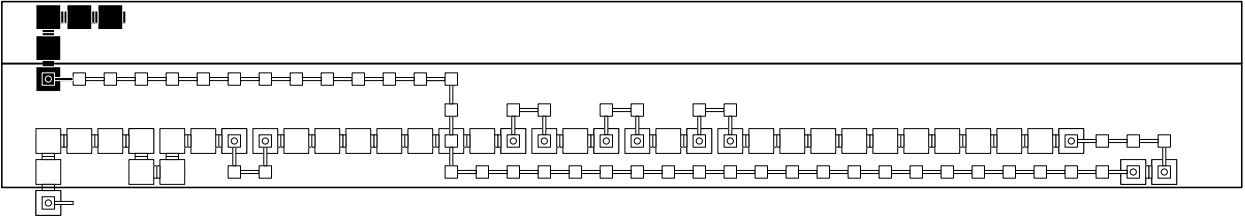}
	\caption{\label{fig:Special_initial_overview_seed_to_most_significant_digit_region}  A {\tt Seed\_to\_next\_significant\_digit\_region} gadget is shown here. Note that this is the same general gadget initiated in Figure~\ref{fig:Initial_overview_seed_to_next_significant_digit_region}, just with different glues. In general, we create one {\tt Seed\_to\_next\_significant\_digit\_region} gadget, replacing the {\tt Reset\_turn\_corner} from Figure~\ref{fig:General_overview_return_turn_corner} and contributing $O(1)$ tile types.}
\end{figure}

\begin{figure}[!h]
	\centering
	\includegraphics[width=\textwidth]{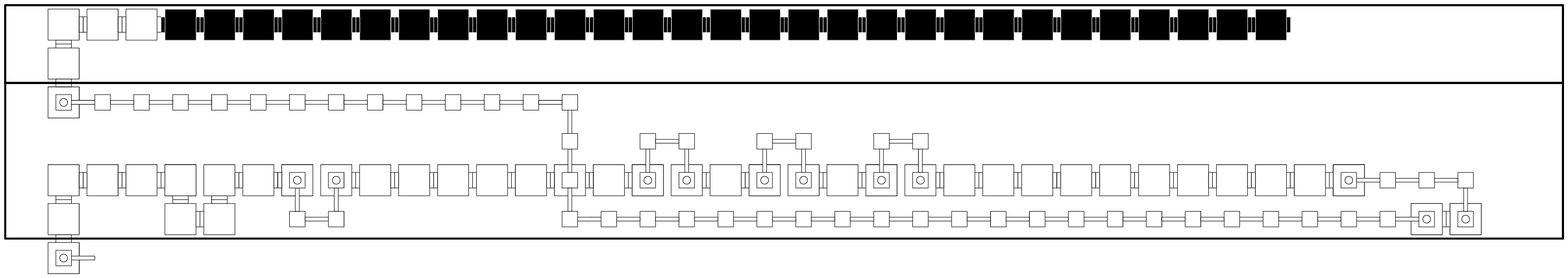}
	\caption{\label{fig:Special_initial_overview_single_tile_0} A path of \texttt{Single\_tile} gadgets is shown here. We create one \texttt{Single\_tile} gadget for each location in the general version of the depicted path of length $3m + 6 + 3\left( m +1 \right) + 8$, contributing $O(m)$ tile types.  }
\end{figure}

\begin{figure}[!h]
	\centering
	\includegraphics[width=\textwidth]{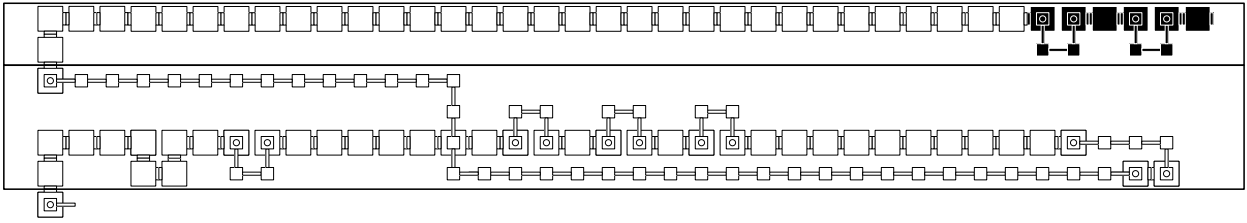}
	\caption{\label{fig:Special_initial_overview_write_even_digit}  A series of two {\tt Write\_even\_digit} gadgets is shown here. In general, we create one {\tt Write\_even\_digit} gadget for each bit of the digit $d_{w - 1}$, where $w - 1$ is an even number, contributing $O(m)$ tile types. }
\end{figure}

\begin{figure}[!h]
	\centering
	\includegraphics[width=\textwidth]{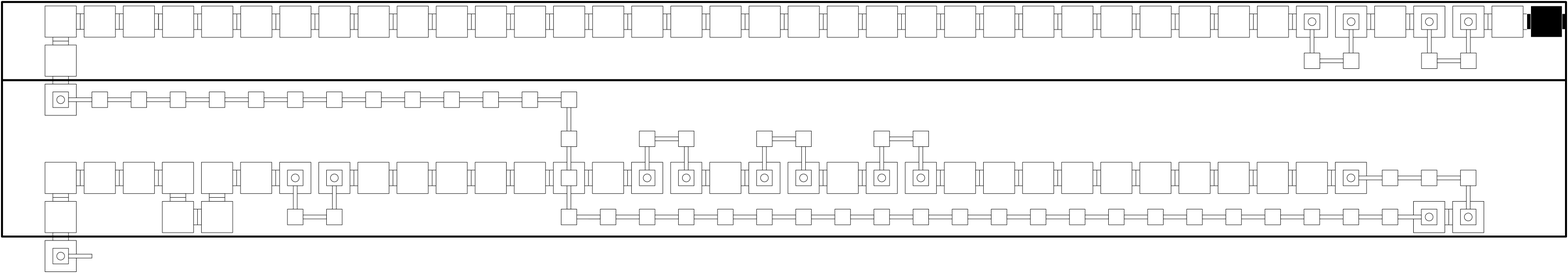}
	\par\bigskip
	\includegraphics[width=\textwidth]{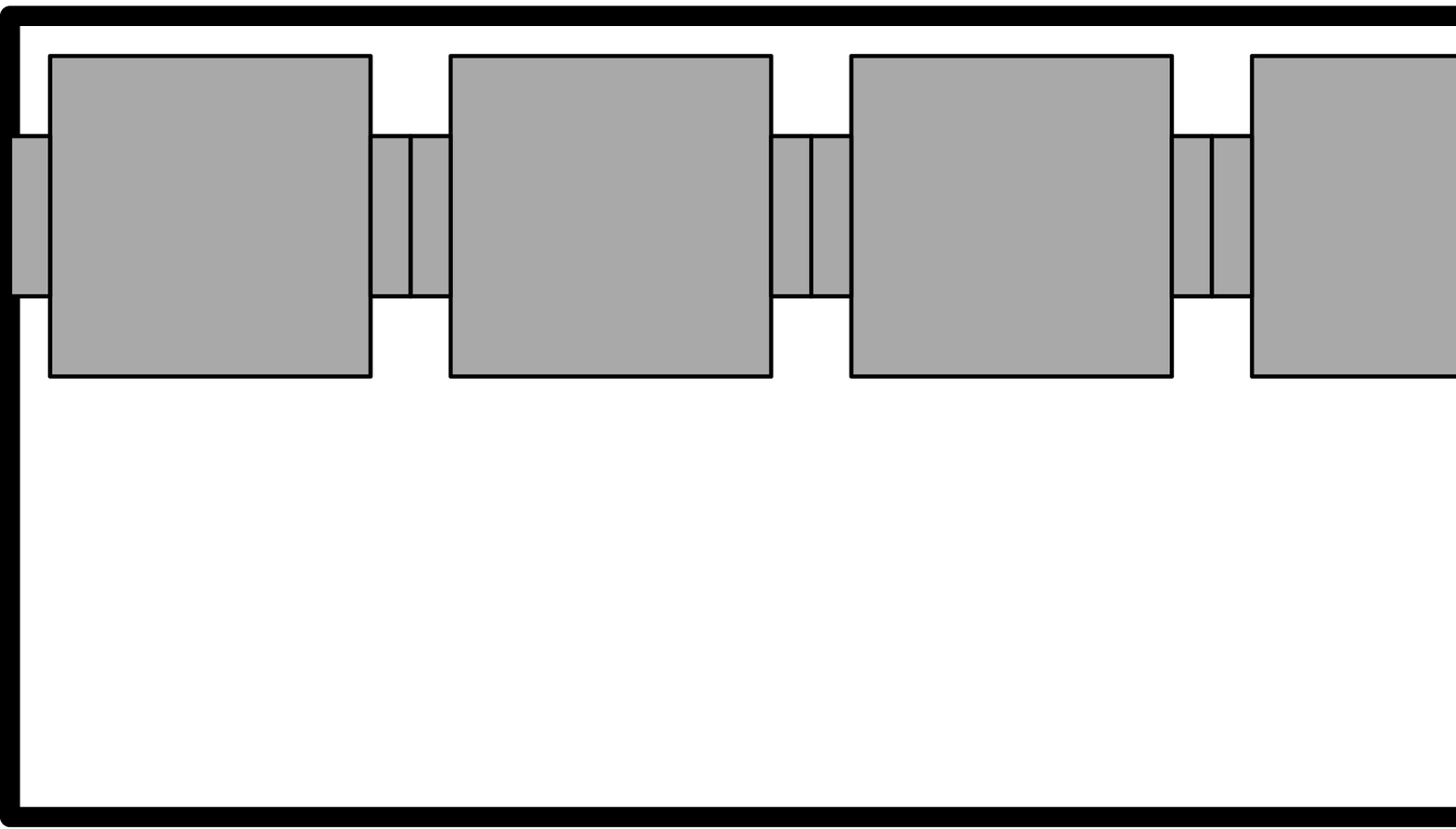}
	\caption{\label{fig:Special_initial_overview_single_tile_1}  A path of \texttt{Single\_tile} gadgets is shown here. The path starts in the current special case digit region (top) and terminates in the corresponding special case digit region (bottom). We create one \texttt{Single\_tile} gadget for each location in the general version of the depicted path of length $1+4+3m + 6 + 3\left( m+1 \right)+3$, contributing $O(m)$ tile types.
	}
\end{figure}

\begin{figure}[!h]
	\centering
	\includegraphics[width=\textwidth]{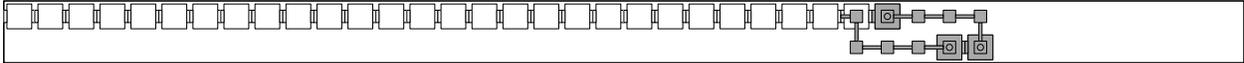}
	\caption{\label{fig:Special_initial_overview_special_stopper}  A {\tt Special\_stopper} gadget is shown here. It is used exclusively within a special case digit region for hindering a repeating path of tiles propagating the value of the most significant digit when $k \mod 4 = 2$. In general, we create one {\tt Special\_stopper} gadget, contributing $O(1)$ tile types.}
\end{figure}

\begin{figure}[!h]
	\centering
	\includegraphics[width=\textwidth]{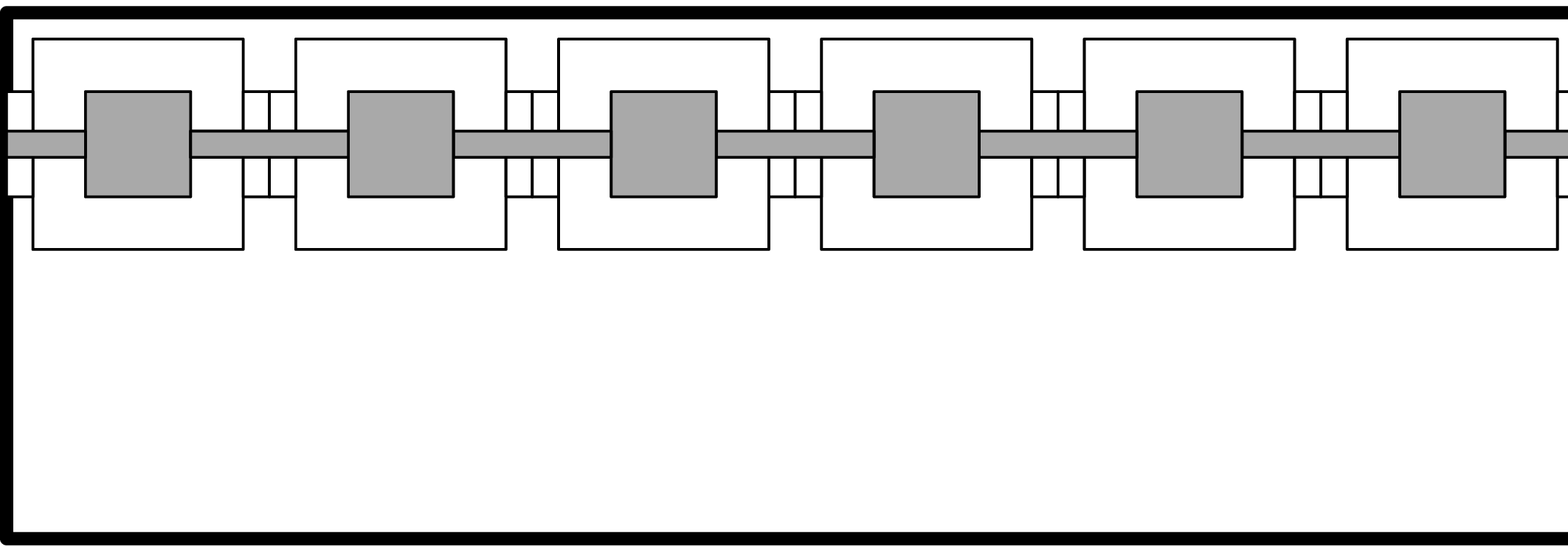}
	\caption{\label{fig:Special_initial_overview_single_tile_opposite_0}  A path of {\tt Single\_tile\_opposite} gadgets is shown here. We create one \texttt{Single\_tile\_opposite} gadget for each location in the general version of the depicted path of length $4+3m + 6+ 3\left( m +1\right)+2$, contributing $O(m)$ tile types.}
\end{figure}

\begin{figure}[!h]
	\centering
	\includegraphics[width=\textwidth]{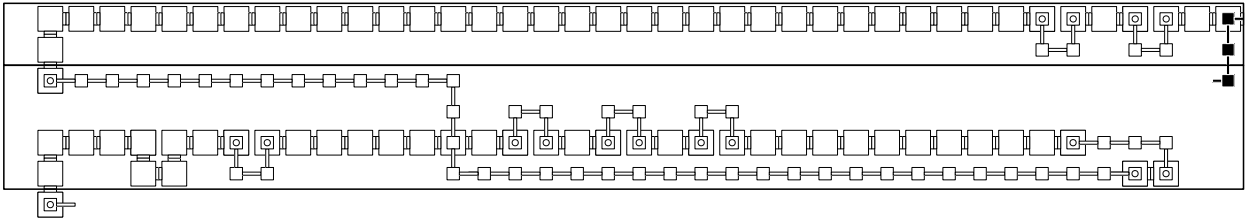}
	\caption{\label{fig:Special_initial_overview_special_at_MSB_of_most_significant_digit}  A {\tt Special\_at\_MSB\_of\_most\_significant\_digit} gadget is shown here. This is one example of a gadget that is used exclusively for the self-assembly of a special case digit region that self-assembles in the digit region in which the next least significant digits, relative to the current special digit region, are contained. The east-facing input glue of the {\tt Special\_at\_MSB\_of\_most\_significant\_digit} gadget binds to the west-facing output glue of the last {\tt Single\_tile\_opposite} gadget to attach in the path from Figure~\ref{fig:Special_initial_overview_single_tile_opposite_0}. In general, we create one {\tt Special\_at\_MSB\_of\_most\_significant\_digit} gadget, contributing $O(1)$ tile types.}
\end{figure}

\begin{figure}[!h]
	\centering
	\includegraphics[width=\textwidth]{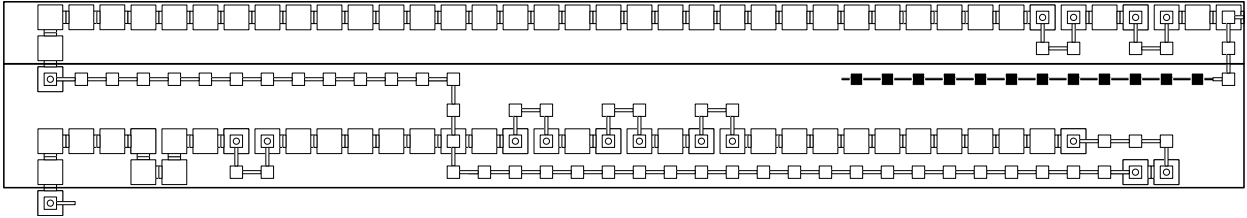}
	\caption{\label{fig:Special_initial_overview_single_tile_opposite_1}  A path of {\tt Single\_tile\_opposite} gadgets is shown here. These gadgets are examples of gadgets that are used exclusively for the self-assembly of a special case digit region that self-assembles in the digit region in which the next least significant digits, relative to the current digit region, are contained. We create one \texttt{Single\_tile\_opposite} gadget for each location in the general version of the depicted path of length $6 + 3 m$, contributing $O(m)$ tile types.}
\end{figure}

\begin{figure}[!h]
	\centering
	\includegraphics[width=\textwidth]{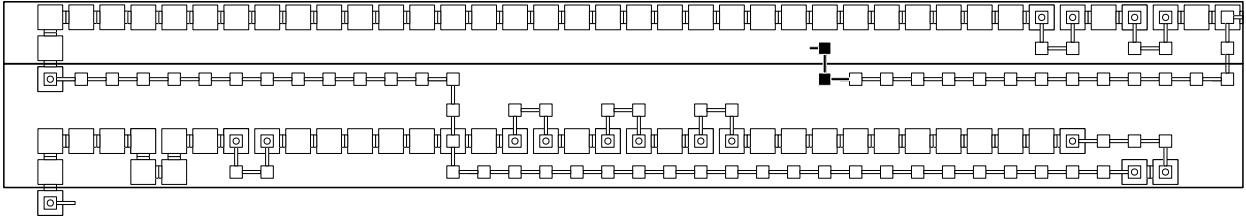}
	\caption{\label{fig:Special_initial_overview_at_MSB_of_odd_digit}  An {\tt At\_MSB\_of\_odd\_digit} gadget is shown here. This is the last example of a gadget that is used exclusively for the self-assembly of a special case digit region that self-assembles in the digit region in which the next least significant digits, relative to the current digit region, are contained. In general, we create one {\tt At\_MSB\_of\_odd\_digit}, contributing $O(1)$ tile types. }
\end{figure}

\begin{figure}[!h]
	\centering
	\includegraphics[width=\textwidth]{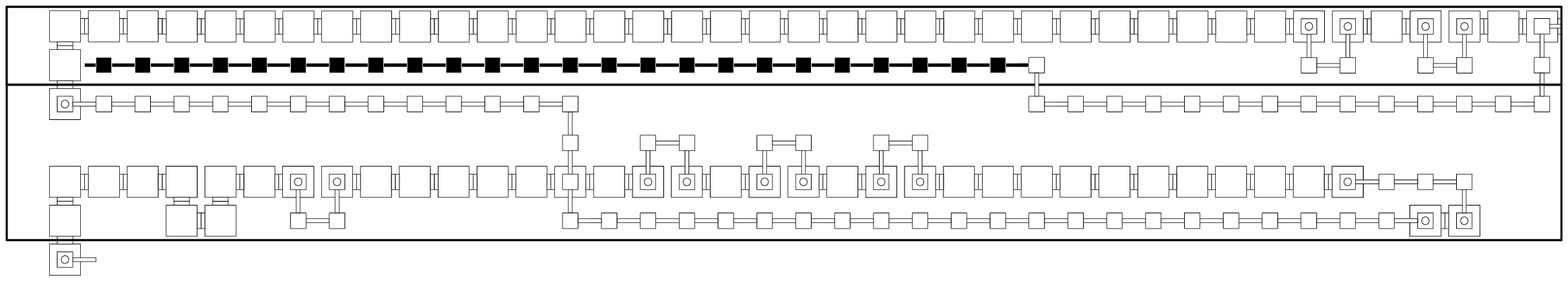}
	\caption{\label{fig:Special_initial_overview_single_tile_opposite_2}  A path of {\tt Single\_tile\_opposite} gadgets is shown here. We create one \texttt{Single\_tile\_opposite} gadget for each location in the general version of the depicted path of length $2+3 m+ 6 + 3\left( m + 1\right) +1$, contributing $O(m)$ tile types.}
\end{figure}

\begin{figure}[!h]
	\centering
	\includegraphics[width=\textwidth]{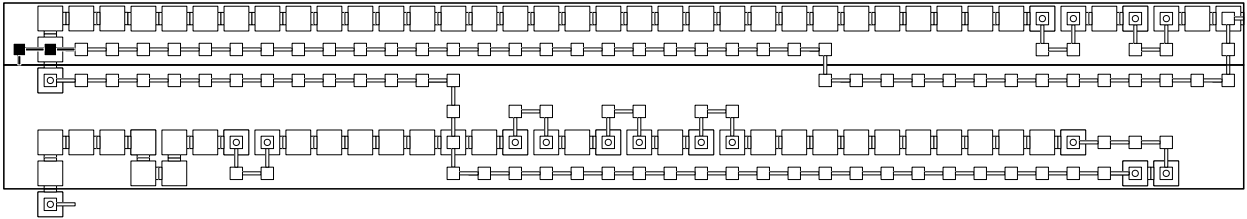}
	\caption{\label{fig:Special_initial_overview_return_turn_corner}  A {\tt Reset\_turn\_corner} gadget is shown here.  The gadget being created here replaces the gadget being created in Figure~\ref{fig:General_overview_return_turn_corner}. In general, we create one {\tt Reset\_turn\_corner} gadget, contributing $O(1)$ tile types.}
\end{figure}

\begin{figure}[!h]
	\centering
	\includegraphics[width=\textwidth]{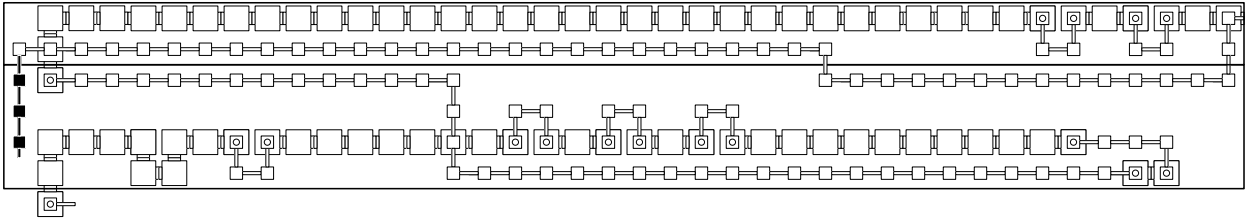}
	\caption{\label{fig:Special_initial_overview_return_single_tile}  A (beginning portion of a) path of {\tt Reset\_single\_tile} gadgets is shown here. The gadgets being created here replace the gadgets that were created in Figure~\ref{fig:General_overview_return_single_tile}. The {\tt Reset\_read\_even\_digit} gadget created in Figure~\ref{fig:General_overview_return_read_even_digit} attaches to the last {\tt Reset\_single\_tile} gadget in the series being created here. We create one {\tt Reset\_single\_tile} gadget for each location in the general version of the depicted path of length $k - 3$, contributing $O(k)$ tile types. }
\end{figure}

\begin{figure}[!th]
	\centering
	\includegraphics[width=\textwidth]{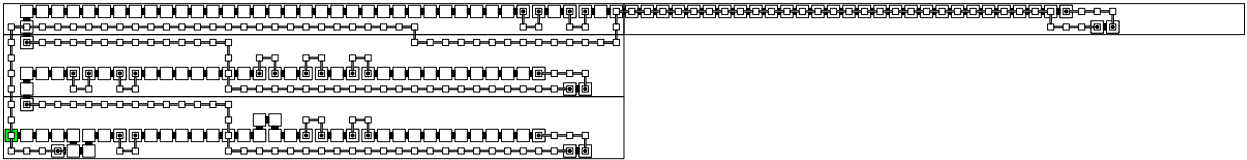}
	\caption{\label{fig:Special_initial_overview_full_initial_value} A fully assembled example of the initial value of the counter, when $k \mod 4 = 2$, showing the current and corresponding special case digit regions. }
\end{figure}

\clearpage

%
%
%
%
%

In Figures~\ref{fig:Special_overview_single_tile_opposite} through~\ref{fig:Special_overview_write_even_digit}, we create the gadgets that implement the self-assembly algorithm that increments the value of the counter, when $k \mod 4 = 2$. Figures~\ref{fig:Special_overview_single_tile_opposite} through~\ref{fig:Special_overview_write_even_digit} also show an example assembly sequence, where, unless specified otherwise, each figure continues the sequence from the resulting assembly in the previously-numbered figure.

\begin{figure}[!h]
	\centering
	\includegraphics[width=\textwidth]{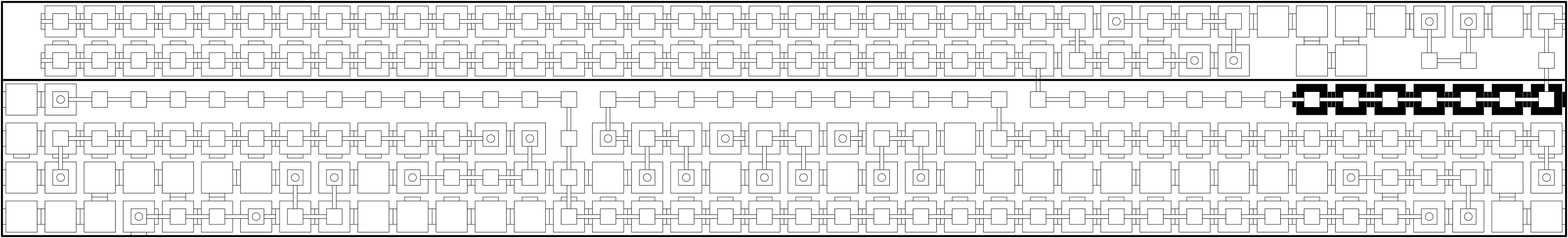}
	\par\bigskip
	\includegraphics[width=\textwidth]{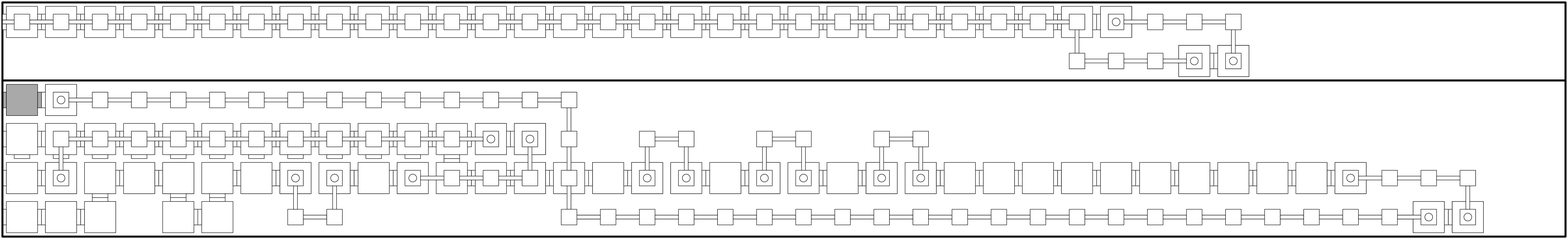}
	\caption{\label{fig:Special_overview_single_tile_opposite}  A path of {\tt Single\_tile\_opposite} gadgets is shown here, originating in the corresponding digit region (bottom) and terminating in the current digit region (top). If $k \mod 4 = 2$, then we would not create the {\tt Reset\_turn\_corner} gadget in Figure~\ref{fig:Special_initial_overview_return_turn_corner}. Instead, we would create the {\tt Z1\_to\_z0} gadget in Figure~\ref{fig:General_overview_z_1_to_z_0}, even though the odd digit would have its indicator bit set to 1, to which the first gadget in the general version of the depicted path of {\tt Single\_tile\_opposite} gadgets would attach. We create $O(1)$ {\tt Single\_tile\_opposite} gadgets for each location in the general version of the depicted path of length $3m + 2$, contributing $O(m)$ tile types. }
\end{figure}

\begin{figure}[!h]
	\centering
	\includegraphics[width=\textwidth]{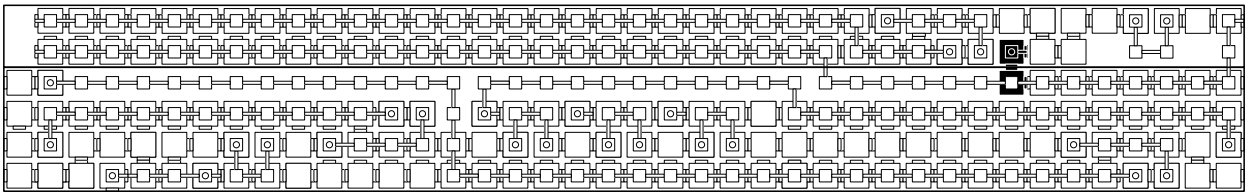}
	\caption{\label{fig:Special_overview_start_read_most_significant_even_digit}  A {\tt Start\_digit\_region} gadget is shown here. In general, we create $O(1)$ {\tt Start\_digit\_region} gadgets, contributing $O(1)$ tile types. }
\end{figure}

\begin{figure}[!h]
	\centering
	\includegraphics[width=\textwidth]{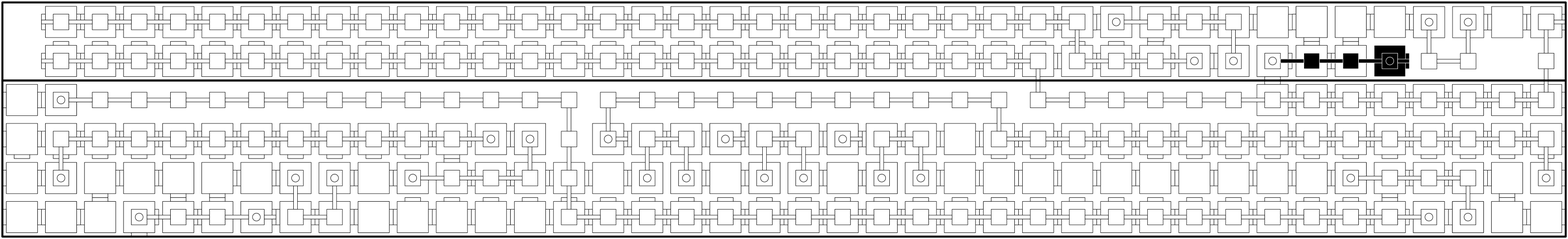}
	\caption{\label{fig:Special_overview_read_non_MSB}  A {\tt Read\_non\_MSB} gadget is shown here. In general, we create $O(M)$ {\tt Read\_non\_MSB} gadgets, contributing $O(M)$ tile types.
	}
\end{figure}

\begin{figure}[!h]
	\centering
	\includegraphics[width=\textwidth]{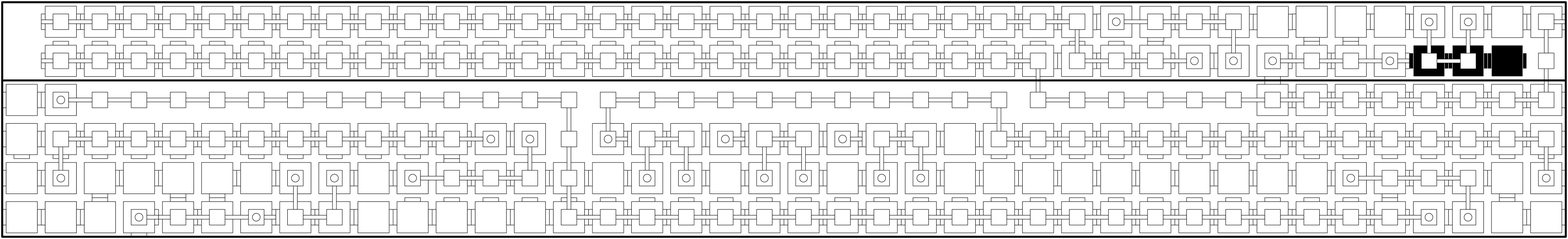}
	\caption{\label{fig:Special_overview_read_MSB}  A {\tt Read\_MSB} gadget is shown here. In general, we create $O(M)$ {\tt Read\_MSB} gadgets, contributing $O(M)$ tile types. }
\end{figure}

\begin{figure}[!h]
	\centering
	\includegraphics[width=\textwidth]{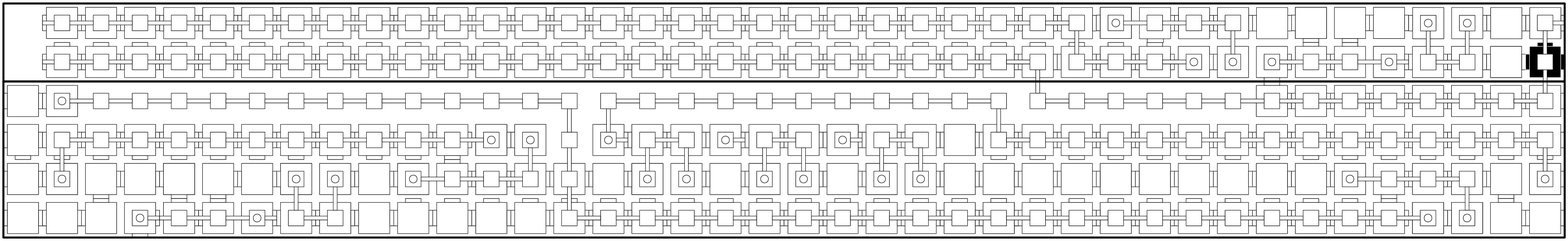}
	\par\bigskip
	\includegraphics[width=\textwidth]{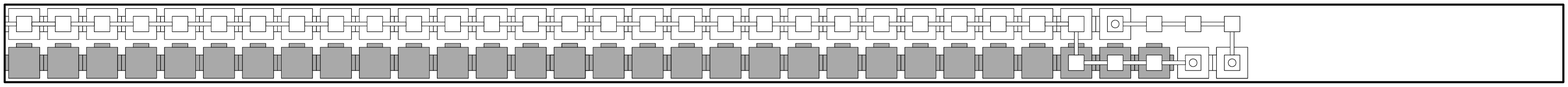}

	\caption{\label{fig:Special_overview_repeating_after_even_digit}  A path of {\tt Repeating\_after\_even\_digit} gadgets is shown here. We propagate the result of reading the most significant digit along a path of {\tt Repeating\_after\_even\_digit} tiles, starting in the current special case digit region (top) and terminating, by being hindered by the {\tt Special\_stopper} gadget, in the corresponding special case digit region (bottom). In general, we create $O(M)$ {\tt Repeating\_after\_even\_digit} gadgets, contributing $O(M)$ tile types.}
\end{figure}

\begin{figure}[!h]
	\centering
	\includegraphics[width=\textwidth]{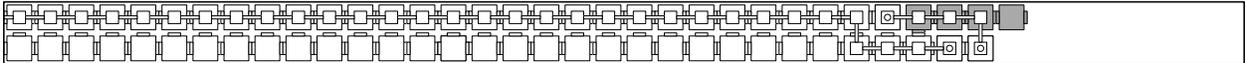}
	\caption{\label{fig:Special_overview_at_special_stopper}  An {\tt At\_special\_stopper} gadget is shown here. It has a fixed size. The north-facing glue of the last {\tt Repeating\_after\_even\_digit} gadget to attach in the path shown in Figure~\ref{fig:Special_overview_repeating_after_even_digit} will initiate the self-assembly of an {\tt At\_special\_stopper} gadget. If $x \in \{0,1\}^{m}$ and $c \in \{0,1\}$ are contained in the output glue of the former, where $c=1$ indicates the presence of an arithmetic carry and $c=0$ otherwise, then the output glue of the latter will contain the $m$-bit binary representation of $\left( x + c \right) \mod M$. If $\left( x + c \right) \mod M = 0$, then $c=1$ is contained in the output glue of the gadgets being created here, in which case the counter has rolled over to 0. In general, we create $O(M)$ {\tt At\_special\_stopper} gadgets, contributing $O(M)$ tile types.}
\end{figure}

\begin{figure}[!h]
	\centering
	\includegraphics[width=\textwidth]{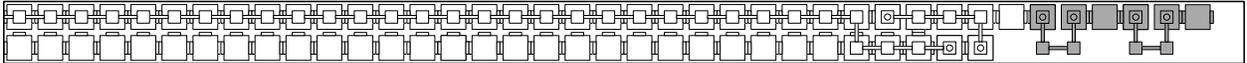}
	\caption{\label{fig:Special_overview_write_even_digit}  A series of {\tt Write\_even\_digit} gadgets is shown here. After the last {\tt Write\_even\_digit} gadget self-assembles, depending on whether the counter has rolled over to 0, the gadgets created in Figures~\ref{fig:Special_initial_overview_single_tile_1},~\ref{fig:Special_initial_overview_special_stopper},~\ref{fig:Special_initial_overview_single_tile_opposite_0},~\ref{fig:Special_initial_overview_special_at_MSB_of_most_significant_digit},~\ref{fig:Special_initial_overview_single_tile_opposite_1},~\ref{fig:Special_initial_overview_at_MSB_of_odd_digit},~\ref{fig:Special_initial_overview_single_tile_opposite_2},~\ref{fig:Special_initial_overview_return_turn_corner},and~\ref{fig:Special_initial_overview_return_single_tile} may self-assemble. In general, we create $O(M)$ {\tt Write\_even\_digit} gadgets, contributing $O(M)$ tile types.}
\end{figure}

\clearpage

\subsection{Full details}
\label{sec:appendix-upper-bound-full}

In this section, we give the full details of our construction.

We say that a gadget is \emph{general} if its input and output glues are undefined. If {\tt Gadget} is a general gadget, then we use the notation ${\tt Gadget}({\tt a}, {\tt b})$ to represent the creation of the \emph{specific gadget}, or simply \emph{gadget}, referred to as {\tt Gadget}, with input glue label {\tt a} and output glue label {\tt b} (all positive glue strengths are $1$). If a gadget has two possible output glues, then we will use the notation ${\tt Gadget}({\tt a}, {\tt b}, {\tt c})$ to denote the specific version of {\tt Gadget}, where {\tt a} is the input glue and {\tt b} and {\tt c} are the two possible output glues, listed in the order north, east, south and west, with all of the $z=0$ output glues listed before the $z=1$ output glues. If a gadget has only one output glue (and no input glue), like a gadget that contains the seed, or if a gadget has only one input glue (and no output glue), then we will use the notation ${\tt Gadget}({\tt a})$. We use the notation $\langle \cdot \rangle$ to denote some standard encoding of the concatenation of a list of symbols.
We group the general gadgets that we use in our construction into eight groups named Write (Figure~\ref{fig:Write_gadgets}), Read (Figure~\ref{fig:Read_gadgets}), Seed (Figure~\ref{fig:Seed_gadgets}), Hardcoded-length spacer (Figure~\ref{fig:Hardcoded_length_gadgets}), Blocking-based spacer (Figure~\ref{fig:Blocking_based_gadgets}), Transition (Figure~\ref{fig:Transition_gadgets}), Reset (Figure~\ref{fig:Reset_for_next_increment_gadgets}), and Special case (Figure~\ref{fig:Special_case_gadgets}). 
%

%
%
\begin{figure}[!ht]
    \centering
    \begin{subfigure}[t]{0.45\textwidth}
        \centering
        \includegraphics[width=38.75px]{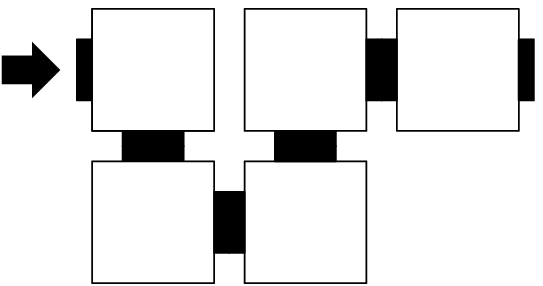}
        \caption{\label{fig:Gadget_write_even_digit_0} {\tt Write\_even\_digit\_0}}
    \end{subfigure}%
    \begin{subfigure}[t]{0.45\textwidth}
        \centering
        \includegraphics[width=38.75px]{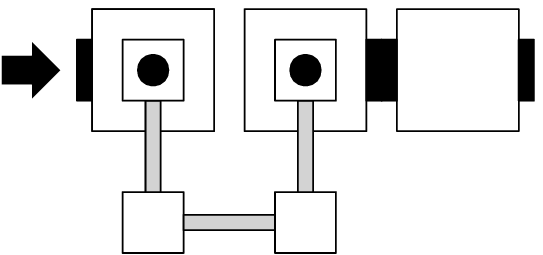}
        \caption{\label{fig:Gadget_write_even_digit_1} {\tt Write\_even\_digit\_1}}
    \end{subfigure} 
    
    \vspace{10pt}

    \begin{subfigure}[t]{0.45\textwidth}
        \centering
        \includegraphics[width=38.75px]{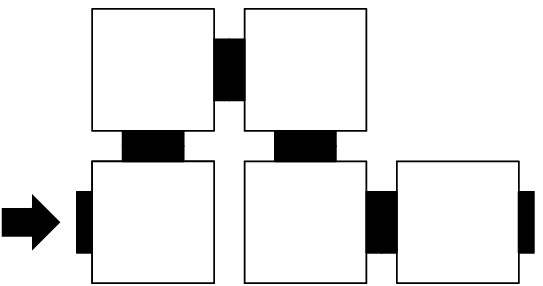}
        \caption{\label{fig:Gadget_write_odd_digit_0} {\tt Write\_odd\_digit\_0}}
    \end{subfigure}%
    \begin{subfigure}[t]{0.45\textwidth}
        \centering
        \includegraphics[width=38.75px]{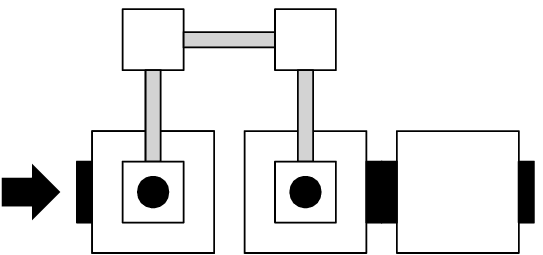}
        \caption{\label{fig:Gadget_write_odd_digit_1} {\tt Write\_odd\_digit\_1}}
    \end{subfigure}%
    \caption{\label{fig:Write_gadgets}
        The ``Write'' gadgets.
        These are the gadgets used in both the initial and all subsequent values of the counter to encode bits of a digit.
        Regardless of the parity of the digits, bits with a value of 0 are encoded using the $z = 0$ plane, and bits with a value of 1 are encoded using the $z = 1$ plane.
    }
\end{figure}

%
%
\begin{figure}[!ht]
    \centering

    \begin{subfigure}[t]{0.47\textwidth}
        \centering
        \includegraphics[width=38.75px]{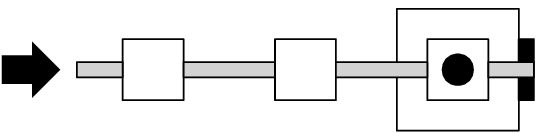}
        \caption{\label{fig:Gadget_read_non_MSB_0} {\tt Read\_non\_MSB\_0}}
    \end{subfigure}%
    \begin{subfigure}[t]{0.47\textwidth}
        \centering
        \includegraphics[width=38.75px]{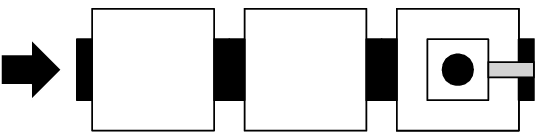}
        \caption{\label{fig:Gadget_read_non_MSB_1} {\tt Read\_non\_MSB\_1}}
    \end{subfigure}%
    
     \vspace{10pt}

    \begin{subfigure}[t]{0.47\textwidth}
        \centering
        \includegraphics[width=38.75px]{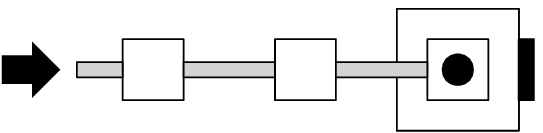}
        \caption{\label{fig:Gadget_read_MSB_0} {\tt Read\_MSB\_0}}
    \end{subfigure}%
    \begin{subfigure}[t]{0.47\textwidth}
        \centering
        \includegraphics[width=38.75px]{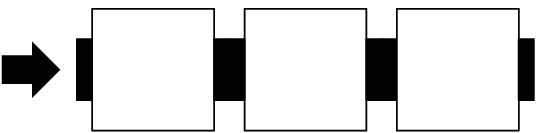}
        \caption{\label{fig:Gadget_read_MSB_1} {\tt Read\_MSB\_1}}
    \end{subfigure}%
    \caption{\label{fig:Read_gadgets}%
        The ``Read'' gadgets.
        These gadgets are used by the counter to read the bits of a digit. 
        Since a bit with a value of 0 is encoded using the $z = 0$ plane, a gadget that reads a 0 begins self-assembling in the $z = 1$ plane.
        Similarly, a bit with a value of 1 is read by a gadget that begins self-assembling in the $z = 0$ plane.
    }
\end{figure}

%
%
\begin{figure}[!ht]
    \centering
    \begin{subfigure}[t]{0.24\textwidth}
        \centering
        \includegraphics[width=43.0px]{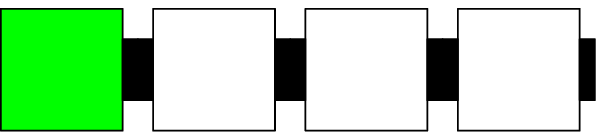}
        \caption{\label{fig:Gadget_seed_start} {\tt Seed\_start}}
    \end{subfigure}%
    \begin{subfigure}[t]{0.31\textwidth}
        \centering
        \includegraphics[width=71.5px]{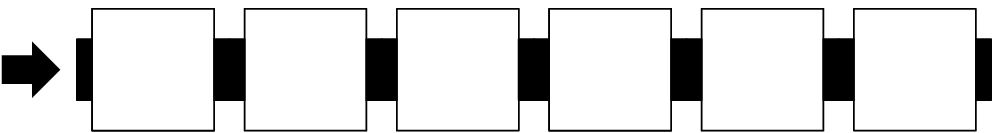}
        \caption{\label{fig:Gadget_seed_even_digit_to_odd_digit} {\tt Seed\_even\_digit\_to\_odd\_digit}}
    \end{subfigure}\hspace{0.05\textwidth}%
    \begin{subfigure}[t]{0.4\textwidth}
        \centering
        \includegraphics[width=32.0px]{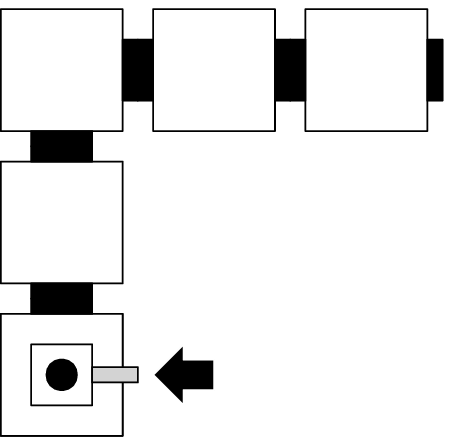}
        \caption{\label{fig:Gadget_seed_to_next_significant_digit_region} {\tt Seed\_to\_next\_significant\_digit\_region}}
    \end{subfigure}%
    \caption{\label{fig:Seed_gadgets}
        The ``Seed'' or ``initial value'' gadgets.
        These gadgets are used exclusively in the self-assembly of the initial value.
    }
\end{figure}

%
%
\begin{figure}[!ht]
\centering
    \begin{subfigure}[t]{0.45\textwidth}
        \centering
        \includegraphics[width=16.75px]{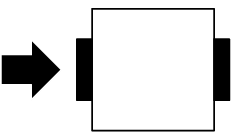}
        \caption{\label{fig:Gadget_single_tile} {\tt Single\_tile}}
    \end{subfigure}%
    \begin{subfigure}[t]{0.45\textwidth}
        \centering
        \includegraphics[width=16.75px]{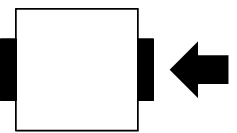}
        \caption{\label{fig:Gadget_single_tile_opposite} {\tt Single\_tile\_opposite}}
    \end{subfigure}%
    \caption{\label{fig:Hardcoded_length_gadgets}
        The ``Hardcoded-length spacer'' gadgets.
        These gadgets are single tile gadgets used throughout the construction.
        Except when used in the initial value, these gadgets never carry information about the bits of the counter, which is key because these gadgets self-assemble in a path whose length depends on $m$.
    }
\end{figure}

%
%
\begin{figure}[!ht]
    \centering
    \begin{subfigure}[t]{0.47\textwidth}
        \centering
        \includegraphics[width=16.75px]{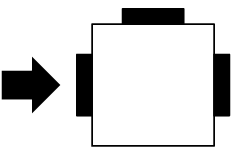}
        \caption{\label{fig:Gadget_repeating_after_even_digit} {\tt Repeating\_after\_even\_digit}}
    \end{subfigure}%
    \begin{subfigure}[t]{0.47\textwidth}
        \centering
        \includegraphics[width=48.5px]{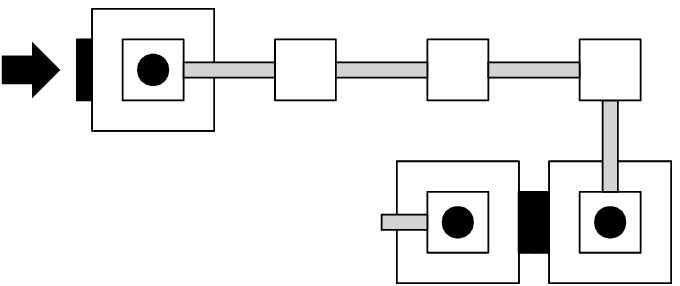}
        \caption{\label{fig:Gadget_stopper_after_odd_digit} {\tt Stopper\_after\_odd\_digit}}
    \end{subfigure}%

	 \vspace{10pt}
	 
    \begin{subfigure}[t]{0.47\textwidth}
        \centering
        \includegraphics[width=16.75px]{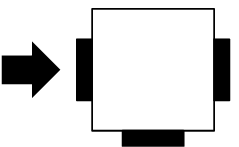}
        \caption{\label{fig:Gadget_repeating_after_odd_digit} {\tt Repeating\_after\_odd\_digit}}
    \end{subfigure}%
    \begin{subfigure}[t]{0.47\textwidth}
        \centering
        \includegraphics[width=48.5px]{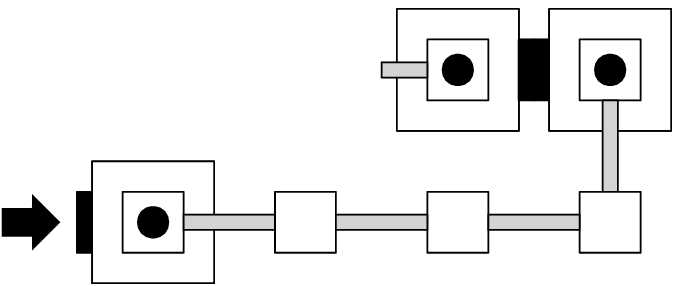}
        \caption{\label{fig:Gadget_stopper_after_even_digit} {\tt Stopper\_after\_even\_digit}}
    \end{subfigure}%
    
     \vspace{10pt}

    \begin{subfigure}[t]{0.47\textwidth}
        \centering
        \includegraphics[width=97.75px]{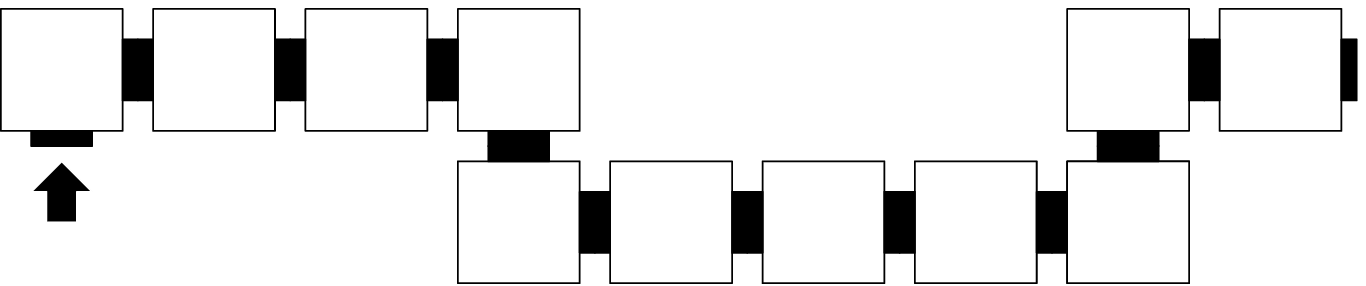}
        \caption{\label{fig:Gadget_at_stopper_after_odd_digit} {\tt At\_stopper\_after\_odd\_digit}}
    \end{subfigure}%
    \begin{subfigure}[t]{0.47\textwidth}
        \centering
        \includegraphics[width=54.0px]{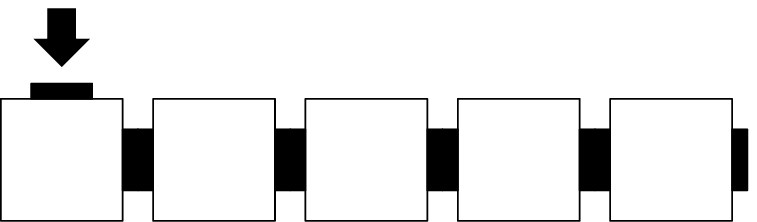}
        \caption{\label{fig:Gadget_at_stopper_after_even_digit} {\tt At\_stopper\_after\_even\_digit}}
    \end{subfigure}%
    \caption{\label{fig:Blocking_based_gadgets}
        The ``Blocking-based spacer'' gadgets.
        These gadgets work together so that the gadgets in Figures~\ref{sub@fig:Gadget_repeating_after_even_digit}~and~\ref{sub@fig:Gadget_repeating_after_odd_digit} can self-assemble into arbitrary length paths until they are eventually blocked by the gadgets in Figures~\ref{sub@fig:Gadget_stopper_after_odd_digit}~and~\ref{sub@fig:Gadget_stopper_after_even_digit}, respectively.
        Once the gadgets in Figures~\ref{sub@fig:Gadget_repeating_after_even_digit} and~\ref{sub@fig:Gadget_repeating_after_odd_digit} are blocked by the gadgets in Figures~\ref{sub@fig:Gadget_stopper_after_odd_digit}~and~\ref{sub@fig:Gadget_stopper_after_even_digit}, the gadgets in Figures~\ref{sub@fig:Gadget_at_stopper_after_odd_digit}~and~\ref{sub@fig:Gadget_at_stopper_after_even_digit} can self-assemble, respectively.
    }
\end{figure}

%
%
\begin{figure}[!ht]
    \centering
    \begin{subfigure}[t]{0.47\textwidth}
        \centering
        \includegraphics[width=16.75px]{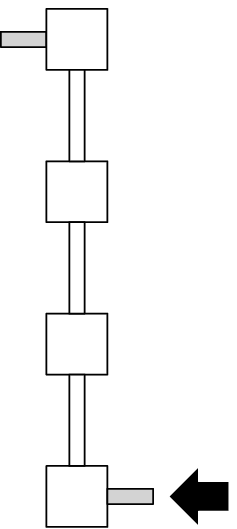}
        \caption{\label{fig:Gadget_between_digits} {\tt Between\_digits}}
    \end{subfigure}%
    \begin{subfigure}[t]{0.47\textwidth}
        \centering
        \includegraphics[width=38.75px]{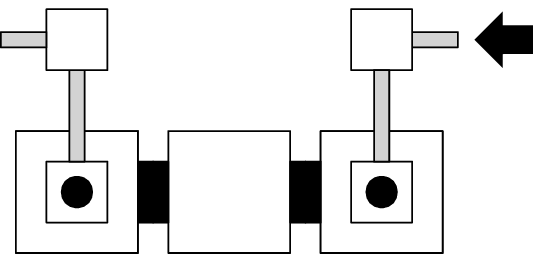}
        \caption{\label{fig:Gadget_between_digit_regions} {\tt Between\_digit\_regions}}
    \end{subfigure}%
    
     \vspace{10pt}

    \begin{subfigure}[t]{0.47\textwidth}
        \centering
        \includegraphics[width=16.75px]{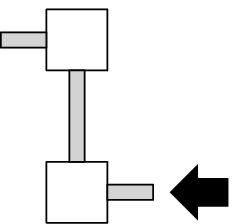}
        \caption{\label{fig:Gadget_at_MSB_of_odd_digit} {\tt At\_MSB\_of\_odd\_digit}}
    \end{subfigure}%
    \begin{subfigure}[t]{0.47\textwidth}
        \centering
        \includegraphics[width=15.5px]{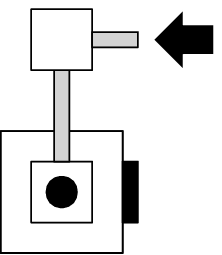}
        \caption{\label{fig:Gadget_start_read_odd_digit} {\tt Start\_read\_odd\_digit}}
    \end{subfigure}%
    
     \vspace{10pt}

    \begin{subfigure}[t]{0.47\textwidth}
        \centering
        \includegraphics[width=15.5px]{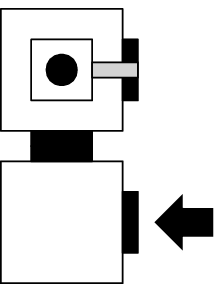}
        \caption{\label{fig:Gadget_start_digit_region} {\tt Start\_digit\_region}}
    \end{subfigure}%
    \begin{subfigure}[t]{0.47\textwidth}
        \centering
        \includegraphics[width=16.75px]{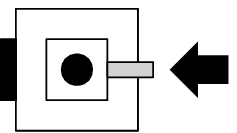}
        \caption{\label{fig:Gadget_z_1_to_z_0} {\tt Z1\_to\_z0}}
    \end{subfigure}
    \caption{\label{fig:Transition_gadgets}
        The ``Transition'' gadgets.
        These gadgets are used by the counter to transition after reading/writing one digit to begin reading/writing the next digit, within the same value.
    }
\end{figure}

%
%
\begin{figure}[!ht]
    \centering
    \begin{subfigure}[t]{0.32\textwidth}
        \centering
        \includegraphics[width=24.25px]{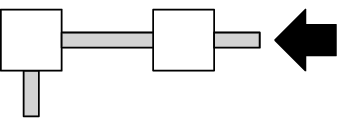}
        \caption{\label{fig:Gadget_return_turn_corner} {\tt Reset\_turn\_corner}}
    \end{subfigure}%
    \begin{subfigure}[t]{0.32\textwidth}
        \centering
        \includegraphics[width=4.5px]{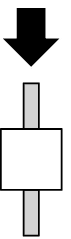}
        \caption{\label{fig:Gadget_return_single_tile} {\tt Reset\_single\_tile}}
    \end{subfigure}%
    \begin{subfigure}[t]{0.32\textwidth}
        \centering
        \includegraphics[width=40.75px]{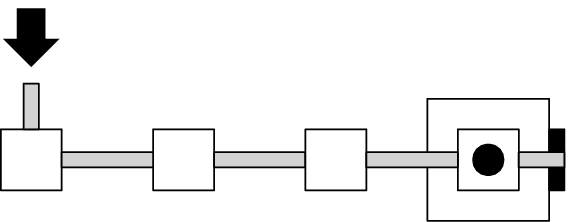}
        \caption{\label{fig:Gadget_return_read_even_digit} {\tt Reset\_read\_even\_digit}}
    \end{subfigure}%
    \caption{\label{fig:Reset_for_next_increment_gadgets}
        The ``Reset'' gadgets.
        These gadgets reset the counter to begin the next increment step. 
    }
\end{figure}

\clearpage

%
%
\begin{figure}[!t]
    \centering
    \begin{subfigure}[t]{0.22\textwidth}
        \centering
        \includegraphics[width=43.0px]{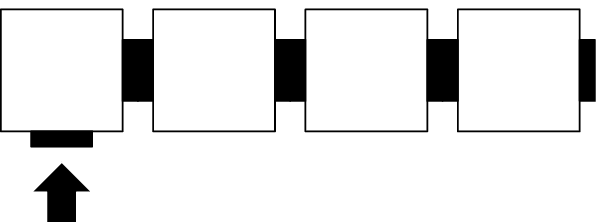}
        \caption{\label{fig:Gadget_at_special_stopper} {\tt At\_special\_stopper}}
    \end{subfigure}%
    \begin{subfigure}[t]{0.28\textwidth}
        \centering
        \includegraphics[width=59.5px]{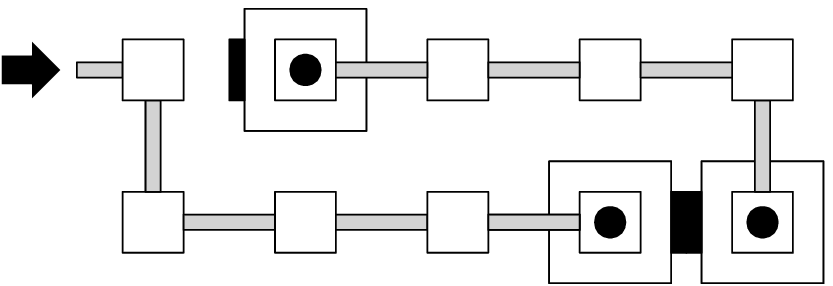}
        \caption{\label{fig:Gadget_special_stopper} {\tt Special\_stopper}}
    \end{subfigure}%
    \begin{subfigure}[t]{0.42\textwidth}
        \centering
        \includegraphics[width=16.75px]{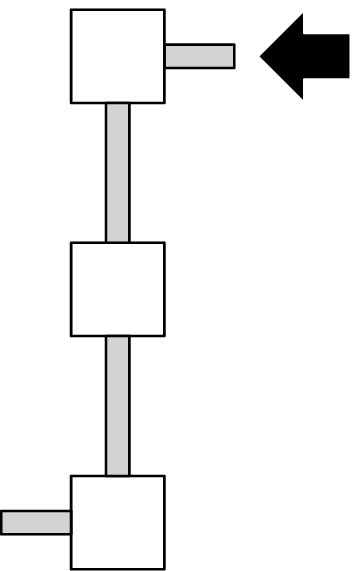}
        \caption{\label{fig:Gadget_special_at_MSB_of_most_significant_digit} {\tt Special\_at\_MSB\_of\_most\_significant\_digit}}
    \end{subfigure}%
    \caption{\label{fig:Special_case_gadgets}
        The ``Special case'' gadgets.
        These gadgets are used only if $k \mod 4 \in \{2, 3\}$, as they are specific to a special case digit region.
    }
\end{figure}

We now create the tile types for our construction. What follows is a list of ``Create'' statements in which specific gadgets are instantiated from the general gadgets in Figures~\ref{fig:Write_gadgets} through 
\ref{fig:Special_case_gadgets}. 

%
%
Create
\begin{align*}
        \texttt{Seed\_start}( & \! \left< {\tt seed\_write\_even\_digit}, 0, 0 \right> )
\end{align*}
from the general gadget shown in Figure~\ref{fig:Gadget_seed_start}.
This step creates the gadget shown in Figure~\ref{fig:Initial_overview_seed_start}.
A single gadget is created by this step.
%

%
%
For each $i = 0, \ldots, \left \lfloor \frac{w}{2} \right \rfloor - 1$, where $i$ ranges over indices of the digit regions,
\begin{itemize}
    \item For each $j = 0, \ldots, m - 2$, where $j$ ranges over the indices of a digit's bits (except for the most significant bit), create
    \begin{align*}
        {\tt Write\_even\_digit\_0}( & \! \left< {\tt seed\_write\_even\_digit}, 2i, j  \right> \!, \\
                                     & \! \left< {\tt seed\_write\_even\_digit}, 2i, j + 1 \right> )
    \end{align*}
    from the general gadget shown in Figure~\ref{fig:Gadget_write_even_digit_0}, if the $j^{th}$ bit of $d_{2i}$ (starting with $j=0$ for the least significant bit) is 0, otherwise create
    \begin{align*}
        {\tt Write\_even\_digit\_1}( & \! \left< {\tt seed\_write\_even\_digit}, 2i, j  \right> \!, \\
                                     & \! \left< {\tt seed\_write\_even\_digit}, 2i, j + 1  \right> )
    \end{align*}
    from the general gadget shown in Figure~\ref{fig:Gadget_write_even_digit_1}.
    This step creates gadgets that correspond to all but the last gadget to self-assemble in Figure~\ref{fig:Initial_overview_write_even_digits}.
    These are the non-most significant bits of the most significant even digit.
    The total number of gadgets created by this step is
    $O(k m)$.

    \item Create
    \begin{align*}
        {\tt Write\_even\_digit\_0}( & \! \left< {\tt seed\_write\_even\_digit},                2i, m - 1 \right> \!, \\
                                     & \! \left< {\tt seed\_even\_digit\_to\_odd\_digit}, 2i                    \right> )
    \end{align*}
    from the general gadget shown in Figure~\ref{fig:Gadget_write_even_digit_0}, if the most significant bit of $d_{2i}$ is 0, otherwise create
    \begin{align*}
        {\tt Write\_even\_digit\_1}( & \! \left< {\tt seed\_write\_even\_digit},                2i, m - 1 \right> \!, \\
                                     & \! \left< {\tt seed\_even\_digit\_to\_odd\_digit}, 2i                    \right> )
    \end{align*}
    from the general gadget shown in Figure~\ref{fig:Gadget_write_even_digit_1}.
    This step creates a gadget that corresponds to the last gadget to self-assemble in Figure~\ref{fig:Initial_overview_write_even_digits}.
    These are the most significant bits of the even digits.
    The total number of gadgets created by this step is
    $O(k)$.
\end{itemize}

%
%
For each $i = 0, \ldots, \left \lfloor \frac{w}{2} \right \rfloor - 1$, create
\begin{align*}
    {\tt Seed\_even\_digit\_to\_odd\_digit}( & \! \left< {\tt seed\_even\_digit\_to\_odd\_digit}, 2i                      \right> \!, \\
                                             & \! \left< {\tt seed\_write\_odd\_digit},                 2i + 1,  -1 \right> )
\end{align*}
from the general gadget shown in Figure~\ref{fig:Gadget_seed_even_digit_to_odd_digit}.
This step creates gadgets that correspond to the gadget shown in Figure~\ref{fig:Initial_overview_seed_even_digit_to_odd_digit}.
The total number of gadgets created by this step is
$O(k)$.
%

%
%
For each $i = 0, \ldots, \left \lfloor \frac{w}{2} \right \rfloor - 2$, create
\begin{align*}
    {\tt Write\_odd\_digit\_0}( & \! \left< {\tt seed\_write\_odd\_digit}, 2i + 1, -1 \right> \!, \\
                                & \! \left< {\tt seed\_write\_odd\_digit}, 2i + 1, 0 \right> )
\end{align*}
from the general gadget shown in Figure~\ref{fig:Gadget_write_odd_digit_0}.
This step creates gadgets that correspond to the first gadget to self-assemble in Figure~\ref{fig:Initial_overview_write_odd_digits}.
These are the indicator bits for the non-most significant odd digits.
The total number of gadgets created by this step is
$O(k)$.
%

%
%
Create
\begin{align*}
    {\tt Write\_odd\_digit\_1} & \Bigg( \Bigg. \! \! \left< {\tt seed\_write\_odd\_digit}, 2\left\lfloor \frac{w}{2} \right \rfloor -1,  -1 \right>  \!, \\
                               & \;                  \left< {\tt seed\_write\_odd\_digit}, 2\left\lfloor \frac{w}{2} \right \rfloor -1,   0 \right> \! \! \Bigg. \Bigg)
\end{align*}
from the general gadget shown in Figure~\ref{fig:Gadget_write_odd_digit_1}.
This step creates the gadget that corresponds to the first gadget to self-assemble in Figure~\ref{fig:Initial_overview_write_odd_digits}, if the current digit region is the most significant (general) one, or the second most significant digit, if $k \mod 4 \in \{2,3\}$.
This is the indicator bit for the most significant odd digit.
A single gadget is created by this step.

%
%
For each $i = 0, \ldots, \left \lfloor \frac{w}{2} \right \rfloor - 1$:
\begin{itemize}
    \item For each $j = 0,\ldots, m - 2$: create
    \begin{align*}
        {\tt Write\_odd\_digit\_0}( & \! \left< {\tt seed\_write\_odd\_digit}, 2i + 1, j    \right> \!, \\
                                    & \! \left< {\tt seed\_write\_odd\_digit}, 2i + 1, j + 1 \right> )
    \end{align*}
    from the general gadget shown in Figure~\ref{fig:Gadget_write_odd_digit_0}, if the $j^{th}$ bit of $d_{2i + 1}$ (starting with $j=0$ for the least significant bit) is 0, otherwise create
    \begin{align*}
        {\tt Write\_odd\_digit\_1}( & \! \left< {\tt seed\_write\_odd\_digit}, 2i + 1, j     \right> \!, \\
                                    & \! \left< {\tt seed\_write\_odd\_digit}, 2i + 1, j + 1 \right> )
    \end{align*}
    from the general gadget shown in Figure~\ref{fig:Gadget_write_odd_digit_1}.
    This step creates gadgets that correspond to all but the first and last gadgets to self-assemble in Figure~\ref{fig:Initial_overview_write_odd_digits}.
    These are the non-most significant bits of the odd digits.
    The total number of gadgets created by this step is
    $O(k m)$.
    \item Create
    \begin{align*}
        {\tt Write\_odd\_digit\_0}( & \! \left< {\tt seed\_write\_odd\_digit}, 2i + 1, m - 1 \right> \!, \\
                                    & \! \left< {\tt seed\_single\_tile\_0},      2i + 1, 0  \right> )
    \end{align*}
    from the general gadget shown in Figure~\ref{fig:Gadget_write_odd_digit_0}, if the most significant bit of $d_{2i + 1}$ is 0, otherwise create
    \begin{align*}
        {\tt Write\_odd\_digit\_1}( & \! \left< {\tt seed\_write\_odd\_digit}, 2i + 1, m - 1 \right> \!, \\
                                    & \! \left< {\tt seed\_single\_tile\_0},      2i + 1, 0 \right> )
    \end{align*}
    from the general gadget shown in Figure~\ref{fig:Gadget_write_odd_digit_1}.
    This step creates a gadget that corresponds to the last gadget to self-assemble in Figure~\ref{fig:Initial_overview_write_odd_digits}.
    These are the most significant bits of the odd digits.
    The total number of gadgets created by this step is
    $O(k)$.
\end{itemize}

%
%
For each $i = 0, \ldots, \left \lfloor \frac{w}{2} \right \rfloor - 1$:
\begin{itemize}
    \item For each $j = 0, \ldots,\left(\left(8 + 3m + 1\right) - (2 + 4)\right) - 2$, create
    \begin{align*}
        {\tt Single\_tile}( & \! \left< {\tt seed\_single\_tile\_0}, 2i + 1, j   \right> \!, \\
                            & \! \left< {\tt seed\_single\_tile\_0}, 2i + 1, j + 1  \right> )
    \end{align*}
    from the general gadget shown in Figure~\ref{fig:Gadget_single_tile}.
    This step creates gadgets that correspond to all but the last gadget to self-assemble in Figure~\ref{fig:Initial_overview_single_tile_0}.
    The total number of gadgets created by this step is
    $O(k m)$.
    \item Create
    \begin{align*}
        {\tt Single\_tile}( & \! \left< {\tt seed\_single\_tile\_0},            2i + 1, \left(\left(8 + 3m + 1\right) - (2 + 4)\right) - 1 \right> \!, \\
                            & \! \left< {\tt seed\_stopper\_after\_odd\_digit}, 2i + 1                                                     \right> )
    \end{align*}
    from the general gadget shown in Figure~\ref{fig:Gadget_single_tile}.
    This step creates a gadget that corresponds to the last gadget to self-assemble in Figure~\ref{fig:Initial_overview_single_tile_0}.
    The total number of gadgets created by this step is
    $O(k)$.
    \item Create
    \begin{align*}
        {\tt Stopper\_after\_odd\_digit}( & \! \left< {\tt seed\_stopper\_after\_odd\_digit}, 2i + 1  \right> \!, \\
                                          & \! \left< {\tt seed\_single\_tile\_opposite\_0},  2i + 1, 0  \right> )
    \end{align*}
    from the general gadget shown in Figure~\ref{fig:Gadget_stopper_after_odd_digit}.
    This step creates a gadget that corresponds to the gadget shown in Figure~\ref{fig:Initial_overview_stopper_after_odd_digit}.
    The total number of gadgets created by this step is
    $O(k)$.
    \item For each $j = 0, \ldots, \left(6m + 9\right) - 2$, create
    \begin{align*}
        {\tt Single\_tile\_opposite}( & \! \left< {\tt seed\_single\_tile\_opposite\_0}, 2i + 1, j  \right> \!, \\
                                      & \! \left< {\tt seed\_single\_tile\_opposite\_0}, 2i + 1, j + 1  \right> )
    \end{align*}
    from the general gadget shown in Figure~\ref{fig:Gadget_single_tile_opposite}.
    This step creates gadgets that correspond to all but the last gadget to self-assemble in Figure~\ref{fig:Initial_overview_single_tile_opposite_0}.
    The total number of gadgets created by this step is
    $O(k m)$.
    \item Create
    \begin{align*}
        {\tt Single\_tile\_opposite}( & \! \left< {\tt seed\_single\_tile\_opposite\_0}, 2i + 1, \left(6 m + 9\right) - 1  \right> \!, \\
                                      & \! \left< {\tt seed\_between\_digits},           2i + 1                                       \right> )
    \end{align*}
    from the general gadget shown in Figure~\ref{fig:Gadget_single_tile_opposite}.
    This step creates a gadget that corresponds to the last gadget to self-assemble in Figure~\ref{fig:Initial_overview_single_tile_opposite_0}.
    The total number of gadgets created by this step is
    $O(k)$.
    \item Create
    \begin{align*}
        {\tt Between\_digits}( & \! \left< {\tt seed\_between\_digits},           2i + 1                \right> \!, \\
                               & \! \left< {\tt seed\_single\_tile\_opposite\_1}, 2i + 1, 0 \right> )
    \end{align*}
    from the general gadget shown in Figure~\ref{fig:Gadget_between_digits}.
    This step create a gadget that corresponds to the gadget shown in Figure~\ref{fig:Initial_overview_between_digits}.
    The total number of gadgets created by this step is
    $O(k)$.
    \item For each $j = 0, \ldots, \left(3m + 6\right) - 2$, create
    \begin{align*}
        {\tt Single\_tile\_opposite}( & \! \left< {\tt seed\_single\_tile\_opposite\_1}, 2i + 1, j  \right> \!, \\
                                      & \! \left< {\tt seed\_single\_tile\_opposite\_1}, 2i + 1, j + 1  \right> )
    \end{align*}
    from the general gadget shown in Figure~\ref{fig:Gadget_single_tile_opposite}.
    This step creates gadgets that correspond to all but the last gadget to self-assemble in Figure~\ref{fig:Initial_overview_single_tile_opposite_1}.
    The total number of gadgets created by this step is
    $O(k m)$.
\end{itemize}


%
%

For each $i = 0, \ldots, \left\lfloor \frac{w}{2} \right \rfloor - 2$:
\begin{itemize}
    \item Create
    \begin{align*}
        {\tt Single\_tile\_opposite}( & \! \left< {\tt seed\_single\_tile\_opposite\_1},                  2i + 1, \left(3m + 6 \right) - 1   \right> \!, \\
                                      & \! \left< {\tt seed\_to\_next\_significant\_digit\_region}, 2i + 1                                        \right> )
    \end{align*}
    from the general gadget shown in Figure~\ref{fig:Gadget_single_tile_opposite}.
    This step creates a gadget that corresponds to the last gadget to self-assemble in Figure~\ref{fig:Initial_overview_single_tile_opposite_1}.
    The total number of gadgets created by this step is
    $O(k)$.
    \item Create
    \begin{align*}
        {\tt Seed\_to\_next\_significant\_digit\_region}( & \! \left< {\tt seed\_to\_next\_significant\_digit\_region}, 2i + 1      \right> \!, \\
                                                          & \! \left< {\tt seed\_write\_even\_digit},                         2(i + 1), 0 \right> )
    \end{align*}
    from the general gadget shown in Figure~\ref{fig:Gadget_seed_to_next_significant_digit_region}.
    This step creates a gadget that corresponds to the gadget shown in Figure~\ref{fig:Initial_overview_seed_to_next_significant_digit_region}.
    The total number of gadgets created by this step is
    $O(k)$.
\end{itemize}

%
%
Create
\begin{align*}
    {\tt Single\_tile\_opposite}( & \! \left< {\tt seed\_single\_tile\_opposite\_1}, w - 1, \left(3m + 6 \right) - 1\right> \!, \\
                                  & \! \left< {\tt reset\_turn\_corner}                                                   \right> )
\end{align*}
from the general gadget shown in Figure~\ref{fig:Gadget_single_tile_opposite}.
This step creates the gadget that corresponds to the gadget from which the gadget shown in Figure~\ref{fig:General_overview_return_turn_corner} self-assembles. A single gadget is created by this step.
%

%
%
If $k \mod 4 \in \{0,1\}$, create
\begin{align*}
    {\tt Reset\_turn\_corner}( & \! \left< {\tt reset\_turn\_corner}    \right> \!, \\
                                & \! \left< {\tt reset\_single\_tile}, 0 \right> )
\end{align*}
from the general gadget shown in Figure~\ref{fig:Gadget_return_turn_corner}. 
This step creates the gadget that corresponds to the gadget shown in Figure~\ref{fig:General_overview_return_turn_corner}. This step is conditional because we create a special {\tt Reset\_turn\_corner} gadget when $k \mod 4 \in \{2,3\}$.  A single gadget is created by this step.
If $k \mod 4 \in \{0,1\}$, create
\begin{align*}
    {\tt Reset\_single\_tile}( & \! \left< {\tt reset\_single\_tile}, 0    \right> \!, \\
                                & \! \left< {\tt reset\_single\_tile}, 1 \right> )
\end{align*}
from the general gadget shown in Figure~\ref{fig:Gadget_return_single_tile}. 
This step creates the gadget that corresponds to the first gadget to self-assemble in Figure~\ref{fig:General_overview_return_single_tile}. This step is conditional because, when $k \mod 4 \in \{2,3\}$, this gadget is not used. A single gadget is created by this step.
For each $j = 1, \ldots, (k - 2) - 2$, create
\begin{align*}
    {\tt Reset\_single\_tile}( & \! \left< {\tt reset\_single\_tile}, j     \right> \!, \\
                                & \! \left< {\tt reset\_single\_tile}, j + 1 \right> )
\end{align*}
from the general gadget shown in Figure~\ref{fig:Gadget_return_single_tile}.
This step creates gadgets that correspond to all but the last gadget to self-assemble in Figure~\ref{fig:General_overview_return_single_tile}.
The total number of gadgets created by this step is
$O(k)$.
Create
\begin{align*}
    {\tt Reset\_single\_tile}( & \! \left< {\tt reset\_single\_tile}, k - 3 \right> \!, \\
                                & \! \left< {\tt reset\_read\_even\_digit}   \right> )
\end{align*}
from the general gadget shown in Figure~\ref{fig:Gadget_return_single_tile}.
This step creates the gadget that corresponds to the last gadget to self-assemble in Figure~\ref{fig:General_overview_return_single_tile}.
A single gadget is created by this step.
If $M = 2$, create
\begin{align*}
    {\tt Reset\_read\_even\_digit}( & \! \left< {\tt reset\_read\_even\_digit} \right> \!, \\
                                     & \! \left< {\tt read\_MSB}, 1,   c, p      \right> \!, \\
                                     & \! \left< {\tt read\_MSB}, 0,   c, p      \right> )
\end{align*}
from the general gadget shown in Figure~\ref{fig:Gadget_return_read_even_digit}, otherwise create
\begin{align*}
    {\tt Reset\_read\_even\_digit}( & \! \left< {\tt reset\_read\_even\_digit} \right> \!, \\
                                     & \! \left< {\tt read\_non\_MSB}, 1,   c, p \right> \!, \\
                                     & \! \left< {\tt read\_non\_MSB}, 0,   c, p \right> )
\end{align*}
from the general gadget shown in Figure~\ref{fig:Gadget_return_read_even_digit}, where $c = 1$ is a value indicating that there is an incoming arithmetic carry and $p=0$ is the parity of the digit being read.
This step creates a gadget that corresponds to the gadget shown in Figure~\ref{fig:General_overview_return_read_even_digit}.
The total number of gadgets created by this step is
$O \left( 1 \right)$.
We will now create the gadgets that self-assemble in a general digit region.
%

%
%
For each $c \in \{0, 1\}$:
\begin{itemize}
    \item For each $x \in \{0, 1\}^{i}$, for $1 \leq i < m - 1$, create
    \begin{align*}
        {\tt Read\_non\_MSB\_0}( & \! \left< {\tt read\_non\_MSB},  x, c, 0 \right> \!, \\
                                 & \! \left< {\tt read\_non\_MSB}, x1, c, 0 \right> \!, \\ 
                                 & \! \left< {\tt read\_non\_MSB}, x0, c, 0 \right> )      
    \end{align*}
    from the general gadget shown in Figure~\ref{fig:Gadget_read_non_MSB_0} if $x$ ends with $0$, otherwise create
    \begin{align*}
        {\tt Read\_non\_MSB\_1}( & \! \left< {\tt read\_non\_MSB},  x, c, 0 \right> \!, \\
                                 & \! \left< {\tt read\_non\_MSB}, x1, c, 0 \right> \!, \\ 
                                 & \! \left< {\tt read\_non\_MSB}, x0, c, 0 \right> )      
    \end{align*}
    from the general gadget shown in Figure~\ref{fig:Gadget_read_non_MSB_1}.
    This step creates gadgets that correspond to all but the last gadget to self-assemble in Figure~\ref{fig:General_overview_read_non_MSB}.
    The total number of gadgets created by this step is
    $O(M)$.
    Note that our geometric scheme used for the digits (both even and odd) positions the bits in Little-Endian order, i.e., with the least significant bit to the left.
    So, once a digit has been completely read by {\tt Read\_non\_MSB} and {\tt Read\_MSB} gadgets, since each bit is appended to the right of the bits that were already read, the end result is a binary string that preserves the original order of the bits, i.e., the bits stay in Little-Endian.
    %

    \item For each $x \in \{0, 1\}^{m - 1}$, create
    \begin{align*}
        {\tt Read\_non\_MSB\_0}( & \! \left< {\tt read\_non\_MSB}, x,  c, 0 \right> \!, \\
                                 & \! \left< {\tt read\_MSB},     x1,  c, 0 \right> \!, \\ 
                                 & \! \left< {\tt read\_MSB},     x0,  c, 0 \right> )      
    \end{align*}
    from the general gadget shown in Figure~\ref{fig:Gadget_read_non_MSB_0} if $x$ ends with $0$, otherwise create
    \begin{align*}
        {\tt Read\_non\_MSB\_1}( & \! \left< {\tt read\_non\_MSB}, x, c, 0 \right> \!, \\
                                 & \! \left< {\tt read\_MSB},     x1, c, 0 \right> \!, \\ 
                                 & \! \left< {\tt read\_MSB},     x0, c, 0 \right> )      
    \end{align*}
    from the general gadget shown in Figure~\ref{fig:Gadget_read_non_MSB_1}.
    This step creates gadgets that correspond to the last gadget to self-assemble in Figure~\ref{fig:General_overview_read_non_MSB}.
    The total number of gadgets created by this step is
    $O(M)$.
    %

    \item For each $x \in \{0, 1\}^{m}$, create
    \begin{align*}
        {\tt Read\_MSB\_0}( & \! \left< {\tt read\_MSB},                     x, c, 0 \right> \!, \\
                            & \! \left< {\tt repeating\_after\_even\_digit}, x, c    \right> )
    \end{align*}
    from the general gadget shown in Figure~\ref{fig:Gadget_read_MSB_0} if $x$ ends with $0$, otherwise create
    \begin{align*}
        {\tt Read\_MSB\_1}( & \! \left< {\tt read\_MSB},                     x, c, 0 \right> \!, \\
                            & \! \left< {\tt repeating\_after\_even\_digit}, x, c    \right> )
    \end{align*}
    from the general gadget shown in Figure~\ref{fig:Gadget_read_MSB_1}.
    This step creates gadgets that correspond to the gadget shown in Figure~\ref{fig:General_overview_read_MSB}.
    The total number of gadgets created by this step is
    $O(M)$.
\end{itemize}

For each $c \in \{0, 1\}$:
\begin{itemize}
    \item For each $x \in \{0, 1\}^{m}$, create
    \begin{align*}
        {\tt Repeating\_after\_even\_digit}( & \! \left< {\tt repeating\_after\_even\_digit},  x, c \right> \!, \\
                                             & \! \left< {\tt at\_stopper\_after\_odd\_digit}, x, c \right> \!, \\ 
                                             & \! \left< {\tt repeating\_after\_even\_digit},  x, c \right> )      
    \end{align*}
    from the general gadget shown in Figure~\ref{fig:Gadget_repeating_after_even_digit}.
    This step creates gadgets that correspond to the gadget shown in Figure~\ref{fig:General_overview_repeating_after_even_digit}.
    The total number of gadgets created by this step is
    $O(M)$.
\end{itemize}

%
%
For each $x \in \{0, 1\}^{m}$:
\begin{itemize}
    \item Create
    \begin{align*}
        {\tt At\_stopper\_after\_odd\_digit}( & \! \left< {\tt at\_stopper\_after\_odd\_digit}, x, 0 \right> \!, \\
                                              & \! \left< {\tt write\_even\_digit},             x, 0 \right> )
    \end{align*}
    from the general gadget shown in Figure~\ref{fig:Gadget_at_stopper_after_odd_digit}. Note that the last argument in the encodings for the input and output glues corresponds to the value of $c$ from the previous {\tt Repeating\_after\_even\_digit} gadget. 
    This step creates a gadget that corresponds to the gadget shown in Figure~\ref{fig:General_overview_at_stopper_after_odd_digit}.
    The total number of gadgets created by this step is
    $O(M)$.
    \item When $(x + 1) \mod M = 0$, create
    \begin{align*}
        {\tt At\_stopper\_after\_odd\_digit}( & \! \left< {\tt at\_stopper\_after\_odd\_digit}, x, 1 \right> \!, \\
                                              & \! \left< {\tt write\_even\_digit},         0^{m}, 1 \right> )
    \end{align*}
    from the general gadget shown in Figure~\ref{fig:Gadget_at_stopper_after_odd_digit}.
    Otherwise, create
    \begin{align*}
        {\tt At\_stopper\_after\_odd\_digit}( & \! \left< {\tt at\_stopper\_after\_odd\_digit},  x, 1 \right> \!, \\
                                              & \! \left< {\tt write\_even\_digit},              z, 0 \right> )
    \end{align*}
    from the general gadget shown in Figure~\ref{fig:Gadget_at_stopper_after_odd_digit}, where $z \in \{0, 1\}^m$ is the zero-padded binary representation of the value $x + 1$.
    This step creates a gadget that corresponds to the gadget shown in Figure~\ref{fig:General_overview_at_stopper_after_odd_digit}.
    The total number of gadgets created by this step is
    $O(M)$.
\end{itemize}

%
%
For each $c \in \{0, 1\}$:
\begin{itemize}

    %
    %
    \item For each $x \in \{0, 1\}^{i}$, for $1 \leq i < m$, create
    \begin{align*}
        {\tt Write\_even\_digit\_0}( & \! \left< {\tt write\_even\_digit}, 0x, c \right> \!, \\
                                     & \! \left< {\tt write\_even\_digit},  x, c \right> )
    \end{align*}
    from the general gadget shown in Figure~\ref{fig:Gadget_write_even_digit_0} and create
    \begin{align*}
        {\tt Write\_even\_digit\_1}( & \! \left< {\tt write\_even\_digit}, 1x, c \right> \!, \\
                                     & \! \left< {\tt write\_even\_digit},  x, c \right> )
    \end{align*}
    from the general gadget shown in Figure~\ref{fig:Gadget_write_even_digit_1}.
    This step creates gadgets that correspond to all but the last gadget to self-assemble in Figure~\ref{fig:General_overview_write_even_digit}.
    The total number of gadgets created by this step is
    $O(M)$.
    %

    %
    %
    \item Create
    \begin{align*}
        {\tt Write\_even\_digit\_0}( & \! \left< {\tt write\_even\_digit},        0, c \right> \!, \\
                                     & \! \left< {\tt stopper\_after\_even\_digit},  c \right> )
    \end{align*}
    from the general gadget shown in Figure~\ref{fig:Gadget_write_even_digit_0}.
    This step creates a gadget that corresponds to the last gadget to self-assemble in Figure~\ref{fig:General_overview_write_even_digit}, which is the most significant bit.
    The total number of gadgets created by this step is
    $O \left( 1 \right)$.
    %

    %
    %
    \item Create \begin{align*}
        {\tt Write\_even\_digit\_1}( & \! \left< {\tt write\_even\_digit},       1, c \right> \!, \\
                                     & \! \left< {\tt stopper\_after\_even\_digit}, c \right> )
    \end{align*}
    from the general gadget shown in Figure~\ref{fig:Gadget_write_even_digit_1}.
    This step creates a gadget that corresponds to the last gadget to self-assemble in Figure~\ref{fig:General_overview_write_even_digit}, which is the most significant bit.
    The total number of gadgets created by this step is
    $O \left( 1 \right)$.
\end{itemize}

%
%
For each $c \in \{0, 1\}$:
\begin{itemize}
    %
    %
    \item Create
    \begin{align*}
        {\tt Stopper\_after\_even\_digit}( & \! \left< {\tt stopper\_after\_even\_digit}, c    \right> \!, \\
                                           & \! \left< {\tt single\_tile\_opposite\_0},   c, 0 \right> )
    \end{align*}
    from the general gadget shown in Figure~\ref{fig:Gadget_stopper_after_even_digit}.
    This step creates gadgets that correspond to the gadget shown in Figure~\ref{fig:General_overview_stopper_after_even_digit}.
    The total number of gadgets created by this step is
    $O \left( 1 \right)$.
    %

    \item For each $i = 0,\dots,(3m + 4) - 2$, create
    \begin{align*}
        {\tt Single\_tile\_opposite}( & \! \left< {\tt single\_tile\_opposite\_0}, c, i     \right> \!, \\
                                      & \! \left< {\tt single\_tile\_opposite\_0}, c, i + 1 \right> )
    \end{align*}
    from the general gadget shown in Figure~\ref{fig:Gadget_single_tile_opposite}.
    This step creates gadgets that correspond to all but the last gadget to self-assemble in Figure~\ref{fig:General_overview_single_tile_opposite_0}.
    The total number of gadgets created by this step is
    $O(m)$.
    %

    \item Create
    \begin{align*}
        {\tt Single\_tile\_opposite}( & \! \left< {\tt single\_tile\_opposite\_0}, c, (3m + 4) - 1 \right> \!, \\
                                      & \! \left< {\tt between\_digit\_regions},   c               \right> )
    \end{align*}
    from the general gadget shown in Figure~\ref{fig:Gadget_single_tile_opposite}.
    This step creates a gadget that corresponds to the last gadget to self-assemble in Figure~\ref{fig:General_overview_single_tile_opposite_0}.
    The total number of gadgets created by this step is
    $O \left( 1 \right)$.
    %

    \item Create
    \begin{align*}
        {\tt Between\_digit\_regions}( & \! \left< {\tt between\_digit\_regions},   c    \right> \!, \\
                                       & \! \left< {\tt single\_tile\_opposite\_1}, c, 0 \right> )
    \end{align*}
    from the general gadget shown in Figure~\ref{fig:Gadget_between_digit_regions}.
    This step creates a gadget that corresponds to the gadget shown in Figure~\ref{fig:General_overview_between_digit_regions}.
    The total number of gadgets created by this step is
    $O \left( 1 \right)$.
    %

    \item For each $i = 0,\dots,(3m + 7) - 2$, create
    \begin{align*}
        {\tt Single\_tile\_opposite}( & \! \left< {\tt single\_tile\_opposite\_1}, c, i     \right> \!, \\
                                      & \! \left< {\tt single\_tile\_opposite\_1}, c, i + 1 \right> )
    \end{align*}
    from the general gadget shown in Figure~\ref{fig:Gadget_single_tile_opposite}.
    This step creates gadgets that correspond to all but the last gadget to self-assemble in Figure~\ref{fig:General_overview_single_tile_opposite_1}.
    The total number of gadgets created by this step is
    $O(m)$.
    %

    \item Create
    \begin{align*}
        {\tt Single\_tile\_opposite}( & \! \left< {\tt single\_tile\_opposite\_1}, c, (3m + 7) - 1 \right> \!, \\
                                      & \! \left< {\tt at\_MSB\_of\_odd\_digit},   c               \right> )
    \end{align*}
    from the general gadget shown in Figure~\ref{fig:Gadget_single_tile_opposite}.
    This step creates a gadget that corresponds to the last gadget to self-assemble in Figure~\ref{fig:General_overview_single_tile_opposite_1}.
    The total number of gadgets created by this step is
    $O \left( 1 \right)$.
    %

    \item Create
    \begin{align*}
        {\tt At\_MSB\_of\_odd\_digit}( & \! \left< {\tt at\_MSB\_of\_odd\_digit},   c    \right> \!, \\
                                       & \! \left< {\tt single\_tile\_opposite\_2}, c, 0 \right> )
    \end{align*}
    from the general gadget shown in Figure~\ref{fig:Gadget_at_MSB_of_odd_digit}.
    This step creates a gadget that corresponds to the gadget  in Figure~\ref{fig:General_overview_at_MSB_of_odd_digit}
    The total number of gadgets created by this step is
    $O \left( 1 \right)$.
    %

    \item For each $i = 0,\dots,3(m + 1) - 2$, create
    \begin{align*}
        {\tt Single\_tile\_opposite}( & \! \left< {\tt single\_tile\_opposite\_2},  c, i     \right> \!, \\
                                      & \! \left< {\tt single\_tile\_opposite\_2},  c, i + 1 \right> )
    \end{align*}
    from the general gadget shown in Figure~\ref{fig:Gadget_single_tile_opposite}.
    This step creates gadgets that correspond to all but the last gadget to self-assemble in Figure~\ref{fig:General_overview_single_tile_opposite_2}.
    The total number of gadgets created by this step is
    $O(m)$.
    %

    \item Create
    \begin{align*}
        {\tt Single\_tile\_opposite}( & \! \left< {\tt single\_tile\_opposite\_2},  c, 3(m + 1) - 1 \right> \!, \\
                                      & \! \left< {\tt start\_read\_odd\_digit},    c               \right> )
    \end{align*}
    from the general gadget shown in Figure~\ref{fig:Gadget_single_tile_opposite}.
    This step creates a gadget that corresponds to the last gadget to self-assemble in Figure~\ref{fig:General_overview_single_tile_opposite_2}.
    The total number of gadgets created by this step is $O \left( 1 \right)$.
\end{itemize}

%
%
For each $c \in \{0, 1\}$:
\begin{itemize}
    \item Create
    \begin{align*}
        {\tt Start\_read\_odd\_digit}( & \! \left< {\tt start\_read\_odd\_digit}, c    \right> \!, \\
                                       & \! \left< {\tt read\_non\_MSB},       1, c, 1 \right> \!, \\
                                       & \! \left< {\tt read\_non\_MSB},       0, c, 1 \right> )
    \end{align*}
    from the general gadget shown in Figure~\ref{fig:Gadget_start_read_odd_digit}.
    This step creates a gadget that corresponds to the gadget shown in Figure~\ref{fig:General_overview_start_read_odd_digit}.
    The total number of gadgets created by this step is
    $O \left( 1 \right)$.
    %

    \item For each $x \in \{0, 1\}^{i}$, for $1 \leq i < m$, create
    \begin{align*}
        {\tt Read\_non\_MSB\_0}( & \! \left< {\tt read\_non\_MSB},  x, c, 1 \right> \!, \\
                                 & \! \left< {\tt read\_non\_MSB}, x1, c, 1 \right> \!, \\ 
                                 & \! \left< {\tt read\_non\_MSB}, x0, c, 1 \right> )      
    \end{align*}
    from the general gadget shown in Figure~\ref{fig:Gadget_read_non_MSB_0} if $x$ ends with $0$, otherwise create
    \begin{align*}
        {\tt Read\_non\_MSB\_1}( & \! \left< {\tt read\_non\_MSB},  x, c, 1 \right> \!, \\
                                 & \! \left< {\tt read\_non\_MSB}, x1, c, 1 \right> \!, \\ 
                                 & \! \left< {\tt read\_non\_MSB}, x0, c, 1 \right> )      
    \end{align*} from the general gadget shown in Figure~\ref{fig:Gadget_read_non_MSB_1}.
    This step creates gadgets that correspond to gadgets that are similar to all but the last gadget to self-assemble in Figure~\ref{fig:General_overview_read_non_MSB}, but the gadgets being created here are for odd digits.
    The total number of gadgets created by this step is
    $O(M)$.
    \item For each $x \in \{0, 1\}^{m}$, create
    \begin{align*}
        {\tt Read\_non\_MSB\_0}( & \! \left< {\tt read\_non\_MSB}, x,  c, 1 \right> \!, \\
                                 & \! \left< {\tt read\_MSB},     x1,  c, 1 \right> \!, \\ 
                                 & \! \left< {\tt read\_MSB},     x0,  c, 1 \right> )      
    \end{align*}
    from the general gadget shown in Figure~\ref{fig:Gadget_read_non_MSB_0} if $x$ ends with $0$, otherwise create
    \begin{align*}
        {\tt Read\_non\_MSB\_1}( & \! \left< {\tt read\_non\_MSB}, x,  c, 1 \right> \!, \\
                                 & \! \left< {\tt read\_MSB},     x1,  c, 1 \right> \!, \\ 
                                 & \! \left< {\tt read\_MSB},     x0,  c, 1 \right> )      
    \end{align*}
    from the general gadget shown in Figure~\ref{fig:Gadget_read_non_MSB_1}.
    This step creates gadgets that correspond to gadgets that are similar to the last gadget to self-assemble in Figure~\ref{fig:General_overview_read_non_MSB}, but the gadgets being created here are for odd digits.
    The total number of gadgets created by this step is
    $O(M)$.
    %

    \item For each $x \in \{0, 1\}^{m + 1}$, create
    \begin{align*}
        {\tt Read\_MSB\_0}( & \! \left< {\tt read\_MSB},                    x, c, 1 \right> \!, \\
                            & \! \left< {\tt repeating\_after\_odd\_digit}, x, c    \right> )
    \end{align*}
    from the general gadget shown in Figure~\ref{fig:Gadget_read_MSB_0} if $x$ ends with $0$, otherwise create
    \begin{align*}
        {\tt Read\_MSB\_1}( & \! \left< {\tt read\_MSB},                    x, c, 1 \right> \!, \\
                            & \! \left< {\tt repeating\_after\_odd\_digit}, x, c    \right> )
    \end{align*}
    from the general gadget shown in Figure~\ref{fig:Gadget_read_MSB_1}.
    This step creates gadgets that correspond to gadgets that are similar to the gadget shown in Figure~\ref{fig:General_overview_read_MSB}, but the gadgets being created here are for odd digits.
    The total number of gadgets created by this step is
    $O(M)$.
\end{itemize}

For each $c \in \{0, 1\}$:
\begin{itemize}
    \item For each $x \in \{0, 1\}^{m + 1}$, create
    \begin{align*}
        {\tt Repeating\_after\_odd\_digit}( & \! \left< {\tt repeating\_after\_odd\_digit},    x, c \right> \!, \\
                                            & \! \left< {\tt repeating\_after\_odd\_digit},    x, c \right> \!, \\ 
                                            & \! \left< {\tt at\_stopper\_after\_even\_digit}, x, c \right> )      
    \end{align*}
    from the general gadget shown in Figure~\ref{fig:Gadget_repeating_after_odd_digit}.
    This step creates gadgets that correspond to the gadgets shown in Figure~\ref{fig:General_overview_repeating_after_odd_digit}.
    The total number of gadgets created by this step is
    $O(M)$.
\end{itemize}

%
%
For each $x \in \{0, 1\}^{m}$ and each $b \in \{0, 1\}$, where $b$ corresponds to the indicator bit of an odd digit:
\begin{itemize}
    \item Create
    \begin{align*}
        {\tt At\_stopper\_after\_even\_digit}( & \! \left< {\tt at\_stopper\_after\_even\_digit}, bx, 0 \right> \!, \\
                                               & \! \left< {\tt write\_odd\_digit},               bx, 0 \right> )
    \end{align*}
    from the general gadget shown in Figure~\ref{fig:Gadget_at_stopper_after_even_digit}. Note that the last argument in the encodings for the input and output glues corresponds to the value of $c$ from the previous {\tt Repeating\_after\_odd\_digit} gadget. 
    This step creates a gadget that corresponds to the gadget shown in Figure~\ref{fig:General_overview_at_stopper_after_even_digit}.
    The total number of gadgets created by this step is
    $O(M)$.
    %

    \item When $(x + 1) \mod M = 0$, create
    \begin{align*}
        {\tt At\_stopper\_after\_even\_digit}( & \! \left< {\tt at\_stopper\_after\_even\_digit}, bx, 1 \right> \!, \\
                                               & \! \left< {\tt write\_odd\_digit},           b 0^{m}, 1 \right> )
    \end{align*}
    from the general gadget shown in Figure~\ref{fig:Gadget_stopper_after_even_digit}.
    Otherwise, create
    \begin{align*}
        {\tt At\_stopper\_after\_even\_digit}( & \! \left< {\tt at\_stopper\_after\_even\_digit}, bx, 1 \right> \!, \\
                                               & \! \left< {\tt write\_odd\_digit},               bz, 0 \right> )
    \end{align*}
    from the general gadget shown in Figure~\ref{fig:Gadget_stopper_after_even_digit}, where $z \in \{0, 1\}^m$ is the zero-padded binary representation of the value $x + 1$.
    This step creates a gadget that corresponds to the gadget shown in Figure~\ref{fig:General_overview_at_stopper_after_even_digit}.
    The total number of gadgets created by this step is
    $O(M)$.
\end{itemize}

%
%
For each $c \in \{0, 1\}$:
\begin{itemize}
    \item For each $x \in \{0, 1\}^{m}$, create
    \begin{align*}
        {\tt Write\_odd\_digit\_0}( & \! \left< {\tt write\_odd\_digit}, 0x, c    \right> \!, \\
                                    & \! \left< {\tt write\_odd\_digit},  x, c, 0 \right> )
    \end{align*}
    from the general gadget shown in Figure~\ref{fig:Gadget_write_odd_digit_0} and create
    \begin{align*}
        {\tt Write\_odd\_digit\_1}( & \! \left< {\tt write\_odd\_digit}, 1x, c    \right> \!, \\
                                    & \! \left< {\tt write\_odd\_digit},  x, c, 1 \right> )
    \end{align*}
    from the general gadget shown in Figure~\ref{fig:Gadget_write_odd_digit_1}.
    This step creates gadgets that correspond to the first gadget to self-assemble in Figure~\ref{fig:General_overview_write_odd_digit}.
    The total number of gadgets created by this step is $O(M)$.
    Here we introduce an additional value to the output glues of these gadgets, indicating whether this digit is the most significant digit.
    We use a $1$ to indicate that it is the most significant digit and a $0$ otherwise.
\end{itemize}

%
%
For each $c \in \{0, 1\}$ and each $d \in \{0, 1\}$, where $d$ is the most significant digit indicator that was introduced in the previous step:
\begin{itemize}
    %
    %
    \item For each $x \in \{0, 1\}^{i}$, for $1 \leq i < m$, create
    \begin{align*}
        {\tt Write\_odd\_digit\_0}( & \! \left< {\tt write\_odd\_digit}, 0x, c, d \right> \!, \\
                                    & \! \left< {\tt write\_odd\_digit},  x, c, d \right> )
    \end{align*}
    from the general gadget shown in Figure~\ref{fig:Gadget_write_odd_digit_0} and create
    \begin{align*}
        {\tt Write\_odd\_digit\_1}( & \! \left< {\tt write\_odd\_digit}, 1x, c, d \right> \!, \\
                                    & \! \left< {\tt write\_odd\_digit},  x, c, d \right> )
    \end{align*}
    from the general gadget shown in Figure~\ref{fig:Gadget_write_odd_digit_1}.
    This step creates gadgets that correspond to all but the first and last gadgets to self-assemble in Figure~\ref{fig:General_overview_write_odd_digit}.
    The total number of gadgets created by this step is $O(M)$.
    %

    %
    %
    \item Create
    \begin{align*}
        {\tt Write\_odd\_digit\_0}( & \! \left< {\tt write\_odd\_digit},  0, c, d    \right> \!, \\
                                    & \! \left< {\tt single\_tile\_0},       c, d, 0 \right> )
    \end{align*}
    from the general gadget shown in Figure~\ref{fig:Gadget_write_odd_digit_0} and create
    \begin{align*}
        {\tt Write\_odd\_digit\_1}( & \! \left< {\tt write\_odd\_digit}, 1, c, d    \right> \!, \\
                                    & \! \left< {\tt single\_tile\_0},      c, d, 0 \right> )
    \end{align*}
    from the general gadget shown in Figure~\ref{fig:Gadget_write_odd_digit_1}.
    This step creates a gadget that corresponds to the last gadget to self-assemble in Figure~\ref{fig:General_overview_write_odd_digit}, which is the most significant bit.
    The total number of gadgets created by this step is $O \left( 1 \right)$.
    %

    %
    %
    %
    %
    %
    %
    %
    %
    %
    %
    %
    %
    %
    %
    %

    \item For each $i = 0,\dots,(3m + 3) - 2$, create
    \begin{align*}
        {\tt Single\_tile}( & \! \left< {\tt single\_tile\_0}, c, d, i     \right> \!, \\
                            & \! \left< {\tt single\_tile\_0}, c, d, i + 1 \right> )
    \end{align*}
    from the general gadget shown in Figure~\ref{fig:Gadget_single_tile}.
    This step creates gadgets that correspond to all but the last gadget to self-assemble in Figure~\ref{fig:General_overview_single_tile}.
    The total number of gadgets created by this step is $O(m)$.
    %

    \item Create
    \begin{align*}
        {\tt Single\_tile}( & \! \left< {\tt single\_tile\_0},            c, d, (3m + 3) - 1 \right> \!, \\
                            & \! \left< {\tt stopper\_after\_odd\_digit}, c, d               \right> )
    \end{align*}
    from the general gadget shown in Figure~\ref{fig:Gadget_single_tile}.
    This step creates a gadget that corresponds to the last gadget to self-assemble in Figure~\ref{fig:General_overview_single_tile}.
    The total number of gadgets created by this step is $O \left( 1 \right)$.
\end{itemize}

For each $c \in \{0, 1\}$ and each $d \in \{0, 1\}$:
\begin{itemize}
    \item Create
    \begin{align*}
        {\tt Stopper\_after\_odd\_digit}( & \! \left< {\tt stopper\_after\_odd\_digit}, c, d    \right> \!, \\
                                          & \! \left< {\tt single\_tile\_opposite\_3},  c, d, 0 \right> )
    \end{align*}
    from the general gadget shown in Figure~\ref{fig:Gadget_stopper_after_odd_digit}.
    This step creates a gadget that corresponds to the gadget shown in Figure~\ref{fig:General_overview_stopper_after_odd_digit}.
    The total number of gadgets created by this step is $O \left( 1 \right)$.
    %

    \item For each $i = 0,\dots,\left(1 + 3\left(m + 1\right) + \left(8 + 3m - 4\right) + 1\right) - 2$, create
    \begin{align*}
        {\tt Single\_tile\_opposite}( & \! \left< {\tt single\_tile\_opposite\_3}, c, d, i     \right> \!, \\
                                      & \! \left< {\tt single\_tile\_opposite\_3}, c, d, i + 1 \right> )
    \end{align*}
    from the general gadget shown in Figure~\ref{fig:Gadget_single_tile_opposite}.
    This step creates gadgets that correspond to all but the last gadget to self-assemble in Figure~\ref{fig:General_overview_single_tile_opposite_3}.
    The total number of gadgets created by this step is $O(m)$.
    %

    \item Create
    \begin{align*}
        {\tt Single\_tile\_opposite}( & \! \left< {\tt single\_tile\_opposite\_3}, c, d, \left(1 + 3\left(m + 1\right) + \left(8 + 3m - 4\right) + 1\right) - 1 \right> \!, \\
                                      & \! \left< {\tt between\_digits},           c, d                                                                         \right> )
    \end{align*}
    from the general gadget shown in Figure~\ref{fig:Gadget_single_tile_opposite}.
    This step creates a gadget that corresponds to the last gadget to self-assemble in Figure~\ref{fig:General_overview_single_tile_opposite_3}.
    The total number of gadgets created by this step is $O \left( 1 \right)$.
    %

    \item Create
    \begin{align*}
        {\tt Between\_digits}( & \! \left< {\tt between\_digits},           c, d    \right> \!, \\
                               & \! \left< {\tt single\_tile\_opposite\_4}, c, d, 0 \right> )
    \end{align*}
    from the general gadget shown in Figure~\ref{fig:Gadget_between_digits}.
    This step creates a gadget that corresponds to the gadget shown in Figure~\ref{fig:General_overview_between_digits}.
    The total number of gadgets created by this step is $O \left( 1 \right)$.
    %

    \item For each $i = 0,\dots,(3m + 6) - 2$, create
    \begin{align*}
        {\tt Single\_tile\_opposite}( & \! \left< {\tt single\_tile\_opposite\_4}, c, d, i     \right> \!, \\
                                      & \! \left< {\tt single\_tile\_opposite\_4}, c, d, i + 1 \right> )
    \end{align*}
    from the general gadget shown in Figure~\ref{fig:Gadget_single_tile_opposite}.
    This step creates gadgets that correspond to all but the last gadget to self-assemble in Figure~\ref{fig:General_overview_single_tile_opposite_4}.
    The total number of gadgets created by this step is $O(m)$.
\end{itemize}

For each $c \in \{0, 1\}$:
\begin{itemize}
    \item Create
    \begin{align*}
        {\tt Single\_tile\_opposite}( & \! \left< {\tt single\_tile\_opposite\_4}, c, 0, (3m + 6) - 1 \right> \!, \\
                                      & \! \left< {\tt z1\_to\_z0},                c                  \right> )
    \end{align*}
    from the general gadget shown in Figure~\ref{fig:Gadget_single_tile_opposite}.
    This step creates a gadget that corresponds to the last gadget to self-assemble in Figure~\ref{fig:General_overview_single_tile_opposite_4}.
    The total number of gadgets created by this step is $O \left( 1 \right)$.
    %


    \item Create
    \begin{align*}
        {\tt Z1\_to\_z0}( & \! \left< {\tt z1\_to\_z0},                c    \right> \!, \\
                          & \! \left< {\tt single\_tile\_opposite\_5}, c, 0 \right> )
    \end{align*}
    from the general gadget shown in Figure~\ref{fig:Gadget_z_1_to_z_0}.
    This step creates a gadget that corresponds to the gadget shown in Figure~\ref{fig:General_overview_z_1_to_z_0}.
    The total number of gadgets created by this step is $O \left( 1 \right)$.
    %

    \item For each $i = 0,\dots,\left(3m + 6 + 3\left(m + 1\right) + 8 + 3m + 1 + 1\right) - 2$, create
    \begin{align*}
        {\tt Single\_tile\_opposite}( & \! \left< {\tt single\_tile\_opposite\_5}, c, i     \right> \!, \\
                                      & \! \left< {\tt single\_tile\_opposite\_5}, c, i + 1 \right> )
    \end{align*}
    from the general gadget shown in Figure~\ref{fig:Gadget_single_tile_opposite}.
    This step creates gadgets that correspond to all but the last gadget to self-assemble in Figure~\ref{fig:General_overview_single_tile_opposite_5}.
    The total number of gadgets created by this step is $O(m)$.
    %

    \item Create
    \begin{align*}
        {\tt Single\_tile\_opposite}( & \! \left< {\tt single\_tile\_opposite\_5}, c, \left(3m + 6 + 3\left(m + 1\right) + 8 + 3m + 1 + 1\right) - 1 \right> \!, \\
                                      & \! \left< {\tt start\_digit\_region},      c                                                                 \right> )
    \end{align*}
    from the general gadget shown in Figure~\ref{fig:Gadget_single_tile_opposite}.
    This step creates a gadget that corresponds to the last gadget to self-assemble in Figure~\ref{fig:General_overview_single_tile_opposite_5}.
    The total number of gadgets created by this step is $O \left( 1 \right)$.
    %

    \item If $M = 2$, create
    \begin{align*}
            {\tt Start\_digit\_region}( & \! \left< {\tt start\_digit\_region}, c    \right> \!, \\
                                        & \! \left< {\tt read\_MSB}, 1,         c, 0 \right> \!, \\
                                        & \! \left< {\tt read\_MSB}, 0,         c, 0 \right> )
    \end{align*}
    from the general gadget shown in Figure~\ref{fig:Gadget_start_digit_region}, otherwise create
    \begin{align*}
        {\tt Start\_digit\_region}( & \! \left< {\tt start\_digit\_region}, c    \right> \!, \\
                                    & \! \left< {\tt read\_non\_MSB}, 1,    c, 0 \right> \!, \\
                                    & \! \left< {\tt read\_non\_MSB}, 0,    c, 0 \right> )
    \end{align*}
    from the general gadget shown in Figure~\ref{fig:Gadget_start_digit_region}.
    This step creates gadgets that correspond to the gadget shown in Figure~\ref{fig:General_overview_start_digit_region}.
    The total number of gadgets created by this step is
    $O \left( 1 \right)$.
\end{itemize}

Here we create the {\tt Single\_tile\_opposite} gadgets that correspond to the last gadget to attach in Figure~\ref{fig:General_overview_single_tile_opposite_4}, and to which a {\tt Reset\_turn\_corner} gadget that corresponds to the gadget shown in Figure~\ref{fig:General_overview_return_turn_corner} attaches. In the gadgets being created here, the value of an incoming arithmetic carry (the second argument in the encoding of the input glue) is 0 and the value of the most significant digit indicator bit (the third argument in the encoding of the input glue) is 1.
If $k \mod 4 \in \{0, 1\}$, create
\begin{align*}
    {\tt Single\_tile\_opposite}( & \! \left< {\tt single\_tile\_opposite\_4}, 0, 1, (3m + 6) - 1 \right> \!, \\
                                  & \! \left< {\tt reset\_turn\_corner}                          \right>  )
\end{align*}
from the general gadget shown in Figure~\ref{fig:Gadget_single_tile_opposite}.
Note that, if $c = 0$, then the counter should start self-assembling back towards the least significant and initiate the next increment step.
This step creates the gadget that corresponds to the last gadget to self-assemble in Figure~\ref{fig:General_overview_single_tile_opposite_4}.
A single gadget was created by this step.
If $k \mod 4 \in \{0, 1\}$, create
\begin{align*}
    {\tt Single\_tile\_opposite}( & \! \left< {\tt single\_tile\_opposite\_4}, 1, 1, (3m + 6) - 1 \right> \!, \\
                                  & \! \left< {\tt purple\_monkey\_dishwasher}                                          \right>  )
\end{align*}
from the general gadget shown in Figure~\ref{fig:Gadget_single_tile_opposite}.
In this case, an arithmetic carry propagated through the most significant digit, which means this gadget will have an output glue that does not match any other input glue, terminating the assembly (or initiating filler tiles).
This step creates the gadget that corresponds to the last gadget to self-assemble in Figure~\ref{fig:General_overview_single_tile_opposite_4}.
A single gadget was created by this step.
If $k \mod 4 \in \{2, 3\}$, then for each $c \in \{0, 1\}$, create
\begin{align*}
    {\tt Single\_tile\_opposite}( & \! \left< {\tt single\_tile\_opposite\_4}, c, 1, (3m + 6) - 1 \right> \!, \\
                                  & \! \left< {\tt special\_z1\_to\_z0},       c                  \right> )
\end{align*}
from the general gadget shown in Figure~\ref{fig:Gadget_single_tile_opposite}.
Since $k \mod 4 \in \{2, 3\}$, this gadget self-assembles after writing the most significant odd digit, with the value of the arithmetic carry, $c$, propagating into the special case digit region in which the most significant digit is contained. 
This step creates a gadget that corresponds to the last gadget to self-assemble in Figure~\ref{fig:General_overview_single_tile_opposite_4}.
The total number of gadgets created by this step is $O \left( 1 \right)$.

The following steps create gadgets for the special case, i.e., in each step it is assumed that $k \mod 4 \in \{2, 3\}$.

Create
\begin{align*}
    {\tt Single\_tile\_opposite}( & \! \left< {\tt seed\_single\_tile\_opposite\_1}, w - 2, (3m + 6) - 1 \right> \!, \\
                                  & \! \left< {\tt seed\_to\_next\_significant\_digit\_region},       w - 2                  \right> )
\end{align*}
from the general gadget shown in Figure~\ref{fig:Gadget_single_tile_opposite}.
This step creates a gadget that corresponds to the last gadget to self-assemble in Figure~\ref{fig:General_overview_single_tile_opposite_4}.
A single gadget is created by this step.
%

Create
\begin{align*}
    {\tt Seed\_to\_next\_digit\_region}( & \! \left< {\tt seed\_to\_next\_significant\_digit\_region}, w - 2                 \right> \!, \\
                                         & \! \left< {\tt seed\_single\_tile\_0},                            w - 1, 0  \right>  )
\end{align*}
from the general gadget shown in Figure~\ref{fig:Gadget_seed_to_next_significant_digit_region}.
This step creates the gadget that corresponds to the gadget shown in Figure~\ref{fig:Special_initial_overview_seed_to_most_significant_digit_region}.
A single gadget is created by this step.
%

For each $i = 0,\ldots,\left(3m + 6 + 3(m + 1) + 8\right)- 2$, create
\begin{align*}
    {\tt Single\_tile}( & \! \left< {\tt seed\_single\_tile\_0}, w - 1, i \right> \!, \\
                        & \! \left< {\tt seed\_single\_tile\_0}, w - 1, i + 1 \right>  )
\end{align*}
from the general gadget shown in Figure~\ref{fig:Gadget_single_tile}.
This step creates gadgets that correspond to all but the last gadget to self-assemble in Figure~\ref{fig:Special_initial_overview_single_tile_0}.
The total number of gadgets created by this step is $O(m)$.
%

Create
\begin{align*}
    {\tt Single\_tile}( & \! \left< {\tt seed\_single\_tile\_0}, w - 1, \left(3m + 6 + 3(m + 1) + 8\right) - 1  \right> \!, \\
                        & \! \left< {\tt seed\_single\_tile\_0}, w - 1, 0                                \right>  )
\end{align*}
from the general gadget shown in Figure~\ref{fig:Gadget_single_tile}.
This step creates the gadget that corresponds to the last gadget to self-assemble in Figure~\ref{fig:Special_initial_overview_single_tile_0}.
A single gadget is created by this step.
%

%
%
For each $j = 0,\ldots, m - 2$, create
\begin{align*}
    {\tt Write\_even\_digit\_0}( & \! \left< {\tt seed\_write\_even\_digit}, w - 1, j  \right> \!, \\
                                 & \! \left< {\tt seed\_write\_even\_digit}, w - 1, j + 1 \right> )
\end{align*}
from the general gadget shown in Figure~\ref{fig:Gadget_write_even_digit_0}, if the $j^{th}$ bit of $d_{w - 1}$ is 0, otherwise create
\begin{align*}
    {\tt Write\_even\_digit\_1}( & \! \left< {\tt seed\_write\_even\_digit}, w - 1, j  \right> \!, \\
                                 & \! \left< {\tt seed\_write\_even\_digit}, w - 1, j + 1   \right> )
\end{align*}
from the general gadget shown in Figure~\ref{fig:Gadget_write_even_digit_1}.
This step creates gadgets that corresponds to all but the last gadget to self-assemble in Figure~\ref{fig:Special_overview_write_even_digit}.
These are the non-most significant bits of the most significant even digit.
The total number of gadgets created by this step is $O(m)$.
Create
\begin{align*}
    {\tt Write\_even\_digit\_0}( & \! \left< {\tt seed\_write\_even\_digit},       w - 1, m - 1 \right> \!, \\
                                 & \! \left< {\tt special\_single\_tile\_1}, 0, 0                     \right> )
\end{align*}
from the general gadget shown in Figure~\ref{fig:Gadget_write_even_digit_0}, if the most significant bit of $d_{w - 1}$ is 0, otherwise create
\begin{align*}
    {\tt Write\_even\_digit\_1}( & \! \left< {\tt seed\_write\_even\_digit},       w - 1, m - 1 \right> \!, \\
                                 & \! \left< {\tt special\_single\_tile\_1}, 0, 0                     \right> )
\end{align*}
from the general gadget shown in Figure~\ref{fig:Gadget_write_even_digit_1}.
This step creates the gadget that corresponds to the last gadget to self-assemble in Figure~\ref{fig:Special_overview_write_even_digit}.
These are the most significant bits of the even digits.
A single gadget is created by this step.
%

%

For each $i = 0,\ldots,\left(1 + 4 + 3m + 6 + 3(m + 1) + 3\right) - 2$, create
\begin{align*}
    {\tt Single\_tile}( & \! \left< {\tt special\_single\_tile\_1}, i     \right> \!, \\
                        & \! \left< {\tt special\_single\_tile\_1}, i + 1 \right>  )
\end{align*}
from the general gadget shown in Figure~\ref{fig:Gadget_single_tile}.
This step creates gadgets that correspond to all but the last gadget to self-assemble in Figure~\ref{fig:Special_initial_overview_single_tile_1}.
The total number of gadgets created by this step is $O(m)$.
%

Create
\begin{align*}
    {\tt Single\_tile}( & \! \left< {\tt special\_single\_tile\_1}, \left(1 + 4 + 3m + 6 + 3(m + 1) + 3\right) - 1 \right> \!, \\
                        & \! \left< {\tt special\_stopper}                                                         \right>  )
\end{align*}
from the general gadget shown in Figure~\ref{fig:Gadget_single_tile}.
This step creates a gadget that corresponds to the last gadget to self-assemble in Figure~\ref{fig:Special_initial_overview_single_tile_1}.
A single gadget is created by this step.
%

Create
\begin{align*}
    {\tt Special\_stopper}( & \! \left< {\tt special\_stopper}                      \right> \!, \\
                            & \! \left< {\tt special\_single\_tile\_opposite\_3}, 0 \right>  )
\end{align*}
from the general gadget shown in Figure~\ref{fig:Gadget_special_stopper}.
This step creates a gadget that corresponds to the last gadget to self-assemble in Figure~\ref{fig:Special_initial_overview_special_stopper}.
A single gadget is created by this step.
%

For each $i = 0,\ldots,\left(4 + 3m + 6 + 3(m + 1) + 2\right) - 2$, create
\begin{align*}
    {\tt Single\_tile\_opposite}( & \! \left< {\tt special\_single\_tile\_opposite\_3}, i     \right> \!, \\
                                    & \! \left< {\tt special\_single\_tile\_opposite\_3},  i + 1 \right>  )
\end{align*}
from the general gadget shown in Figure~\ref{fig:Gadget_single_tile_opposite}.
This step creates all but the last gadget to self-assemble in Figure~\ref{fig:Special_initial_overview_single_tile_opposite_0}.
The total number of gadgets created by this step is $O(m)$.
%

Create
\begin{align*}
    {\tt Single\_tile\_opposite}( & \! \left< {\tt special\_single\_tile\_opposite\_3},              \left(4 + 3m + 6 + 3(m + 1) + 2\right) - 1 \right> \!, \\
                                    & \! \left< {\tt special\_at\_MSB\_of\_most\_significant\_digit}                                              \right>  )
\end{align*}
from the general gadget shown in Figure~\ref{fig:Gadget_single_tile_opposite}.
This step creates a gadget that corresponds to the last gadget to self-assemble in Figure~\ref{fig:Special_initial_overview_single_tile_opposite_0}.
A single gadget is created by this step.
%

Create
\begin{align*}
    {\tt Special\_at\_MSB\_of\_most\_significant\_digit}( & \! \left< {\tt special\_at\_MSB\_of\_most\_significant\_digit}   \right> \!, \\
                                                          & \! \left< {\tt special\_single\_tile\_opposite\_1},            0 \right>  )
\end{align*}
from the general gadget shown in Figure~\ref{fig:Gadget_special_at_MSB_of_most_significant_digit}.
This step creates a gadget that corresponds to the last gadget to self-assemble in Figure~\ref{fig:Special_initial_overview_special_at_MSB_of_most_significant_digit}.
A single gadget is created by this step.
%

For each $i = 0,\ldots,\left(6 + 3m\right) - 2$, create
\begin{align*}
    {\tt Single\_tile\_opposite}( & \! \left< {\tt special\_single\_tile\_opposite\_1},   i     \right> \!, \\
                                    & \! \left< {\tt special\_single\_tile\_opposite\_1}, i + 1 \right>  )
\end{align*}
from the general gadget shown in Figure~\ref{fig:Gadget_single_tile_opposite}.
This step creates gadgets that correspond to all but the last gadget to self-assemble in Figure~\ref{fig:Special_initial_overview_single_tile_opposite_1}.
The total number of gadgets created by this step is $O(m)$.
%

Create
\begin{align*}
    {\tt Single\_tile\_opposite}( & \! \left< {\tt special\_single\_tile\_opposite\_1}, \left(6 + 3m\right) - 1 \right> \!, \\
                                  & \! \left< {\tt special\_at\_MSB\_of\_odd\_digit}                            \right>  )
\end{align*}
from the general gadget shown in Figure~\ref{fig:Gadget_single_tile_opposite}.
This step creates a gadget that corresponds to the last gadget to self-assemble in Figure~\ref{fig:Special_initial_overview_single_tile_opposite_1}.
A single gadget is created by this step.
%

Create
\begin{align*}
    {\tt At\_MSB\_of\_odd\_digit}( & \! \left< {\tt special\_at\_MSB\_of\_odd\_digit}       \right> \!, \\
                                   & \! \left< {\tt special\_single\_tile\_opposite\_2}, 0 \right>  )
\end{align*}
from the general gadget shown in Figure~\ref{fig:Gadget_at_MSB_of_odd_digit}.
This step creates a gadget that corresponds to the gadget shown in Figure~\ref{fig:Special_initial_overview_at_MSB_of_odd_digit}.
A single gadget is created by this step.
%

For each $i = 0,\ldots,\left(2 + 3m + 6 + 3(m + 1) + 1\right) - 2$, create
\begin{align*}
    {\tt Single\_tile\_opposite}( & \! \left< {\tt special\_single\_tile\_opposite\_2},   i     \right> \!, \\
                                  & \! \left< {\tt special\_single\_tile\_opposite\_2}, i + 1 \right>  )
\end{align*}
from the general gadget shown in Figure~\ref{fig:Gadget_single_tile_opposite}.
This step creates gadgets that correspond to all but the last gadget to self-assemble in Figure~\ref{fig:Special_initial_overview_single_tile_opposite_2}.
The total number of gadgets created by this step is $O(m)$.
Create
\begin{align*}
    {\tt Single\_tile\_opposite}( & \! \left< {\tt special\_single\_tile\_opposite\_2}, \left(2 + 3m + 6 + 3(m + 1) + 1\right) - 1 \right> \!, \\
                                  & \! \left< {\tt reset\_turn\_corner}                                                           \right>  )
\end{align*}
from the general gadget shown in Figure~\ref{fig:Gadget_single_tile_opposite}.
This step creates the gadget that corresponds to the last gadget to self-assemble in Figure~\ref{fig:Special_initial_overview_single_tile_opposite_2}.
A single gadget is created by this step.

Create
\begin{align*}
    {\tt Reset\_turn\_corner}( & \! \left< {\tt reset\_turn\_corner}    \right> \!, \\
                                & \! \left< {\tt reset\_single\_tile}, 1 \right> )
\end{align*}
from the general gadget shown in Figure~\ref{fig:Gadget_return_turn_corner}. The second argument in the encoding of the output glue is 1, which allows {\tt Reset\_single\_tile} gadgets that were previously created to self-assemble.
This step creates the gadget that corresponds to the gadget shown in Figure~\ref{fig:General_overview_return_turn_corner}.
A single gadget is created by this step.
For each $c \in \{0, 1\}$:
\begin{itemize}
    \item Create
    \begin{align*}
        {\tt Z1\_to\_z0}( & \! \left< {\tt special\_z1\_to\_z0},                c    \right> \!, \\
                          & \! \left< {\tt special\_single\_tile\_opposite\_0}, c, 0 \right> )
    \end{align*}
    from the general gadget shown in Figure~\ref{fig:Gadget_z_1_to_z_0}.
    This step creates a gadget that corresponds to the gadget shown in Figure~\ref{fig:General_overview_z_1_to_z_0}.
    The total number of gadgets created by this step is $O \left( 1 \right)$.

    \item For each $i = 0,\ldots,(3m + 2) - 2$, create
    \begin{align*}
        {\tt Single\_tile\_opposite}( & \! \left< {\tt special\_single\_tile\_opposite\_0}, c, i     \right> \!, \\
                                      & \! \left< {\tt special\_single\_tile\_opposite\_0}, c, i + 1 \right> )
    \end{align*}
    from the general gadget shown in Figure~\ref{fig:Gadget_single_tile_opposite}.
    This step creates gadgets that correspond to all but the last gadget to self-assemble in Figure~\ref{fig:Special_overview_single_tile_opposite}.
    The total number of gadgets created by this step is $O(m)$.

    \item Create
    \begin{align*}
        {\tt Single\_tile\_opposite}( & \! \left< {\tt special\_single\_tile\_opposite\_0}, c, (3m + 2) - 1 \right> \!, \\
                                      & \! \left< {\tt special\_start\_digit\_region},      c               \right> )
    \end{align*}
    from the general gadget shown in Figure~\ref{fig:Gadget_single_tile_opposite}.
    This step creates a gadget that corresponds to the last gadget to self-assemble in Figure~\ref{fig:Special_overview_single_tile_opposite}.
    The total number of gadgets created by this step is $O \left( 1 \right)$.
    %

    \item If $M = 2$, create
    \begin{align*}
            {\tt Start\_digit\_region}( & \! \left< {\tt special\_start\_digit\_region}, c    \right> \!, \\
                                        & \! \left< {\tt special\_read\_MSB}, 1,         c \right> \!, \\
                                        & \! \left< {\tt special\_read\_MSB}, 0,         c \right> )
    \end{align*}
    from the general gadget shown in Figure~\ref{fig:Gadget_start_digit_region}, otherwise create
    \begin{align*}
        {\tt Start\_digit\_region}( & \! \left< {\tt special\_start\_digit\_region}, c \right> \!, \\
                                    & \! \left< {\tt special\_read\_non\_MSB}, 1,    c \right> \!, \\
                                    & \! \left< {\tt special\_read\_non\_MSB}, 0,    c \right> )
    \end{align*}
    from the general gadget shown in Figure~\ref{fig:Gadget_start_digit_region}.
    This step creates a gadget that corresponds to the last gadget to self-assemble in Figure~\ref{fig:Special_overview_start_read_most_significant_even_digit}.
    The total number of gadgets created by this step is $O \left( 1 \right)$.
    %

    \item For each $x \in \{0, 1\}^{i}$, for $1 \leq i < m - 1$, create
    \begin{align*}
        {\tt Read\_non\_MSB\_0}( & \! \left< {\tt special\_read\_non\_MSB},  x, c \right> \!, \\
                                 & \! \left< {\tt special\_read\_non\_MSB}, x1, c \right> \!, \\ 
                                 & \! \left< {\tt special\_read\_non\_MSB}, x0, c \right> )      
    \end{align*}
    from the general gadget shown in Figure~\ref{fig:Gadget_read_non_MSB_0} if $x$ starts with $0$, otherwise create
    \begin{align*}
        {\tt Read\_non\_MSB\_1}( & \! \left< {\tt special\_read\_non\_MSB},  x, c \right> \!, \\
                                 & \! \left< {\tt special\_read\_non\_MSB}, x1, c \right> \!, \\ 
                                 & \! \left< {\tt special\_read\_non\_MSB}, x0, c \right> )      
    \end{align*}
    from the general gadget shown in Figure~\ref{fig:Gadget_read_non_MSB_1}.
    This step creates gadgets that correspond to all but the last gadget to self-assemble in Figure~\ref{fig:Special_overview_read_non_MSB}.
    The total number of gadgets created by this step is $O(M)$.
    %

    \item For each $x \in \{0, 1\}^{m - 1}$, create
    \begin{align*}
        {\tt Read\_non\_MSB\_0}( & \! \left< {\tt special\_read\_non\_MSB}, x,  c \right> \!, \\
                                 & \! \left< {\tt special\_read\_MSB},     x1,  c \right> \!, \\ 
                                 & \! \left< {\tt special\_read\_MSB},     x0,  c \right> )      
    \end{align*}
    from the general gadget shown in Figure~\ref{fig:Gadget_read_non_MSB_0} if $x$ starts with $0$, otherwise create
    \begin{align*}
        {\tt Read\_non\_MSB\_1}( & \! \left< {\tt special\_read\_non\_MSB}, x, c \right> \!, \\
                                 & \! \left< {\tt special\_read\_MSB},     x1, c \right> \!, \\ 
                                 & \! \left< {\tt special\_read\_MSB},     x0, c \right> )      
    \end{align*}
    from the general gadget shown in Figure~\ref{fig:Gadget_read_non_MSB_1}.
    This step creates gadgets that corresponds to the last gadget to self-assemble in Figure~\ref{fig:Special_overview_read_non_MSB}.
    The total number of gadgets created by this step is $O(M)$.
    %

    \item For each $x \in \{0, 1\}^{m}$, create
    \begin{align*}
        {\tt Read\_MSB\_0}( & \! \left< {\tt special\_read\_MSB},                     x, c \right> \!, \\
                            & \! \left< {\tt special\_repeating\_after\_even\_digit}, x, c \right> )
    \end{align*}
    from the general gadget shown in Figure~\ref{fig:Gadget_read_MSB_0} if $x$ starts with $0$, otherwise create
    \begin{align*}
        {\tt Read\_MSB\_1}( & \! \left< {\tt special\_read\_MSB},                     x, c \right> \!, \\
                            & \! \left< {\tt special\_repeating\_after\_even\_digit}, x, c \right> )
    \end{align*}
    from the general gadget shown in Figure~\ref{fig:Gadget_read_MSB_1}.
    This step creates gadgets that corresponds to the gadget shown in Figure~\ref{fig:Special_overview_read_MSB}.
    The total number of gadgets created by this step is $O(M)$.
    %

    \item For each $x \in \{0, 1\}^{m}$, create
    \begin{align*}
        {\tt Repeating\_after\_even\_digit}( & \! \left< {\tt special\_repeating\_after\_even\_digit}, x, c \right> \!, \\
                                             & \! \left< {\tt at\_special\_stopper},                   x, c \right> \!, \\ 
                                             & \! \left< {\tt special\_repeating\_after\_even\_digit}, x, c \right> )      
    \end{align*}
    from the general gadget shown in Figure~\ref{fig:Gadget_repeating_after_even_digit}.
    This step creates gadgets that correspond to gadget shown in Figure~\ref{fig:Special_overview_repeating_after_even_digit}.
    The total number of gadgets created by this step is $O(M)$.
\end{itemize}

%
%
For each $x \in \{0, 1\}^{m}$:
\begin{itemize}
    \item Create
    \begin{align*}
        {\tt At\_special\_stopper}( & \! \left< {\tt at\_special\_stopper},        x, 0 \right> \!, \\
                                    & \! \left< {\tt special\_write\_even\_digit}, x, 0 \right> )
    \end{align*}
    from the general gadget shown in Figure~\ref{fig:Gadget_at_special_stopper}. Note that the last argument in the encodings for the input and output glues corresponds to the value of $c$ from the previous {\tt Repeating\_after\_even\_digit} gadget.
    This step creates a gadget that corresponds to the gadget shown in Figure~\ref{fig:Special_overview_at_special_stopper}.
    The total number of gadgets created by this step is $O(M)$.
    \item When $(x + 1) \mod M = 0$, create
    \begin{align*}
        {\tt At\_special\_stopper}( & \! \left< {\tt at\_special\_stopper},            x, 1 \right> \!, \\
                                    & \! \left< {\tt special\_write\_even\_digit}, 0^{m}, 1 \right> )
    \end{align*}
    from the general gadget shown in Figure~\ref{fig:Gadget_at_special_stopper}.
    Otherwise, create
    \begin{align*}
        {\tt At\_special\_stopper}( & \! \left< {\tt at\_special\_stopper},        x, 1 \right> \!, \\
                                    & \! \left< {\tt special\_write\_even\_digit}, z, 0 \right> )
    \end{align*}
    from the general gadget shown in Figure~\ref{fig:Gadget_at_special_stopper}, where $z \in \{0, 1\}^m$ is the zero-padded binary representation of the value $x + 1$.
    This step creates a gadget that corresponds to the gadget shown in Figure~\ref{fig:Special_overview_at_special_stopper}.
    The total number of gadgets created by this step is $O(M)$.
\end{itemize}

%
%
For each $c \in \{0, 1\}$ and each $x \in \{0, 1\}^{i}$, for $1 \leq i < m$, create
\begin{align*}
    {\tt Write\_even\_digit\_0}( & \! \left< {\tt special\_write\_even\_digit}, 0x, c \right> \!, \\
                                 & \! \left< {\tt special\_write\_even\_digit},  x, c \right> )
\end{align*}
from the general gadget shown in Figure~\ref{fig:Gadget_write_even_digit_0} and create
\begin{align*}
    {\tt Write\_even\_digit\_1}( & \! \left< {\tt special\_write\_even\_digit}, 1x, c \right> \!, \\
                                 & \! \left< {\tt special\_write\_even\_digit},  x, c \right> )
\end{align*}
from the general gadget shown in Figure~\ref{fig:Gadget_write_even_digit_1}.
This step creates gadgets that correspond to all but the last gadget to self-assemble in Figure~\ref{fig:Special_overview_write_even_digit}.
The total number of gadgets created by this step is $O(M)$.
The following four steps create the gadgets that write the most significant bit of an even digit contained in the special case digit region.
In each of the following steps, the third argument of the input glue for each gadget is the value of the incoming arithmetic carry.
%

%
%
Create
\begin{align*}
    {\tt Write\_even\_digit\_0}( & \! \left< {\tt special\_write\_even\_digit}, 0, 0 \right> \!, \\
                                 & \! \left< {\tt special\_single\_tile\_1},    0    \right> )
\end{align*}
from the general gadget shown in Figure~\ref{fig:Gadget_write_even_digit_0}.
This step creates a gadget that corresponds to the last gadget to self-assemble in Figure~\ref{fig:Special_overview_write_even_digit}, when the most significant bit is 0 and the value of an incoming arithmetic carry is 0.
A single gadget was created in this step.
%

%
%
Create \begin{align*}
    {\tt Write\_even\_digit\_1}( & \! \left< {\tt special\_write\_even\_digit}, 1, 0 \right> \!, \\
                                 & \! \left< {\tt special\_single\_tile\_1},    0    \right> )
\end{align*}
from the general gadget shown in Figure~\ref{fig:Gadget_write_even_digit_1}.
This step creates a gadget that corresponds to the last gadget to self-assemble in Figure~\ref{fig:Special_overview_write_even_digit}, when the most significant bit is 1 and the value of an incoming arithmetic carry is 0.
A single gadget was created in this step.

%
%
Create
\begin{align*}
    {\tt Write\_even\_digit\_0}( & \! \left< {\tt special\_write\_even\_digit}, 0, 1 \right> \!, \\
                                 & \! \left< {\tt purple\_monkey\_dishwasher}        \right> )
\end{align*}
from the general gadget shown in Figure~\ref{fig:Gadget_write_even_digit_0}.
This step creates a gadget that corresponds to the last gadget to self-assemble in Figure~\ref{fig:Special_overview_write_even_digit}, when the most significant bit is 0 and the value of an incoming arithmetic carry is 1.
A single gadget was created in this step.
%

%
%
Create \begin{align*}
    {\tt Write\_even\_digit\_1}( & \! \left< {\tt special\_write\_even\_digit}, 1, 1 \right> \!, \\
                                 & \! \left< {\tt purple\_monkey\_dishwasher}        \right> )
\end{align*}
from the general gadget shown in Figure~\ref{fig:Gadget_write_even_digit_1}.
This step creates a gadget that corresponds to the last gadget to self-assemble in Figure~\ref{fig:Special_overview_write_even_digit}, when the most significant bit is 0 and the value of an incoming arithmetic carry is 1.
A single gadget was created in this step.

Note that the output glues of the gadgets created in the previous two steps have labels that do not match the label of any other glue.
%

%

%
Each gadget has a fixed size, so the total number of tile types in the tile set output by our construction is $O(M + km)$.

\end{document}